\documentclass[letterpaper,twocolumn,american,pr,aps,superscriptaddress,eqsecnum,nofootinbib,rmp]{revtex4}
\usepackage{amsfonts}
\usepackage{amsmath}
\usepackage{amssymb}
\usepackage{graphicx}
\usepackage[caption=false]{subfig}
\usepackage{bbm}
\usepackage{mathrsfs}
\usepackage{verbatim}
\usepackage{comment}
\usepackage{hyperref}
\usepackage{subfig}
\providecommand{\U}[1]{\protect\rule{.1in}{.1in}}
\newtheorem{theorem}{Theorem}

\newtheorem{corollary}[theorem]{Corollary}

\newtheorem{definition}[theorem]{Definition}

\newtheorem{lemma}[theorem]{Lemma}

\newtheorem{proposition}[theorem]{Proposition}
\newtheorem{remark}[theorem]{Remark}

\newenvironment{proof}[1][Proof]{\vspace{2ex}\noindent\textbf{#1.} }{\ \rule{0.5em}{0.5em}}

\newcommand*{\eins}{\ensuremath{\mathbbm 1}}

\newcommand*{\bbR}{\mathbb{R}}

\newcommand*{\bbZ}{\mathbb{Z}}

\newcommand*{\tr}{\textnormal{tr}}

\newcommand*{\rank}{\mathsf{rank}}

\newcommand*{\supp}{\mathrm{supp}}

\newcommand*{\fr}[2]{\frac{#1}{#2}}

\newcommand*{\assign}{\ensuremath{\kern.5ex\raisebox{.1ex}{\mbox{\textnormal:}}\kern -.3em =}}

\newcommand*{\sgn}{\textnormal{Sgn}}
\DeclareMathOperator*{\conv}{\longmapsto}
\newcommand*{\cconv}{\conv\limits^{\textnormal{noisy}}}

\newcommand*{\epsiloncconv}{\conv\limits^{\textnormal{$\varepsilon$-noisy}}}
\newcommand*{\deltacconv}{\conv\limits^{\textnormal{$\delta$-noisy}}}
\newcommand*{\epsplusdeltacconv}{\conv\limits^{\textnormal{$(\varepsilon + \delta)$-noisy}}}

\newcommand*{\rR}{\mathscr R}

\begin{document}
\preprint{ }
\title{The resource theory of informational nonequilibrium in thermodynamics}
\author{Gilad Gour}
\email{gour@ucalgary.ca}
\affiliation{Institute for Quantum Science and Technology, University of Calgary, 2500 University Drive NW, Calgary, Alberta, Canada T2N 1N4}
\affiliation{Department of Mathematics and Statistics, University of Calgary, 2500 University Drive NW, Calgary, Alberta, Canada T2N 1N4}
\author{Markus P. M\"uller}
\email{markus@mpmueller.net}
\affiliation{Perimeter Institute for Theoretical Physics, 31 Caroline Street North, Waterloo, Ontario Canada N2L 2Y5}
\affiliation{Institut f\"ur Theoretische Physik, Universit\"at Heidelberg, Philosophenweg 19, D-69120 Heidelberg, Germany}
\author{Varun Narasimhachar }
\email{vnarasim@ucalgary.ca}
\affiliation{Institute for Quantum Science and Technology, University of Calgary, 2500 University Drive NW, Calgary, Alberta, Canada T2N 1N4}
\affiliation{Department of Physics and Astronomy, University of Calgary, 2500 University Drive NW, Calgary, Alberta, Canada T2N 1N4}
\author{Robert W. Spekkens}
\email{rspekkens@perimeterinstitute.ca}
\affiliation{Perimeter Institute for Theoretical Physics, 31 Caroline Street North, Waterloo, Ontario Canada N2L 2Y5}
\author{Nicole Yunger Halpern}
\email{NicoleYH@caltech.edu}
\affiliation{Perimeter Institute for Theoretical Physics, 31 Caroline Street North, Waterloo, Ontario Canada N2L 2Y5}
\affiliation{Institute for Quantum Information and Matter, California Institute of Technology, Pasadena, CA 91125, USA.}
\date{Sept.\ 24, 2013; revised April 13, 2015}
\begin{abstract}
We review recent work on the foundations of thermodynamics in the light of quantum information theory.  We adopt a resource-theoretic perspective, wherein thermodynamics is formulated as a theory of what agents can achieve under a particular restriction, namely, that the only state preparations and transformations that they can implement for free are those that are thermal at some fixed temperature.  States that are out of thermal equilibrium are the resources.
We consider the special case of this theory wherein all systems have trivial Hamiltonians
(that is, all of their energy levels are degenerate).  In this case, the only free operations are those that add noise to the system (or implement a reversible evolution) and the only nonequilibrium states are states of \emph{informational nonequilibrium}, that is, states that deviate from the maximally mixed state.  The degree of this deviation we call the state's \emph{nonuniformity}; it is the resource of interest here, the fuel that is consumed, for instance, in an erasure operation.
We consider the different types of state conversion: exact and approximate, single-shot and asymptotic, catalytic and noncatalytic. In each case, we present the necessary and sufficient conditions for the conversion to be possible for any pair of states, emphasizing a geometrical representation of the conditions in terms of Lorenz curves.
We also review the problem of quantifying the nonuniformity of a state, in particular through the use of generalized entropies, and 
that of quantifying the gap between the nonuniformity one must expend to achieve a single-shot state preparation or state conversion and the nonuniformity one can extract in the reverse operation.
Quantum state conversion problems in this resource theory can be shown to be always reducible to their classical counterparts, so that there are no inherently quantum-mechanical features arising in such problems.  This body of work also demonstrates that the standard formulation of the second law of thermodynamics is inadequate as a criterion for deciding whether or not a given state transition is possible.

\end{abstract}
\maketitle
\tableofcontents

\section{Introduction}

A resource theory is defined by identifying a set of experimental operations that are considered free, that is, a set which can be used without limit, while the rest are considered expensive and are thus treated as resources~\cite{Coecke2014}.  A {\em quantum} resource theory is defined by identifying a set of free operations within the set of all quantum operations.
For instance, the restriction to local quantum operations and classical communication defines the resource theory of entanglement~\cite{horodecki2009quantum}, the restriction to symmetric quantum operations defines the resource theory of asymmetry~\cite{Asymy1,Asymy2,Asymy3}, and the restriction to stabilizer operations defines a theory of quantum computational resources~\cite{veitch2013resource}.

The resource theory that is the subject of this article is defined by the following set of free quantum operations:
\begin{itemize}
\item one can implement any unitary on a system,
\item one can prepare any system in the uniform quantum state (i.e., the completely mixed state),
\item one can take a partial trace over any subsystem of a system.
\end{itemize}
In the most general such operation, an ancilla prepared in the uniform state is adjoined to the system of interest, they are coupled unitarily, and finally one takes a partial trace over some subsystem.  These have been called \emph{noisy operations} in \cite{HHOShort,HHOLong}, where this resource theory was first studied\footnote{They have also been called \emph{exactly factorizable maps} by Shor \cite{Shortalk}; see also~\cite{HaagerupMusat}}.
If one comes to possess a system that is \emph{not} in the uniform state, then because such a state cannot be prepared for free, it is a resource. We call such states \emph{nonuniform}, and we call the aspects of a state that are relevant to the resource theory the state's \emph{nonuniformity properties}.   

Nonuniform states are also resources insofar as they can be useful for simulating operations that are outside the noisy class: if a certain process taking a given initial state to a desired final state is not achievable via noisy operations, the process can always be rendered possible by adding an appropriate nonuniform state to the given initial state and suffering the consumption  (partial or complete) of this added resource. An example of a process that would normally be disallowed under noisy operations, but becomes possible by supplying (and using up) an added nonuniformity resource, is an erasure operation. Erasure is ``informational work'', and therefore one can understand the nonuniformity of a state as the potential for doing informational work.

The resource theory of nonuniformity is interesting primarily because of the light that it sheds on
another, more intricate and more relevant, resource theory: the one where the free operations are those that are thermal relative to some temperature $T$, and the resources are states that are not thermal at this temperature, or \emph{athermal states}
\cite{janzing2000thermodynamic,FundLimitsNature,1ShotAtherm2,AthermalityTheory,SSP1,SSP2}.\footnote{The resource-theoretic approach to thermodynamics is cognate with the earlier work of Lieb and Yngvason \cite{Lieb19991}, where classical thermodynamics is formulated in terms of interconvertibility relations between equilibrium states.}

The thermal operations relative to a background temperature $T$ are defined by the following capabilities: 
One can prepare any subsystem in its thermal state at temperature $T$, one can implement any unitary that commutes with the total Hamiltonian\footnote{It is presumed that the couplings between subsystems can be  
entirely controlled by the experimenter, so that in the absence of any such intervention, the overall Hamiltonian is just a sum of the free Hamiltonians for each subsystem.}, and one can trace over the subsystems.  
It follows that the resource theory of nonuniformity is simply a special case of the resource theory of athermality (for any temperature) where all systems of interest have trivial Hamiltonians  (i.e., all their energy levels are degenerate).  

To achieve a proper conceptual understanding of thermodynamics, it is critical to disentangle those aspects of the theory that are due to considerations of energetics and those that are due to considerations of information theory.
Indeed, there is now a large literature on the role of information theory in thermodynamics, centered around such topics as Maxwell's demon \cite{MaxDemonRev}, Szilard's engine \cite{Szilard}, the thermodynamic reversibility of computation and Landauer's principle \cite{BennettLogical,BennettComp,Landauer}, and the use of maximum entropy principles in statistical mechanics \cite{JaynesI,JaynesII}.  (See also \cite{LeffRex}.) By focussing on the resource theory of athermality for the special case of systems that are energy-degenerate, we can begin to understand what aspects of thermodynamics are purely informational.   In essence, we are here studying the particular type of thermal nonequilibrium corresponding to purely \emph{informational nonequilibrium}.  The term `nonuniformity' can be understood as a shorthand for `informational nonequilibrium'.

In addition, many results in athermality theory can be inferred from results in nonuniformity theory.
As such, developing the latter can help to answer questions in the former.  Another motivation for focussing on the case of trivial Hamiltonians is that in building up one's understanding of a field, it is always useful to start with the simplest case.  

In this article, we focus on determining necessary and sufficient conditions for it to be possible to convert one state to another under noisy operations---exactly and approximately,
single-shot and asymptotically, with and without a catalyst---and on finding measures of nonuniformity.  

We demonstrate that 
the answers to these sorts of questions in the quantum case can be inferred from their answers in the classical resource theory of nonuniformity, which is defined in precise analogy with the quantum theory, as follows.  A classical system is represented by a physical state space, or equivalently by a random variable, the valuations of which correspond to the physical states (i.e. the points in that space). The states of such a system that are of interest in the resource theory (and which are the analogues of the quantum states) are \emph{statistical states}, that is, probability distributions over the physical state space. 
The free operations are:
\begin{itemize}
\item one can implement any permutation on a system's physical state space,
\item one can prepare any system in the uniform statistical state (i.e. the uniform probability distribution),
\item one can marginalize over any subsystem of a system.
\end{itemize}

The fact that state conversion problems in the quantum resource theory of nonuniformity reduce to their classical counterparts shows that there is nothing inherently quantum in any problem of state conversion in this theory.

This article seeks to present what is currently known about the resource theory of nonuniformity in a systematic and pedagogical fashion.  In particular, we characterize the equivalence classes of states under noisy operations by \emph{Lorenz curves}, thereby providing a geometric perspective on the problem of defining measures of nonuniformity and state conversion problems, a perspective that makes the proofs of certain results more straightforward and intuitive. 

We draw on various sources.  

The resource-theoretic approach to thermodynamics has seen a great deal of activity in the past few years, primarily by researchers in the field of quantum information \cite{janzing2000thermodynamic,AthermalityTheory,FundLimitsNature,1ShotAtherm2,QuantLandauer,WorkValOfInfo,LawsOfThermo,SanduRefrig,SSP1,SSP2,ThermoMeaningNegEntropy,AbergWork,ABN04}.
Many of the results we present are rederivations of results from these works, or special cases thereof.  In this sense, our article provides a review of some of this literature.

Insofar as such work (and this article) does not confine its attention to macroscopic thermodynamic variables, but rather focusses on the quantum state of the microscopic degrees of freedom, it is perhaps best described as a resource-theoretic approach to \emph{statistical mechanics}.   Because it also does not confine its attention to the family of states that are thermal for some temperature, it concerns \emph{nonequilibrium} statistical mechanics.\footnote{It is important to distinguish between two notions of nonequilibrium here.  In the resource theory of athermality, there is always a background temperature $T$ defining which states are free, and any state that is not \emph{thermal at temperature $T$} is considered a resource. In this context, if a state is thermal for a temperature $T'$ differing from $T$, it is natural to describe it as being \emph{out of equilbrium}.  Conventionally, however, we use the term \emph{equilibrium statistical mechanics} to refer to situations wherein every system is in a thermal state, but where different systems may be at different temperatures. In this context, \emph{nonequilibrium} means not thermal relative to \emph{any} temperature.}    
Finally, because it seeks to understand not just what occurs \emph{on average} but what happens in a single shot, in particular, in the worst case, it is an application to statistical mechanics of the ideas of single-shot information theory.\footnote{And insofar as the \emph{thermodynamic limit} typically involves such averaging procedures, it is concerned with thermodynamic questions outside of this limit.}

Because of the reduction of quantum questions about state conversion to their classical counterparts, there are a number of works on \emph{classical} statistical mechanics that can be brought to bear on our problem.  In particular, our analysis is informed by the work of Ernst Ruch and his collaborators, who adopted a kind of resource-theoretic perspective on state conversion problems 
already in the 1970s~\cite{ruch1975,ruch1978information,RuchDiagrams,RuchSchrannerSeligman,RuchMead}. 

We also make heavy 
use of the mathematical literature on majorization theory \cite{hardy1952inequalities, Bhatia, horn2012matrix,joe1990majorization}, in particular, the canonical text on the subject by Marshall, Olkin and Arnold \cite{MarshallOlkin}, which we shall refer to throughout the article as simply MOA.

Finally, in various places (in particular in our discussion of measures of nonuniformity) 
we make use of work concerning the zoo of generalized entropies~\cite{gorban2010entropy,vEH2010,RennerThesis,SmoothRenyiHAndApplns,datta2009min}.

We now review the structure of the article.  

Section~\ref{Sec:Prelims} covers preliminary material. 
In particular, we clarify the definition of noisy quantum operations, nonuniformity monotones and conversion witnesses, and 
we explain why state conversion problems in the quantum resource theory of nonuniformity can be reduced to their classical counterparts.  

In Section~\ref{Sec:quasiorder}, we characterize the quasi-order of states under noisy operations.  We begin by determining the equivalence classes of states under noisy operations.   We show that these can be associated with a mathematical object known as the Lorenz curve of the state.  We then present the necessary and sufficient conditions for one state to be deterministically converted to another by noisy operations in terms of the Lorenz curves of the states.  

In Section~\ref{sec:NUmonotones}, we discuss various techniques for constructing nonuniformity monotones.  We show that various generalized entropies and relative entropies lead to useful monotones, some of which have operational interpretations.  We also show how certain monotones can be inferred from the geometry of the Lorenz curve. 

 We consider various kinds of exact state conversion in Section~\ref{Sec:Exactstateconversion}.  We begin by defining a standard form of the resource of nonuniformity in terms of a one-parameter family of states which we term the \emph{sharp states}.   
In terms of these sharp states, we can determine the amount of nonuniformity one requires to form a given state, as well as the amount of nonuniformity that one can distill from that state.  These are shown to be related to nonuniformity monotones based on the R\'enyi entropies of order 0 and order $\infty$. 
We also consider the amount of nonuniformity that must be paid to achieve a given state conversion or which can be recovered in addition to achieving the state conversion.  We apply these results to determine some simple sufficient conditions for state conversion.  
Next, we consider exact state conversion in the presence of a catalyst, that is, a nonuniform state which can be used in the conversion process but which must be returned intact at the end of the protocol.  In particular, we note the presence of nontrivial catalysis in our resource theory, but also the uselessness of sharp states as catalysts.  We review the necessary and sufficient conditions for a state conversion to be possible by catalysis. 
Finally, we discuss the inadequacy of the standard formulation of the second law of thermodynamics (nondecrease of entropy) as a criterion for deciding whether a state conversion is possible, with or without a catalyst, and we discuss the criteria that actually do the job.

In Section~\ref{Sec:Approximatestateconversion}, we shift our attention to approximate state conversion.  We begin by formalizing the notion of approximate state conversion in terms of a state being mapped to one that is $\varepsilon$-close to another relative to some contractive metric on the state space, and discuss subtleties of the quantum to classical reduction.  We then review the notion of smoothed entropies.  We reconsider the problem of determining the nonuniformity of formation and the distillable nonuniformity when the state conversion is allowed to be approximate, and we show that the answers are provided by the smoothed versions of the entropies that characterize exact conversion.  We then use these results to provide a simple proof of the rate of asymptotic conversion between states, that is, the rate at which one state can be converted approximately to another in the limit of arbitrarily many copies.
The asymptotic results are in turn applied to determine the nonuniformity cost and yield of approximate state conversion when one has arbitrarily many copies to convert.  Finally, we consider various notions of approximate state conversion in the presence of a catalyst.

In Section~\ref{Sec:Conclusions}, we discuss a few open problems and highlight some of the overarching conclusions of the article.

This article provides a review and synthesis of many known results, but along the way we also present a number of novel results, which we now summarize.

We demonstrate that all exact state conversion problems in the quantum resource theory of nonuniformity reduce to the corresponding problems in the classical resource theory, and that all {\em approximate} state conversion problems can also be reduced to their classical counterparts as long as one makes a judicious choice of the metric over states by which one judges the degree of approximation. We generalize many known results on state conversion and on majorization to the case where the input system and the output system are of different dimensions.  In particular, we describe a general scheme for building nonuniformity monotones from Schur-convex functions in this case.
We introduce the notion of a {\em state conversion witness}: a function of a pair of states in terms of which one can specify a necessary condition for the possibility of conversion from one state to the other, or a necessary condition for the impossibility of such a conversion, or both. We demonstrate how such witnesses are the appropriate tool for summarizing some known results, and are as such more versatile than monotones.   We introduce a class of states, which we call the {\em sharp states}, that serve as a gold standard form of nonuniformity and that naturally generalize the notion of `$n$ pure bits' to the case that $n$ is not an integer.
We use sharp states to quantify the nonuniformity of formation (or of distilation) of an arbitrary state and the nonuniformity cost (or yield) of an arbitrary state conversion.   In the case of approximate conversion, we prove exact expressions for the nonuniformity of formation and of distillation for a class of judiciously-chosen metrics over the states.
For exact conversion using a catalyst, we note the uselessness of sharp states as catalysts and its implication that having access to an ideal measurement cannot catalyze any state conversion.
Finally, we present a new and simplified proof of the rate of asymptotic conversion between nonuniform states.

\section{Preliminaries}\label{Sec:Prelims}
 
\subsection{Free operations, monotones, and witnesses} \label{Freeopsmonotoneswitnesses}

We begin by formalizing the definition of a noisy quantum operation.
A quantum operation is a completely positive trace-preserving linear map $\mathcal{E}:\mathcal{L}(\mathcal{H}_{\textnormal{in}}) \to \mathcal{L} (\mathcal{H}_{\textnormal{out}})$,
where $\mathcal{L}(\mathcal{H})$ is the set of linear bounded operators acting on $\mathcal{H}$. In this paper, all Hilbert spaces are assumed to be finite-dimensional.

\begin{definition}[Noisy quantum operations]
\label{DefNoisyQOperations}
A noisy quantum operation $\mathcal{E}$ is one that admits of the following decomposition.\footnote{For a fixed pair of Hilbert spaces $\mathcal{H}_{\textnormal{in}},\mathcal{H}_{\textnormal{out}}$, the set of maps from $\mathcal{H}_{\textnormal{in}}$
to $\mathcal{H}_{\textnormal{out}}$ of the form~(\ref{eqNoisyQuantum}) is not in general topologically closed~\cite{Shortalk}.
Thus, strictly speaking, we define noisy quantum operations as those linear maps that can be \emph{arbitrarily well approximated}
in operator norm by maps of the form~(\ref{eqNoisyQuantum}). This reflects the intuition that there is
no physical difference between exact and arbitrarily accurate implementation of any map. 
Consequently, the set of noisy quantum operations becomes topologically
closed. This allows us, in the following, to prove noisiness of quantum operations by showing how they can be approximated by maps of the
form~(\ref{eqNoisyQuantum}). Our noisy quantum operations are what~\cite{Shortalk} calls \emph{strongly factorizable maps}.}
There must exist a finite-dimensional ancilla space $\mathcal{H}_a$ and a unitary $U$ on $\mathcal{H}_{\textnormal{in}} \otimes \mathcal{H}_a$ such that, for all input states $\rho_{\textnormal{in}}$,
\begin{equation}
\mathcal{E}(\rho_{\textnormal{in}}) = \textnormal{Tr}_{a'}\left(U \left(\rho_{\textnormal{in}}\otimes \tfrac{1}{d_{a}} I_{a}\right) U^{\dag}\right),
\label{eqNoisyQuantum}
\end{equation}
where $\mathcal{H}_{a'}$ is the space complementary to $\mathcal{H}_{\textnormal{out}}$ in the total Hilbert space, that is,
$\mathcal{H}_{\textnormal{out}} \otimes \mathcal{H}_{a'} = \mathcal{H}_{\textnormal{in}} \otimes \mathcal{H}_a$, and where $d_a=\dim(\mathcal{H}_a)$.
\end{definition}

Although at first glance, ``degree of purity'' might appear to be a good name for what we are calling \textquotedblleft
nonuniformity\textquotedblright, at second glance one recognizes that it is not because standard usage of the term ``purity'' takes a pure state of a 2-level system and a pure state of a 3-level to have \emph{equal} degrees of purity, while the latter is a stronger resource than the former in the resource theory we are considering.  
The fact that these two pure states are not equivalent under noisy operations is counterintuitive for many quantum information theorists because it contrasts with the situation in entanglement theory, 
the resource theory with which they are most familiar.  The difference arises because the free operations in entanglement theory (local operations and classical communication) allow one to prepare pure product states for free, thereby allowing a given state to be embedded into an arbitrarily large space, whereas noisy operations do not allow one to prepare pure states for free.  More generally, two states with identical structure on their support have different amounts of resourcefulness depending on the dimension of the space in which the state is embedded.
\footnote{This inequivalence can also be made slightly more intuitive by noting that it holds also in the resource theory of athermality and
 accounts for the fact that a Szilard engine at temperature $T$ extracts different amounts of work from the two states: if it operates on a 2-level system that is prepared in a pure initial state, it can extract an expected work of $k_B T \log2$, whereas if it operates on a 3-level system that is prepared in a pure initial state, it can extract $k_BT \log3$.}
The resourcefulness of a state under noisy operations is therefore not determined by its proximity to a pure state, but rather its distance from the uniform state, and therefore the term ``nonuniformity'' is more appropriate than the term ``purity'' as a description of the resource.

An important question in any resource theory is how to quantify the resource. A \emph{resource monotone} is a function $f$ from states to the real numbers such that if $\rho \mapsto \sigma$ by the free operations, then $f(\rho)\ge f(\sigma)$. For the set of operations given by noisy operations, we refer to the resource monotones as nonuniformity monotones:

\begin{definition}[Quantum nonuniformity monotone]\label{def:QNUMon}
A function $M$, mapping density operators to real numbers, is a nonuniformity monotone if for any two states $\rho$ and $\sigma$ (possibly of different dimensions), $\rho\mapsto\sigma$ by noisy operations implies that $M(\rho)\ge M(\sigma)$.
\end{definition}

Any measure of nonuniformity must at least be a nonuniformity monotone.
When quantifying the relative resourcefulness of different states, there is a strong compulsion to think that 
the researcher's task is to find the ``one true measure to rule them all''.  This tendency stems from an implicit assumption that any property worth quantifying must necessarily be totally ordered.  In fact, it is more common to find that resource states only form a quasi-order, and no individual scalar measure can capture a quasi-order. So the search for the one true measure is in vain.  The ``total-order fallacy'' must be resisted.

Nonetheless, resource monotones still have an important role to play.  A \emph{set} of such monotones {\em can} capture the quasi-order, in which case they are called a \emph{complete set} of monotones.
Also, if we define an operational task that requires the resource, it is then a well-defined question which resource monotone accurately quantifies the degree of success achievable in the task (according to some figure of merit) for a given state.  The states become totally ordered relative to the task.

In this article, we describe some methods for generating a large number of nontrivial nonuniformity monotones, and we describe operational interpretations for some of these.  We also identify some complete sets of nonuniformity monotones.  

Besides quantifying nonuniformity, we will be primarily interested in problems of state conversion.  A useful tool for characterizing such problems is the concept of a \emph{state conversion witness}.  
\begin{definition}[Quantum state conversion witness]\label{QSCWitness}
Let $W$ be a real-valued function on pairs of quantum states, $\rho$ and $\sigma$ (possibly of different dimensions).
$W$ is said to be a \emph{go witness} if $W(\rho,\sigma) \ge 0$ implies that $\rho \mapsto \sigma$ under the free quantum operations. 
$W$ is said to be a \emph{no-go witness} if $W(\rho, \sigma) < 0$ implies that $\rho \not\mapsto \sigma$ under the free quantum operations.
Finally, $W$ is said to be a \emph{complete witness} if it is both a go witness and a no-go witness.  
\end{definition}
No-go witnesses were previously introduced under the name of \emph{relative monotones} \cite{SandersGour}.

Any resource monotone $M$ defines a no-go witness by $W(\rho, \sigma)\assign M(\rho)-M(\sigma)$ because $W(\rho, \sigma)<0$ implies $M(\sigma) > M(\rho)$, which implies that $\rho \not \mapsto \sigma$.  It follows that $M$ also defines a \emph{resource witness} -- that is, a function $w$ such that $w(\sigma) < 0$ implies that $\sigma$ is a resource -- by $w(\sigma) \assign W(\rho_{\textnormal{free}} , \sigma) = M(\rho_{\textnormal free})-M(\sigma)$, where $\rho_{\textnormal{free}}$ is any free state.  A complete set $\{ M_k \}$ of monotones  defines a complete witness by $W(\rho , \sigma)\assign \min_{k} \left( M_k(\rho)-M_k(\sigma) \right)$.

\subsection{Reducing quantum state conversion problems to their classical counterparts} \label{section:QuantToClassl}

To define the \emph{classical} resource theory of nonuniformity, we must characterize the set of noisy \emph{classical} operations.  Just as we have confined our attention to finite-dimensional Hilbert spaces in the quantum case, we confine our attention to discrete variables in the classical case (that is, finite information-carrying capacity in both cases).  
Let $\Omega_{\textnormal{in}}$ be the discrete physical state space of the input system and $\Omega_{\textnormal{out}}$ be the discrete physical state space of the output system. 
Let $S(\Omega)$ be the simplex of probability distributions over $\Omega$, and let $V(\Omega)$ be the smallest linear vector space in which $S(\Omega)$ can be embedded.  Clearly, every probability distribution on a sample space $\Omega$ can be represented by a vector in $V(\Omega)$. We call each such probability distribution a {\em state}. When we say that states are {\em normalized}, we refer to the fact that their components sum to total probability one.

The classical analogue of a quantum operation that is completely positive and trace-preserving is a classical operation which preserves positivity and normalization, hence taking probability distributions to probability distributions.  For finite dimension, this is represented by a stochastic matrix, that is, a matrix $D$ whose entries are real, are nonnegative, and satisfy $\sum_j D_{jk} =1$.

We can now provide a definition of the free operations.
\begin{definition}[Noisy classical operations]
A noisy classical operation is a positivity-preserving and normalization-preserving map $D:V_{\textnormal{in}}\to V_{\textnormal{out}}$ that admits of the following decomposition.\footnote{Analogously to the definition of noisy quantum operations, the set of classical operations of the form~(\ref{eq:NCO}) between fixed vector spaces $V_{\textnormal in}$ and $V_{\textnormal out}$ is not topologically closed. This follows from the simple observation that input vectors $x_{\textnormal in}$ with all rational entries are
mapped to output vectors $D x_{\textnormal in}$ with all rational entries. Thus, as in the quantum case, we define noisy classical operations from $V_{\textnormal{in}}$ to $V_{\textnormal{out}}$
as those linear maps that can be \emph{arbitrarily well approximated} in operator norm by maps of the form~(\ref{eq:NCO}).
}  There must exist an ancilla system with a discrete physical state space $\Omega_a$ and a permutation on $\Omega_{\textnormal{in}} \times \Omega_{a}$ with an induced representation $\pi$ on $V(\Omega_{\textnormal{in}}) \otimes V(\Omega_{a})$  such that, for all input states $x_{\textnormal{in}} $,
\begin{equation}\label{eq:NCO}
D x_{\textnormal{in}} = \sum_{\Omega_{a'}} \pi \left(x_{\textnormal{in}} \otimes m_a\right),
\end{equation}
where $m_a$ is the normalized uniform distribution on $\Omega_{a}$ and $\Omega_{a'}$ is the physical state space complementary to $\Omega_{\textnormal{out}}$, that is, $\Omega_{\textnormal{out}} \times \Omega_{a'} = \Omega_{\textnormal{in}} \times \Omega_{a}$.
\end{definition}

We note an important difference between the structure of the set of noisy quantum operations and the set of noisy classical operations.

Recall that a \emph{unital} operation $\mathcal{E}:\mathcal{L}(\mathcal{H}_{\textnormal{in}}) \to \mathcal{L}(\mathcal{H}_{\textnormal{out}})$ is one for which
$\mathcal{E}(\eins_{\textnormal{in}}/d_{\textnormal{in}})=\eins_{\textnormal{out}}/d_{\textnormal{out}}$, where $\eins_{\textnormal{in}}$ and $\eins_{\textnormal{out}}$ are the identity operators on $\mathcal{H}_{\textnormal{in}}$ and $\mathcal{H}_{\textnormal{out}}$, respectively~\cite{MendlWolf}.  
\begin{lemma}\label{lemma:ClassifyingNO}
Noisy quantum operations form a strict subset of the unital operations, and, in the case of equal dimension of input and output space, a strict superset of the mixtures of unitaries. 
\end{lemma}

The proof is described in Appendix~\ref{app:proofs}.

To compare with the classical case, we must specify the classical analogues of each of the classes of operations appearing in Lemma~\ref{lemma:ClassifyingNO}.

The classical analogue of a unital quantum operation is a stochastic matrix $D$ that takes the uniform distribution on the input system to the uniform distribution on the output system, hence $\sum_k D_{jk} =d_{\textnormal{in}}/d_{\textnormal{out}}$ where $d_{\textnormal{in}}$ ($d_{\textnormal{out}}$) is the dimension of the input (output) vector space.  We will call such a stochastic matrix \emph{uniform-preserving}.  
Note that, if the input and output spaces are of equal dimension, a uniform-preserving stochastic matrix $D$ satisfies $\sum_j D_{jk}=\sum_k D_{jk}=1$, in which case it is said to be doubly-stochastic.  Finally, the classical analogue of a mixture of unitaries is a mixture of permutations of the physical state space.

The structure of the set of noisy classical operations is much more straightforward than the quantum one.

\begin{lemma} \label{lemma:NCO}
The set of noisy classical operations coincides with the set of uniform-preserving stochastic matrices and, in the case of equal dimension of input and output spaces (where the uniform-preserving stochastic matrices are the doubly-stochastic matrices), it coincides with the set of mixtures of permutations.
\end{lemma}

Here also, we relegate the proof to Appendix~\ref{app:proofs}.  Note that the fact that every doubly-stochastic matrix can be written as a convex combination of permutations is known as Birkhoff's theorem~\cite{Birkhoff}.

Because all unitaries are free operations in the quantum resource theory of nonuniformity, the only feature of a quantum state that is relevant for its nonuniformity properties is the vector of its
eigenvalues (where an eigenvalue appears multiple times in the vector if it is
degenerate and we explicitly include any zero eigenvalues).
For any state $\rho$, we denote the vector of eigenvalues, listed in non-increasing order,
by $\lambda\left(  \rho\right)$. Therefore, the condition under which we can
deterministically convert $\rho$ to $\sigma$ by noisy operations must be
expressible in terms of $\lambda\left(  \rho\right)  $ and $\lambda\left(\sigma\right)$. 

We can now state the reduction of the quantum state conversion problem to the corresponding classical problem, leaving the proof to Appendix~\ref{app:proofs}.

\begin{lemma}\label{lemma:quantumNOtoclassicalNO}
There exists a noisy quantum operation that achieves the quantum state conversion $\rho \mapsto \sigma$ if and only if there is a noisy classical operation that achieves the classical state conversion $\lambda(\rho) \mapsto \lambda(\sigma)$.
\end{lemma}

Lemma~\ref{lemma:quantumNOtoclassicalNO} implies that all questions about exact state interconversion in the quantum resource theory of nonuniformity are answered by the classical theory, so studying the latter is sufficient.\footnote{Note, however, that one \emph{does} find a separation between the quantum and the classical theories if one assumes an additional restriction on the free operations, namely, that the systems to which one has access are correlated with others to which one has no access. 
The reason is that a mixture of entangled states cannot be transformed by a local unitary to a mixture of product states.  
} Later on (in Lemma~\ref{LemApproximateQuantumClassical}), we will prove an analogous statement for approximate conversion.

In particular, Lemma~\ref{lemma:quantumNOtoclassicalNO} implies a reduction of quantum nonuniformity monotones and quantum state conversion witnesses to their classical counterparts, which are defined analogously to the quantum notions, Definitions~\ref{def:QNUMon} and \ref{QSCWitness}, but with quantum states replaced by probability distributions, and noisy quantum operations replaced by noisy classical operations. 

\begin{corollary}\label{corollary:classicalmonotonetoquantummonotone}
Consider $M$, a real-valued function over probability distributions, and $M'$, a real-valued function over quantum states, such that $M'(\rho) \assign M(\lambda(\rho))$, where $\lambda(\rho)$ is the vector of eigenvalues of $\rho$.  Then $M'$ is a quantum nonuniformity monotone  if and only if $M$ is a classical nonuniformity monotone.
\end{corollary}

\begin{corollary}\label{corollary:classicalwitnesstoquantumwitness}
Consider $W$, a real-valued function over pairs of probability distributions, and $W'$, a real-valued function over pairs of quantum states, such that $W'(\rho,\sigma) \assign W(\lambda(\rho), \lambda(\sigma))$.  Then $W'$ is a witness for state conversion under noisy quantum operations if and only if $W$ is a witness for state conversion under classsical noisy operations.
\end{corollary}

It may seem surprising that the quantum problem of state conversion reduces to its classical counterpart even though the strict inclusions of the mixtures of unitaries within the noisy operations, and the noisy operations within the unital operations (Lemma~\ref{lemma:ClassifyingNO}), have no classical counterparts (Lemma~\ref{lemma:NCO}).  The solution to this puzzle is that, for the purposes of state conversion, the three classes of quantum operations have the same power.  

\begin{lemma}\label{lemma:NOconviffUntlconv}
Let $\rho\in\mathcal{L}(\mathcal{H}_{\textnormal{in}})$ and $\sigma\in\mathcal{L}(\mathcal{H}_{\textnormal{out}})$. Then, the following propositions are equivalent:
\begin{itemize}
\item[(i)]  $\rho \mapsto \sigma$ by a noisy operation
\item[(ii)] $\rho \mapsto \sigma$ by a unital operation.
 \end{itemize}
If $\rho$ and $\sigma$ are of equal dimension, then (i) and (ii) are also equivalent to
\begin{itemize}
\item[(iii)] $\rho \mapsto \sigma$ by a mixture of unitaries.
 \end{itemize}
\end{lemma}

See Appendix~\ref{app:proofs} for the proof.
For the case of $\rho$ and $\sigma$ of equal dimension, Uhlmann's Theorem~\cite{uhlmann1970} implies that $\rho \mapsto \sigma$ by a mixture of unitaries if and only if  the spectrum of $\rho$ majorizes that of $\sigma$; see \cite{nielsen2001majorization} for a discussion. \footnote{The implication from an operation being unital to majorization of the final state's spectrum by the initial state's has also been noted in \cite{chefles2002quantum}.}

There are questions in the resource theory that do not concern state conversion.  For instance, one may ask about the possibility of simulating an operation that is outside the free set, given access to some resource state.  For such questions, the quantum solution does not necessarily reduce to the classical one.  We do not consider such problems in this article.  From this point onwards, therefore, we can restrict our attention to the classical resource theory.

\section{Quasi-order of states under noisy operations}\label{Sec:quasiorder}

\subsection{Equivalence classes under noisy operations}

The first step in understanding state conversion within the classical resource theory of nonuniformity is to determine
when two states are equivalent relative to noisy operations, by which it is meant that they can be \emph{reversibly interconverted}, one to the other, deterministically, by noisy operations.  In this case, we say that the states have precisely the same \emph{nonuniformity properties}.   

\begin{definition}[Exact state conversion]
We write $x\cconv y$ if there exists a noisy classical operation $D$ such that $y=Dx$.
\end{definition}

\begin{definition}[Noisy-equivalence of states]
Two states, $x$ and $y$, are said to be noisy-equivalent if they are reversibly interconvertible, that is, if $x\cconv y$ and $y\cconv x$.
\end{definition}

Since noisy classical operations  can have an input of dimension $d_{\rm in}$ and an output of dimension $d_{\rm out}$, state conversion is a map from a vector on  $\bbR_+^{d_{\rm in}}$ to a vector on $\bbR_+^{d_{\rm out}}$. 
As such, it is useful to introduce the set
$$\rR\assign\bigcup\limits_{d=1}^\infty\bbR_+^d.$$
Any finite-dimensional probability distribution is a vector in $\bbR_+^d$ for some $d$, and therefore a member of the set $\rR$. For a distribution $x\in\rR$, we will denote by $d_x$ the integer for which $x\in\bbR_+^{d_x}$.

The simplest case to consider is when $x$ and $y$ are of equal dimension.
In this case, if $x$ and $y$ are reversibly interconvertible by noisy operations, there is a permutation that takes one to the other.   The proof is as follows. Suppose the doubly stochastic matrix taking $x$ to $y$ is denoted $D$ and the one taking $y$ to $x$ is denoted $D'$.  We can write $D$ (respectively $D'$) as a mixture of permutations $\pi_i$ (respectively $\pi'_j$):
\[
   D=\sum_i \lambda_i\pi_i, \quad D'=\sum_j \lambda'_j \pi'_j,
\]
where $\sum_i\lambda_i=\sum_j\lambda'_j=1$ and all $\lambda_i,\lambda'_j>0$. Thus
\[
   x=D'Dx=\sum_{i,j} \lambda'_j \lambda_i \pi'_j \pi_i x.
\]
Consider the convex set
consisting of all convex combinations of permutations of $x$. The state $x$ itself is an extremal point of this convex set. Thus, for all $i, j$, we must have
$\pi'_j \pi_i x=x$, and so $\pi_i x=\left(\pi'_j\right)^{-1} x$ is the same state (call it $z$) for all $i$. But
$y=Dx=\sum_i \lambda_i \pi_i x = z$, so $y=\pi_i x$ is a permutation of $x$.

If we let $x^{\downarrow}$ denote the vector having the same components as $x$ but permuted such that they are in descending order, the condition for noisy-equivalence can be expressed as $x^{\downarrow}=y^{\downarrow}$.

Note that if $x$ and $y$ are of equal dimension and noisy-equivalent, and if $x$ is the marginal of a correlated state on a larger system, and similarly for $y$, then because we can get from $x$ to $y$ by a permutation, the conversion can be achieved while preserving all correlations with other systems.

The more interesting case is where $x$ and $y$ have unequal dimension:  $d_x \neq d_y.$ 
First note that one can reversibly interconvert $x$ and $x\otimes m$ for any uniform state $m$ on an arbitrary ancilla.  We get from $x$ to $x\otimes m$ simply by injecting an ancilla in the uniform state, which is allowed under noisy operations, and we get from $x\otimes m$ to $x$ simply by discarding the ancilla, which is also allowed under noisy operations.

It follows that one can reversibly interconvert $x$ and $y$ by noisy operations if and only if one can reversibly interconvert $x \otimes m$ and $y \otimes m'$, where $m$ and $m'$ are uniform states for any arbitrary pair of ancillas.  In particular, this is true if and only if one can reversibly interconvert $x \otimes m$ and $y \otimes m'$ for ancillas having dimensions $d_{m}$ and $d_{m'}$ such that $d_x d_{m} = d_y d_{m'}$, so that $x \otimes m$ and $y \otimes m'$ are of equal dimension.  As shown above, in this case, the interconversion is possible if and only if there is a permutation that takes $x \otimes m$ to $y \otimes m'$.  Given that uniform states are invariant under permutations, this condition is equivalent to 
\[
\label{eq:Noisyequivalence}
x^{\downarrow} \otimes m = y^{\downarrow} \otimes m'.
\]
This condition, therefore, is necessary and sufficient for reversible interconvertibility of $x$ and $y$ under noisy operations.

In the resource theory of nonuniformity, the only properties of a state that are relevant to determining its value as a resource are its nonuniformity properties, that is, the features of the state that determine its noisy-equivalence class. 
It is therefore useful to replace the state $x$ by a mathematical object that represents only the nonuniformity properties of $x$.  

\begin{figure*}
\centering
\includegraphics[width=.95\textwidth, clip=true]{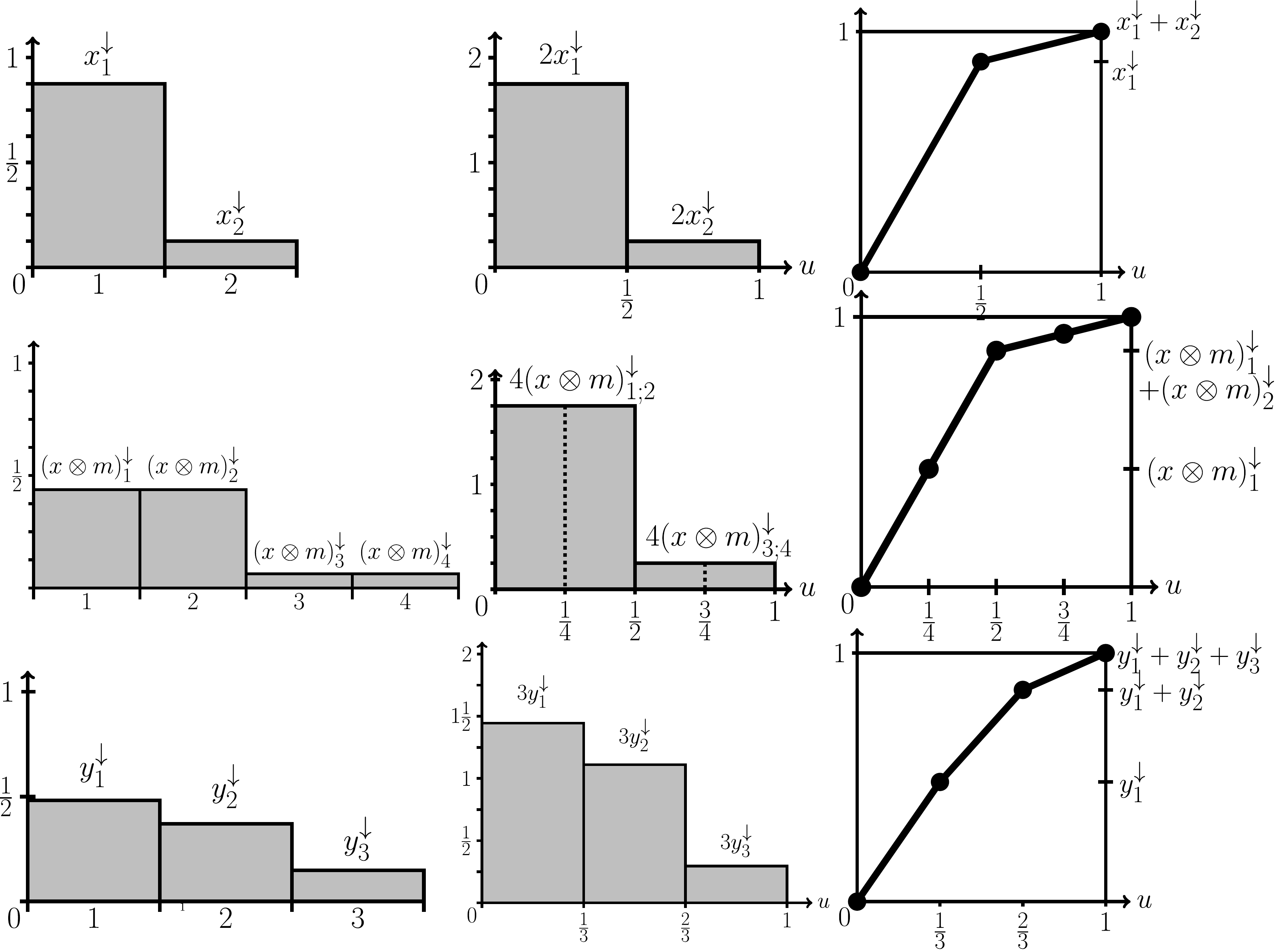}
\caption{An illustration of uniform-rescaled histograms and Lorenz curves.  The first column depicts the histogram of a state's components, in descending order, the second column depicts the corresponding uniform-rescaled histogram, and the third column depicts the Lorenz curve.  The states depicted are:  (first row)  $x\equiv (0.9,0.1)$, a state of a 2-level system; (second row) $x \otimes m \equiv (0.45, 0.45, 0.05,0.05)$, the state of a pair of 2-level systems, where the first is in state $x$ and the second is in the uniform state $m \equiv (0.5,0.5)$; (third row) $y$, a generic state of a three-level system.}
\label{fig:IntroToHists}
\end{figure*}

Given that $x$ is reversibly interconvertible with $x \otimes m$, this mathematical object must be invariant under adding and removing ancillas in the uniform state.  
If we plot histograms of $x^{\downarrow}$ and $(x\otimes m)^{\downarrow}$, we see that the envelope of each is a step function and that the steps' relative heights are equal. An example is given in Fig.~\ref{fig:IntroToHists}. We can make these step functions strictly equal by rescaling them in an appropriate way.  The rescaling is chosen such that the area of each bar of the histogram remains the same, while the range of the function becomes $[0,1]$. Specifically, we define the function $h_x(v)$, with range $v \in [0,1]$, such that the $k$th step extends over the range  $v \in [(k-1)/d_x,k/d_x]$ and has height $d_x x^{\downarrow}_k$.  Equivalently,
\begin{equation}
h_x(v) \equiv d_x x^{\downarrow}_{\lfloor d_x v \rfloor},
\end{equation}
where $\lfloor a \rfloor$ denotes the integer floor of $a$. Again, an example is provided in Fig.~\ref{fig:IntroToHists}.  We call $h_x(v)$ the \emph{uniform-rescaled histogram} of $x$.

Clearly, adding or removing a uniform state $m$ (of arbitrary dimension) leaves the uniform-rescaled histogram invariant, $h_x(v) = h_{x \otimes m}(v), \forall v \in [0,1]$, as illustrated in Fig. \ref{fig:IntroToHists}.   Furthermore, if two states of the same dimension are equal up to a permutation, their uniform-rescaled histograms are equal.  In particular, it follows that $x^{\downarrow} \otimes m = y^{\downarrow} \otimes m'$ if and only if $h_{x \otimes m}(v)=h_{y \otimes m'}(v)\;\;\forall v \in [0,1]$.  But $h_x(v) = h_{x \otimes m}(v)$ so we can conclude that $x^{\downarrow} \otimes m = y^{\downarrow} \otimes m'$ holds if and only if $h_{x}(v)=h_{y}(v)$.  We have shown, therefore, that the uniform-rescaled histogram of a state $x$ is a mathematical object that characterizes the noisy-equivalence class of that state.

Now consider the integral of $h_x(v)$ from $v=0$ to $v=u$ as a function of $u \in [0,1]$, 
\begin{equation}\label{eq:Lorenzcurvedefn}
L_x(u) \assign \int_0^u h_x(v) \textnormal{d}v.
\end{equation}
Clearly, the curve traced by $L_x(u)$ extends between (0,0) and (1,1) regardless of the state $x$.  It contains all the information contained in the rescaled histogram $h_x(v)$.  Indeed, one can recover the latter by taking the derivative of $L_x(u)$.  As such, it is another mathematical object that characterizes the noisy-equivalence class.
Examples of this curve for various states are provided in Fig.~\ref{fig:IntroToHists}. For the uniform state $m$, this curve is simply the diagonal line extending between (0,0) and (1,1). 

There is another way of defining this curve which is worth noting.  First, define $S_k(x)$ for $k=1,...,d_x$ as the sum of the $k$ largest components of $x$,
\begin{equation}
S_{k}\left(  x\right): =\sum_{i=1}^{k}x_{i}^{\downarrow},
\end{equation}
and define $S_{0}(x):=0$ (Note that $S_{d_x}(x)=1$ for $x$ a normalized probability distribution).  $S_k(x)$ is sometimes called the Ky Fan $k$-norm of $x$ \cite{horn2012matrix,Bhatia}.
Then we can characterize $L_x(u)$ as the linear interpolation of the points
\begin{equation}
\left(  \frac{k}{d_x},\frac{S_{k}(x)}{S_{d_x}(x)}\right)\;\; \forall k=0, \ldots,d_x.
\end{equation}

$L_x(u)$ is called the \emph{Lorenz curve} of $x$ \cite{lorenz1905methods}, (MOA, Sec.~1.A).\footnote{Actually, the Lorenz curve of $x$ is conventionally taken to be the linear interpolation of $\left(  k/d_x,T_{k}/T_{d_x}\right) \; \; \forall k=0, \ldots, d_x$, where $T_{k}\left(  x\right): =\sum_{i=1}^{k}x_{i}^{\uparrow}$ is the sum of the $k$ \emph{smallest} components of $x$ \cite{lorenz1905methods}, or equivalently, it is taken to be the integral $L'_x(u) \assign \int_{1-u}^1 h_x(v) \textnormal{d}v$.  But the conventional definition is just the inversion about the line extending from (0,0) to (1,1) of our definition, and so the two curves have precisely the same information content.  We here choose to adopt the opposite of the usual convention because in this way the majorization relation between states, $x \succ y$, coincides with the inequality relating the height of the Lorenz curves of those states, $L_x(u) \ge L_y(u), \forall u\in[0.1]$ (see Eq.~\eqref{eq:Lorenzorder} below).  Also, our convention coincides with the one adopted in~\cite{FundLimitsNature}.}
It will be seen to be one of the primary tools for characterizing the resource theory of nonuniformity.

It is worth emphasizing that we can also infer that noisy equivalence implies equality of Lorenz curves \emph{directly} from the condition that $x^{\downarrow} \otimes m = y^{\downarrow} \otimes m'$.  It suffices to note that, via the definition of Lorenz curves in terms of Ky-Fan $k$-norms, this condition is equivalent to $L_{x \otimes m}(u)=L_{y \otimes m'}(u),\;\; \forall u\in [0,1]$ and then to note that the Lorenz curve is invariant under adding and removing ancillas in a uniform state, $L_{x \otimes m}(u)=L_{x}(u),\;\; \forall u\in [0,1]$.

To summarize, we have shown that
\begin{proposition} [Conditions for noisy equivalence]\label{prop:NoisyEquivalence}
A pair of states $x$ and $y$ are noisy-equivalent if and only if the following equivalent conditions holds
\begin{itemize}
\item[(i)] their uniform-rescaled histograms are equal,
\begin{equation}
h_{x}(v)=h_{y}(v),\;\;\; \forall v\in [0,1],
\end{equation}
\item[(ii)] their Lorenz curves are equal, 
\begin{equation}\label{eq:EqualityLorenzCurves}
L_{x}(u)=L_{y}(u)\;\;\; \forall u\in [0,1].
\end{equation}
\end{itemize}
\end{proposition}

The uniform-rescaled histogram of $x$ and  the Lorenz curve of $x$ both capture all and only the nonuniformity properties of $x$.  It follows that for any notion of state conversion we wish to study (exact or approximate; single-shot, multi-copy or asymptotic; catalytic or noncatalytic), the necessary and sufficient conditions under which one state can be converted to another can always be expressed in terms of either of these objects.  Furthermore, any nonuniformity monotone or state conversion witness
must be expressible entirely in terms of them as well.

The application of these mathematical objects to thermodynamics was first recognized in~\cite{ruch1975,ruch1978information,RuchDiagrams,RuchSchrannerSeligman,RuchMead}.  What we have called the ``uniform-rescaled histogram'' was discussed in \cite{RuchSchrannerSeligman} under the title of the ``density diagram''.  More recently, in \cite{LawsOfThermo}, these old tools have begun to be used again in the context of an information-theoretic approach to thermodynamics. In this article, an operation of ``Gibbs-rescaling'' is introduced which is akin to our use of the uniform-rescaled histogram in place of the distribution itself (although with a different scaling convention for the horizontal axis). 
The analogue of the Lorenz curve for an  athermal state, that is, the generalization of Lorenz curves to a system with a nontrivial Hamiltonian, has recently been studied in~\cite{FundLimitsNature} and applied in \cite{1ShotAtherm2}. 

With this characterization of the noisy equivalence class in hand, we can clarify a point that was made in the introduction, namely, that the nonuniformity properties of a state depend on the dimension of the space in which the state is embedded.  
As already noted, this is because embedding a state in a higher-dimensional space---that is, padding the state with extra zeros---is \emph{not} a noisy operation.
In terms of the uniform-rescaled histogram, padding a state with extra zeros corresponds to squeezing the entire histogram of the state to the left and leaving only zeros on the right side, which obviously results in a different histogram.  In the Lorenz curve picture, it corresponds to squeezing the Lorenz curve to the left and adding a horizontal segment at value 1 on the right end, again resulting in something that is obviously distinct from the original Lorenz curve.

\subsection{Deterministic interconversion of nonuniform states} \label{section:Interconvert}

We now turn to a consideration of the necessary and sufficient conditions on a pair of states $x$ and $y$ such that there exists a deterministic noisy operation taking $x$ to $y$.  We are here asking about one-way state conversion, i.e., there need not be any deterministic noisy operation taking $y$ to $x$.  We begin by presenting the general result.

\subsubsection{The result}

\begin{proposition}\textnormal{(\textbf{Conditions for deterministic conversion})} \label{prop:detconv}
$x \cconv y$ if and only if
\begin{itemize}
\item[(i)] the uniform-rescaled histogram of $x$, $h_x(v)$, can be taken to that of $y$, $h_y(v)$, by moving probability density only towards the right (i.e. from lower to higher values of $v$).
\item[(ii)] the Lorenz curve of $x$ is everywhere greater than or equal to the Lorenz curve of $y$:
\begin{equation}\label{eq:Lorenzorder}
L_{x}(u)\geq L_{y}(u)\;\;\; \forall u\in [0,1].
\end{equation}
\end{itemize}
When these conditions hold, we say that $x$ \emph{noisy-majorizes} $y$.
\end{proposition}

The rest of this section provides the proof of Proposition~\ref{prop:detconv}.  We start with the case in which $x$ and $y$ have equal dimension.

Recall the definition
of the majorization relation (Definition A.1 of MOA).
\begin{definition} [Majorization] Letting $x$ and $y$ be normalized probability vectors with equal dimension $d$, we say that $x$ \emph{majorizes} $y$ and write $x \succ y$ if
\begin{align}
&\sum_{i=1}^{k}x_{i}^{\downarrow} \ge \sum_{i=1}^{k}y_{i}^{\downarrow}\;\;\; \forall k=1,...,d-1. \label{eq:majorization1}
\end{align}
\end{definition}
Because $x$ and $y$ are normalized probability distributions, $\sum_{i=1}^{d}x_{i}^{\downarrow} = \sum_{i=1}^{d}y_{i}^{\downarrow}=1$.

The connection with noisy classical operations is made through the following famous result
\cite{hardy1952inequalities}:
\begin{lemma}[Hardy, Littlewood, Polya] \label{lemma:HardyLittlewoodPolya}
$x \succ y$ if and only if 
there is a doubly stochastic matrix $D$ such that $y=Dx$.
\end{lemma}

Given that the set of noisy classical operations with equal input and output dimensions are represented by the set of doubly-stochastic matrices (Lemma~\ref{lemma:NCO}), we immediately obtain that, for $x$ and $y$ of equal dimensions, $x \cconv y$ if and only if $x \succ y$.

An equivalent means of expressing the condition of majorization is in terms of Lorenz curves.  Noting that the expressions in the inequalities of Eq.~\eqref{eq:majorization1} are just the Ky Fan norms $S_k(x)$ and recalling that the Lorenz curve is the linear interpolation of points $(k/d_x,S_k(x)/S_{d_x}(x))$, we see that in the case of $x$ and $y$ of equal dimension, we have that
$x \succ y$ if and only if the Lorenz curve of $x$ is nowhere less than the Lorenz curve of $y$:
\begin{equation}
L_{x}(u)\geq L_{y}(u)\;\;\; \forall u\in [0,1].
\end{equation}
This proves (ii) of Proposition~\ref{prop:detconv} for the case of states of equal dimension.

When $x$ and $y$ are states of unequal dimension, the condition for deterministic conversion is not simply majorization;
this is why we use the term {\em noisy-majorization} to describe the condition in the general case.  It is determined using the same trick that was deployed in the characterization of the noisy equivalence classes. It suffices to note that $x \cconv y$ if and only if there exist uniform states $m$ and $m'$ such that $x \otimes m \cconv y \otimes m'$ (because adding and removing uniform states are noisy operations), and that if we choose $m$ and $m'$ such that $d_x d_{m} =d_y d_{m'}$, then $x \otimes m$ and $y \otimes m'$ are states of equal dimension.  If we define $d$ as the least common multiple of $d_x$ and $d_y$,  $d \assign\textnormal{LCM}(d_x,d_y)$, then it suffices to choose $d_{m} \assign d/d_x$ and $d_{m'} \assign d/d_y$, in which case $x \otimes m$ and $y \otimes m'
$ have dimension $d$.

Recalling that we have established Condition (ii) of Proposition~\ref{prop:detconv} for states of equal dimension, and the fact that the Lorenz curve of $x \otimes m$ for a uniform state $m$ is equal to the Lorenz curve of $x$,  it follows that $x \cconv y$ if and only if the Lorenz curve of $x$ is nowhere below the Lorenz curve of $y$.  This proves Condition (ii) of Proposition~\ref{prop:detconv} for all states.

Finally, recalling that the Lorenz curve is the cumulative integral of the uniform-rescaled histogram, any motion of density rightward in the uniform-rescaled histogram corresponds to a decrease of the height of the Lorenz curve over some subset of its domain, while motion of density leftward corresponds to an increase of height.  Condition (ii) of Proposition~\ref{prop:detconv}, therefore, implies (i).

This concludes the proof of Proposition~\ref{prop:detconv}.

The application of majorization theory to state conversion in thermodynamics was studied extensively in~\cite{ruch1975,ruch1978information,RuchDiagrams,RuchSchrannerSeligman,RuchMead}.  The quantum information community became familiar with majorization due to its role in the resource theory of entanglement~\cite{nielsen1999conditions,nielsen2001majorization}.  The problem of state conversion in thermodynamics was first considered from a quantum information perspective in ~\cite{janzing2000thermodynamic}, where some necessary conditions on state conversion were derived. The necessary and sufficient conditions for state conversion under thermal operations were first determined in~\cite{FundLimitsNature}, where the relation was called \emph{thermo-majorization}.  The results described in this section are the specialization of the thermo-majorization relation to the case of a trivial Hamiltonian.

\subsubsection{Some consequences}
The order over states induced by deterministic conversion is not a total order but a quasi-order.  We call it the \emph{noisy quasi-order}.  One easily generates pairs of states that are not noisy-ordered relative to one another by simply drawing a pair of valid Lorenz curves where one is not everywhere above the other.\footnote{The order is a \emph{quasi-order} (also known as a pre-order) rather than a partial order, because we can have $x \succ y$ and $y \succ x$ for $x \ne y$. This occurs whenever $x$ is a nontrivial permutation of $y$ or requires addition or removal of a uniform state. While the states form a quasi-order, the noisy-equivalence classes of states form a partial order.}  

\begin{figure}[h!]
\centering
\includegraphics[width=.3\textwidth,clip=true]{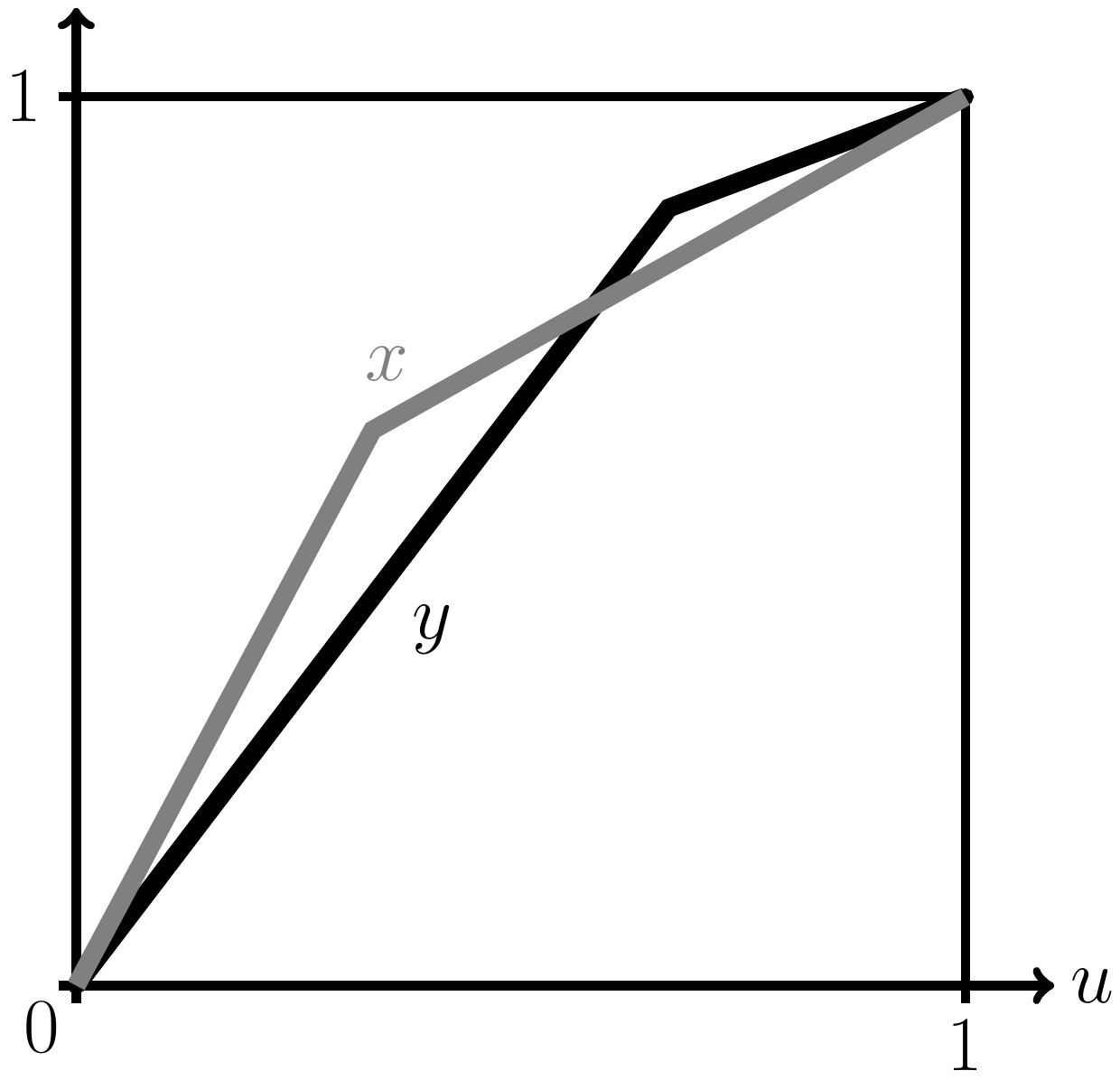}
\captionsetup{singlelinecheck=off,justification=raggedright}
\caption{The Lorenz curves of a pair of states, $x$ and $y$,  such that it is neither the case that $x\cconv y$ nor that $y \cconv x$. 
}
\label{fig:UnorderedLorenzCurves.jpg}
\end{figure}

Note that one can easily recover the condition for noisy-equivalence (Proposition~\ref{prop:NoisyEquivalence}) from the condition for noisy-majorization (Proposition~\ref{prop:detconv}) by recognizing that reversible interconvertibility requires noisy-majorization in both directions.

Another simple corollary of the deterministic-conversion result concerns the relation between a state of a composite and its marginal state.
Suppose $x^{AB}$ is a state of a composite system $\Omega_A \times \Omega_B$ and $x^A$ is the marginal state on $\Omega_A$, that is, $x^A_i =\sum_j x^{AB}_{ij}$, where $i\in \Omega_A, j \in \Omega_B$. 
$x^A$ is noisy-majorized by $x^{AB}$, that is, $L_{x^{AB}}(u) \ge L_{x^{A}}(u)$ for all $u \in [0,1]$.

For noisy-equivalence of $x^{AB}$ and $x^{A}$, we require that $L_{x^{AB}}(u) = L_{x^{A}}(u)$ which implies that $x^{AB\downarrow} = x^{A\downarrow} \otimes m^{B}$.  It follows that marginalization is reversible only if the marginalized system is uncorrelated with the rest and is in a uniform state. 

Proposition~\ref{prop:detconv} also implies that 
the height of the Lorenz curve at a given value of $u$ in the region $[0,1]$ is a nonuniformity monotone, $M_{\textnormal{Lorenz},u}(x)\equiv L_x(u)$, and the set of such heights, $\{ M_{\textnormal{Lorenz},u} :u\in[0,1] \}$, form a complete set of nonuniformity monotones.   Although this is an infinite set, if $x$ and $y$ are both of finite dimension (the only case we consider in this article), one can decide the convertability question by looking at a finite number of monotones.
The following is the pertinent result.

\begin{lemma}
For $x$ and $y$ of finite dimension, $x$ noisy-majorizes $y$ if and only if $L_x(u) \ge L_y(u)$ at the points $u=k/d_y$ for all $k=1,\dots, d_y-1$.  In other words, it suffices to consider the $d_y-1$ monotones
\begin{equation}\label{eq:completesetmonotones}
M_{\textnormal{Lorenz},k/d_y}(x):= L_x(k/d_y),\;\; k=1,\dots, d_y-1.
\end{equation}
\end{lemma}
The proof is simply that $L_y(u)$ is linear between the distinguished values of $u$.  Given that $L_x(u)$ is concave, if $L_x(u)\ge L_y(u)$ at these points, then $L_x(u)\ge L_y(u)$ at all $u\in [0,1]$.

This result can be rephrased in terms of state-conversion witnesses as follows.
 \begin{corollary}
Each of the functions 
\begin{equation}\label{eq:nogowitness}
\Delta_k(x\|y):=L_x(k/d_y)-L_y(k/d_y)
\end{equation}
for $k\in \{ 1,\dots, d_y \}$ is a no-go witness for $x \cconv y$.  That is, if $\Delta_k(x\|y) <0$ then it is not the case that $x \cconv y$.  

The function
\begin{equation}\label{eq:dualwitness}
\Delta(x\|y):=\min_{k \in \{ 1,\dots,d_y\}} \Delta_k(x\|y)
\end{equation}
is a complete witness for the state conversion.  That is, $x \cconv y$ if and only if $\Delta(x\|y) \ge 0$.
\end{corollary}

In Section~\ref{sec:secondlaw}, we will discuss how this result implies the inadequacy of the standard formulation of the second law of thermodynamics.

\subsubsection{Implementation}

In the previous section, we determined the conditions under which it is possible to transform $x$ to $y$ by a noisy operation, but the proof was not constructive.  In this section, we describe a practical implementation of the appropriate noisy operation. We begin with states of equal dimensions.

First of all, we recall the definition of majorization in terms of a sequence of \emph{T-transforms}.   A T-transform is a doubly-stochastic matrix that is nontrivial on a single $2\times 2$ block.  For the block corresponding to levels $i$ and $j$, we denote the T-transform by $T_{ij}$.  The only permutations that act only on levels $i$ and $j$ are the identity, denoted $I$, and the permutation that swaps $i$ and $j$, which we denote by $\Pi_{ij}$. Therefore, by Birkhoff's Theorem \cite{Birkhoff}, the most general form of $T_{ij}$ is 
$$ \label{eq:Ttransform}
T_{ij} = w I +(1-w) \Pi_{ij},
$$
 where $0\le w \le 1$.
It follows that if $x$ and $y$ are $d$-dimensional vectors and $y=T_{ij} x$, then
\begin{align}
y_i =  w x_i + (1-w) x_j, \\
y_j =w x_j + (1-w)  x_i.
\end{align}

\begin{figure}[h!]
\centering
\includegraphics[width=.43\textwidth, clip=true]{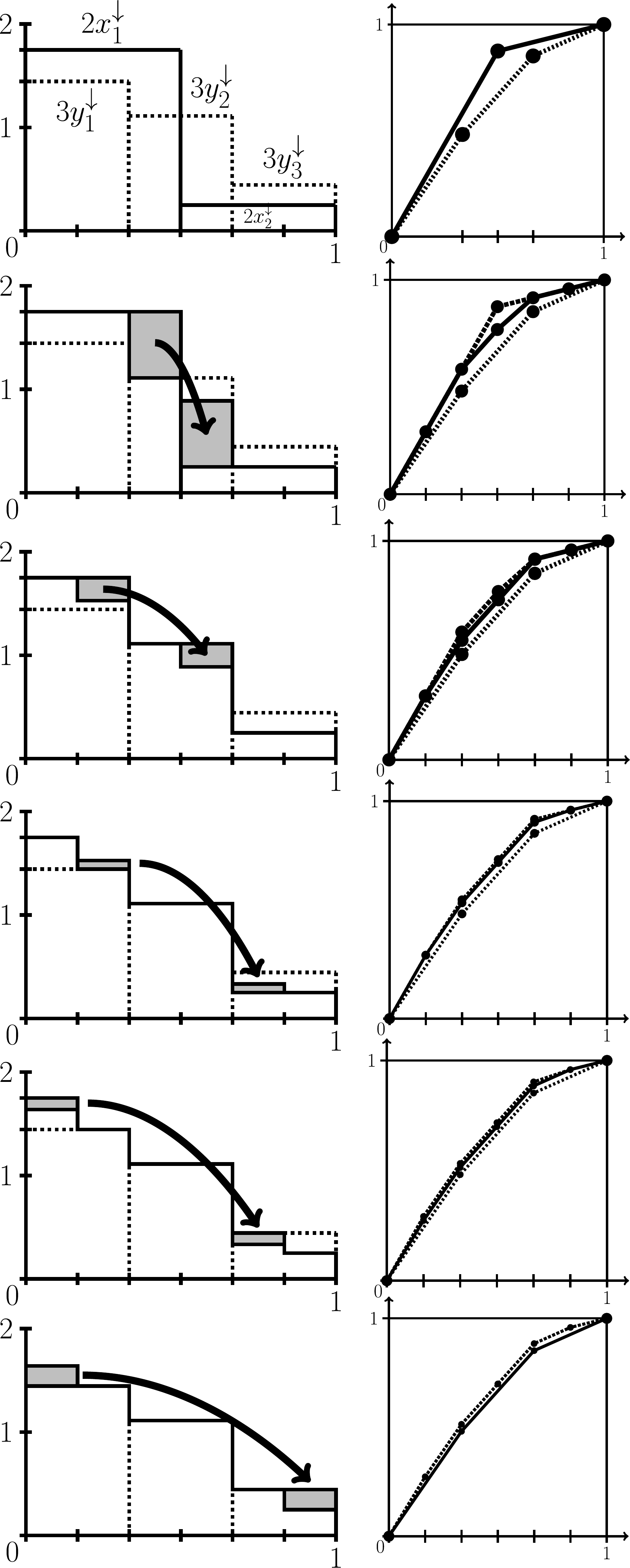}
\captionsetup{singlelinecheck=off,justification=raggedright}
\caption{A depiction of the evolution of the state during a minimal sequence of T-transforms achieving the state conversion $x \cconv y$.  
The left column denotes the uniform-rescaled histograms of the states, while the right depicts their Lorenz curves.}
\label{fig:T-transforms}
\end{figure}

Equivalently, defining $\bar{x}_{ij} \assign \tfrac{1}{2} (x_i + x_j)$ and $q \assign 2\left|w-\tfrac{1}{2}\right|$, we have
\begin{align}
y_i =  q x_i + (1-q)  \bar{x}_{ij}, \\
y_j = q x_j + (1-q)  \bar{x}_{ij}.
\end{align}
We see that if $w=1/2$ ($q=0$), then the weights of the pair of levels $i$ and $j$ become equal, and if $w \ne 0,1$ ($q\ne1$), then these weights become closer to equal (assuming they were unequal to begin with).  In the context of income inequality, a T-transform is called a ``Robin Hood transfer'' \cite{Arnold87}.

\begin{lemma}[Muirhead, Hardy, Littlewood, Polya] \label{lemma:HardyLittlewoodPolyaT}
For states of equal dimensions, $x \succ y$ if and only if there is a finite sequence of T-transforms taking $x$ to $y$.  The number of steps required is at most $d-1$, where $d$ is the dimension.
\end{lemma}

\begin{proof}
We follow the proof in (MOA, Lemma B.1, p. 32).  Let $x^{(n)}$ denote the state after the $n$th step of the sequence of T-transforms, so that $x^{(0)}=x$ and $x^{(n_{\max})}=y$.  Consider the step that takes $x^{(n)}$ to $x^{(n+1)}$.  Let $j_{\textnormal{ex}}$ denote the largest index such that the weight for that index is strictly larger for $x^{(n)\downarrow}$ than it is for $x^{(n+1)\downarrow}$, that is, $x^{(n)\downarrow}_{j_{\textnormal{ex}}} > x^{(n+1)\downarrow}_{j_{\textnormal{ex}}}$  (``ex'' denotes ``excess'').   Let $j_{\textnormal{df}}$ denote the smallest index such that the weight for that index is strictly smaller for $x^{(n)\downarrow}$ than it is for $x^{(n+1)\downarrow}$, that is, $x^{(n)\downarrow}_{j_{\textnormal{df}}} < x^{(n+1)\downarrow}_{j_{\textnormal{df}}}$  (``df'' denotes ``deficient'').  By definition, we must have $j_{\textnormal{ex}} < j_{\textnormal{df}}$. 
This is illustrated in Fig. \ref{fig:T-transforms}.

We consider a protocol where, in the $n^\textnormal{th}$ step in the sequence, one transfers the maximum weight possible from level $j_{\textnormal{ex}}$ to level $j_{\textnormal{df}}$ while ensuring that one still has
$x^{(n)\downarrow}_{j} \ge x^{(n+1)\downarrow}_{j} $  for $j=j_{\textnormal{ex}}$ and for $j=j_{\textnormal{df}}$
(so that, in particular, $x^{(n)}$ still majorizes $x^{(n+1)}$).   This maximum weight, denoted $\delta$, cannot exceed the difference of weights in level $j_{\textnormal{ex}}$ nor the difference of weights in level $j_{\textnormal{df}}$, so
$$
\delta = \min \{ x^{(n)\downarrow}_{j_{\textnormal{ex}}} - x^{(n+1)\downarrow}_{j_{\textnormal{ex}}}, 
x^{(n+1)\downarrow}_{j_{\textnormal{df}}} - x^{(n)\downarrow}_{j_{\textnormal{df}}} \}.
$$
The resulting transformation is
\begin{align}
x^{(n+1)\downarrow}_{j_{\textnormal{ex}}} = x^{(n)\downarrow}_{j_{\textnormal{ex}}} -\delta, \\
x^{(n+1)\downarrow}_{j_{\textnormal{df}}} = x^{(n)\downarrow}_{j_{\textnormal{df}}} +\delta.
\end{align}
Because the transfer tends to make the weights of the two levels closer to equal, it is clearly a T-transform.  Specifically, it is the transform $T_{j_{\textnormal{ex}} j_{\textnormal{df}}} = w I + (1-w) \Pi_{j_{\textnormal{ex}} j_{\textnormal{df}}}$, where $w = 1 - \delta/(x^{(n)\downarrow}_{j_\textnormal{ex}} - x^{(n)\downarrow}_{j_\textnormal{df}})$.  After the $n^\textnormal{th}$ step of the sequence, either the index $j_{\textnormal{ex}}$ is reduced by 1 or the index $j_{\textnormal{df}}$ is increased by 1 (or both), so that at the next step, one is transferring weight between a different pair of levels.  Clearly, if $x$ majorizes $y$, then a sequence of such transformations can take $x$ to $y$ in a finite number of steps. 

Finally, we show that at most $d-1$ steps are required.  Suppose $d(x,y)$ denotes the number of levels wherein $x^{\downarrow}$ and $y^{\downarrow}$ differ in value.  Because the number of differences is reduced by 1 at every step of the sequence, and the last step takes two differences to no differences, it follows that one requires $d(x,y)-1$ steps.  $d(x,y)$, however, is at most $d$.
\end{proof}

This covers the case of states of equal dimension.  Next, we consider how to implement a noisy operation that achieves deterministic conversion of states of unequal dimensions.
It follows from Lemma~\ref{lemma:HardyLittlewoodPolyaT} that we can achieve the transformation by a sequence of T-transforms, with the number of steps in the sequence being at most $d-1$ where $d= \textnormal{LCM}(d_x,d_y)$.  Fig.~\ref{fig:T-transforms} depicts the set of T-transforms that maps $x \otimes m_1$ to $y \otimes m_2$ (where $d_{m_1} \assign d/d_x$ and $d_{m_2} \assign d/d_y)$ by the protocol described in the proof of Lemma~\ref{lemma:HardyLittlewoodPolyaT}.  Note that we must divide the $x$-axis into bins of size $1/d$ to depict how the sequence of T-transforms acts on the uniform-rescaled histograms.  Note also that this figure provides an intuitive proof of Proposition~\ref{prop:detconv}
(i).  Finally, note how the T-transforms act on the Lorenz curves.  Each T-transform acts on the pair of bins for which the difference between the slope differential of the Lorenz curves in the first bin and the slope differential of the Lorenz curves in the second bin is largest.  After a particular T-transform is complete, the slope differential in one of the two bins becomes zero, and consequently, the Lorenz curves have the same slope in that bin.  By this process, the Lorenz curve of the initial state approaches that of the target state.

\section{Nonuniformity monotones}\label{sec:NUmonotones}

In this section, we discuss the properties and method of construction of functions which serve as monotones under classical noisy operations. A reader interested in getting a quick introduction to nonuniformity monotones (without too much technical detail) will find the tables in this section useful: Table~\ref{table:NUmonotones} lists some monotones derived from convex functions of one real variable, outlining the steps of the derivation; Table~\ref{table:LCmonotones} lists some monotones which can be derived from the geometry of the Lorenz curve; and Table~\ref{table:OptlIntnsMonotones} lists the operational significance of some of the monotones.

Given that noisy operations may change the dimension of the space, every nonuniformity monotone is really a family of functions, $\{ G_d : d \in \bbZ_+\}$, where $G_d:\bbR_+^d \to \bbR$ is defined on a $d$-dimensional space.  We will denote by $G:\rR\to\bbR$ the function which reduces to $G_d$ on $\bbR_+^d$. Thus, $\forall x\in\rR$,
\begin{equation}\label{eq:Mf}
G(x)=G_{d_x}(x).
\end{equation}
Where the dimensionality is clear from the context, we will omit the subscript $x$ in $d_x$.

A monotone is said to be \emph{strict} if $M(x)=M(y)$ only on level sets, that is, only when $x$ and $y$ are either noisy-equivalent (i.e., $x\cconv y$ and $y\cconv x$) or not ordered relative to one another (i.e., it is neither the case that $x \cconv y$ nor that $y \cconv x$).

Since appending ancillary systems in the uniform state is a reversible process under noisy operations, every nonuniformity monotone $M$ must satisfy
$M\left(x\otimes m^{(d)}\right)=M(x)\;\;\forall\;d\in\bbZ_+$,
where $m^{(d)}$ is the uniform state of dimension $d$.

Within a space of fixed dimension, the noisy quasi-order reduces to the majorization quasi-order, so any nonuniformity monotone $M$, when restricted to states of a fixed dimension, must be a nonincreasing monotone with respect to the majorizaton quasi-order. The latter
are known as \emph{Schur-convex} functions and have been extensively studied. See, e.g., Chapter 3 of MOA.
\begin{definition}[Schur-convexity]
\label{DefSchurConvexity}
A function $G_d$ mapping $d$-dimensional probability distributions to the reals is Schur-convex iff for all $d$-dimensional distributions $x$ and $y$,
\begin{equation}\label{eq:Schurconvexity}
x\succ y\;\Rightarrow \;G_d(x)\ge G_d(y).
\end{equation}
\label{def:Schur}
\end{definition}

As it turns out, for a function $M: \rR \to \bbR$, the conditions that the restriction of $M$ to every $d$-dimensional vector space be Schur-convex and that $M$ be invariant under the addition of a uniform state are necessary and sufficient for $M$ to be a nonuniformity monotone. We formalize this in the following.
\begin{proposition}
\label{prop:NUMon}
A function $M:\rR\to\bbR$ is a nonuniformity monotone if and only if both of the following conditions hold:
\begin{enumerate}
\item For each $d\in\bbZ_+$, the restriction of $M$ to $\bbR_+^d$ is Schur-convex.
\item For all $d\in\bbZ_+$, for all distributions $x\in\rR$,
\begin{equation}
\label{eq:NUmonotoneAddm}
M\left(x\otimes m^{(d)}\right)=M(x),
\end{equation}
where $m^{(d)}$ is the $d$-dimensional uniform distribution.
\end{enumerate}
\end{proposition}

\begin{proof}
The forward implication is trivial, so our task is to show that the pair of conditions imply that $M$ is a nonuniformity monotone. 
 
We begin by noting that for each $d$, if the restriction of $M$ to $\bbR_+^d$ is Schur-convex, then it follows by definition that if $x$ and $y$ are states of equal dimension, and $x$ majorizes $y$, then $M(x) \ge M(y)$. It remains to show that if $x$ and $y$ are of unequal dimension, and $x$ noisy-majorizes $y$, then $M(x) \ge M(y)$.

By assumption, $M\left(x \otimes m^{(d)}\right) = M(x )$, 
for any $d$.  Then, one can reason as follows. First, we can define states $x \otimes m^{(d_y)}$ and $y \otimes m^{(d_x)}$ that are of equal dimension, namely $d_xd_y$. Given that $x$ noisy-majorizes $y$, and given that adding uniform states does not change the noisy-equivalence class of a state, it follows that $x \otimes m^{({d_y})}$ noisy-majorizes $y \otimes m^{(d_x)}$. But since noisy-majorization between states of equal dimension is just majorization, it follows that $x \otimes m^{({d_y})}$ majorizes $y \otimes m^{(d_x)}$. Then, by the Schur-convexity of the restriction of $M$ to $\bbR^{d_x d_y}$, we conclude that $M\left(x \otimes m^{({d_y})} \right)\ge M\left(y \otimes m^{(d_x)}\right)$. It follows that $M(x)\ge M(y)$.
\end{proof}

Note that if one takes an \emph{arbitrary} family of Schur-convex functions, $\{ G_d : d\in  \bbZ_+\}$  where $G_d:\bbR_+^d \to \bbR$, then the function $G:\rR\to\bbR$ which reduces to $G_d$ on $\bbR_+^d$ need not be a nonuniformity monotone because it need not satisfy the requirement of invariance under adjoining a uniform state, Eq.~\eqref{eq:NUmonotoneAddm}.  So the problem of defining nonuniformity monotones is not solved by merely finding families of Schur-convex functions.
Nonetheless, the powerful characterization theorems for Schur-convex functions can be exploited to construct and characterize
nonuniformity monotones. We will discuss these theorems in the following.

\subsection{Nonuniformity monotones built from convex functions on the reals}
\label{subs:ConMon}

\subsubsection{Schur-convexity}

The following result is well-known; cf.\ (MOA, Proposition C.1).
\begin{lemma}
\label{lemma:SchCon}
For every convex function $g:\bbR_+\rightarrow \bbR$, the function $G_d:\bbR_+^d \to \bbR$ defined by
\[ \label{eq:fd}
G_d(x)\assign\sum\limits_{i=1}^{d}g\left(x_i\right)
\]
is Schur-convex for each $d\in\bbZ_+$. That is, if $x$ and $y$ are of equal dimension $d$, and if $x\succ y$, then $G_d(x) \ge G_d(y)$.
\end{lemma}
\begin{proof}
Recall that a function $g:\bbR \rightarrow \bbR$ is convex if, for any pair of points $a,a' \in \bbR$, 
\[
g(w a + (1-w) a') \le w g(a) + (1-w) g(a')
\]
for all $w \in [0,1]$. 
If $g$ is twice differentiable, an equivalent definition of convexity is that $g''(a)\ge 0$ for all $a \in \bbR$.

Note also that the function $G_d(x)$ is invariant under permutation of the components of $x$. Hence it is sometimes described as a symmetric function. 

The proof is then straightforward.  Given that $x \succ y$, it follows that $y= Dx$ for some doubly-stochastic matrix $D$.  By Birkhoff's Theorem \cite{Birkhoff}, there exists a probability distribution $(w_i)$ and a set of permutations $(\Pi_i)$ such that  $y = \sum_i w_i \Pi_i x$. Because $G_d$ is a sum of convex functions, it is also convex, and therefore, $G_d(y) \le \sum_i w_i G_d(\Pi_i x)$.  But we noted above that $G_d$ is invariant under a reordering of the components of its argument, so that $G_d(\Pi_i x)= G_d(x)$.  It follows that $G_d(y) \le \sum_i w_i G_d(x)=G_d(x)$.
\end{proof}

Now that we have seen how to generate families of Schur-convex functions, the question arises of which of these families can yield a function that is invariant under adjoining uniform states.

\subsubsection{Schur-convexity relative to a distribution $q$}

Towards this end, we introduce the notion of majorization \emph{relative to a distribution $q$}; see (MOA, Ch.\ 14, Sec.\ B, p.\ 585).

Within a space of a given dimension $d$, consider those stochastic matrices that preserve a particular distribution $q$. Call these the $q$-preserving stochastic matrices.
\begin{definition}[Majorization relative to $q$]
A distribution $x$ is said to majorize another distribution $y$ \emph{relative to the distribution} $q$, denoted $x \succ_q y$, if there exists a $q$-preserving stochastic matrix $D$ such that $y =Dx$. 
\end{definition}
\begin{definition}[Schur-convexity relative to $q$]
We say that a function $f: \bbR_+^d \to \bbR$ is \emph{Schur-convex relative to $q$} if $x \succ_q y$ implies $f(x) \ge f(y)$.
\end{definition}

Intuitively, whereas a Schur-convex function quantifies the distance of some distribution $x$ to the uniform distribution $m$, a Schur-convex function relative to $q$ quantifies the distance of $x$ to $q$.

The following is a useful method of constructing Schur-convex functions relative to a distribution $q$, from convex functions of one real variable. 

\begin{lemma}\label{thm:Schurconvexrelativetoq}
 Given a distribution 
 $q \in \bbR_+^{d}$ with all $q_i\neq 0$,
 for every convex function $g:\bbR_+\rightarrow\bbR$, the function $G(\cdot\|q):\bbR_+^{d} \to \bbR$ defined by
 \begin{equation}\label{eq:Schurconcaverelativetom}
G(x\|q)\assign \sum\limits_{i=1}^{d} q_i g\left(\tfrac{x_i}{q_i} \right)
 \end{equation}
is Schur-convex relative to $q$. That is, if $x$ and $y$ are of equal dimension $d$, and if $x\succ_q y$, then $G(x\|q) \ge G(y\|q)$.
\end{lemma}
This is Proposition B.3 in Ch.\ 3 of MOA, proven in \cite{veinott1971}. 
Functions of the form of Eq.~\eqref{eq:Schurconcaverelativetom} have also been proposed as a generalization of the notion of relative entropy in~\cite{csisz1967information} (where it is called the $g$-divergence of $x$ from $q$) and in~\cite{morimoto1963markov}, as discussed in~\cite{gorban2010entropy}.

Note that if we multiplied each $q_i$ by a factor $c>0$  in Eq.~\eqref{eq:Schurconcaverelativetom}, we would still have a Schur-convex function relative to $q$ (because the function $g'(a) \assign c g(a/c)$ is also convex).  Nonetheless, the case of $c=1$ has a special status.  Suppose we define a fiducial distribution $q$ for every system and suppose that the fiducial distribution on a composite system is the product of the fiducial distributions on the components, that is, $q^{AB} =  q^A \otimes q^B$.  In this case, we can prove the following.
\begin{lemma}\label{thm:Schurrelativeadjoining}
 For every convex function $g:\bbR_+\rightarrow\bbR$, the function $M_g:\rR\to\bbR$ defined by $M_g(x) \equiv G(x\| q)$
satisfies $M_g\left(x \otimes q^S\right) = M_g(x)$ for any system $S$.
\end{lemma}

\begin{proof}
It suffices to note that 
\begin{align}
G\left(x^A \otimes q^B \big\|q^{AB}\right) &= G\left(x^A \otimes q^B \big\|q^A \otimes q^B\right) \\
&= \sum_{i,j} q^A_i q^B_j g\left(\tfrac{x^A_i q^B_j }{ q^A_i  q^B_j} \right) \\
&= \sum_{i} q^A_i   g\left(\tfrac{x^A_i }{ q^A_i}\right)  \\
&= G\left(x^A \big\|q^A\right),
\end{align}
where we have used the fact that $\sum_j q^B_j =1$.
\end{proof}

It follows that any family of functions (one for every type of system), each of which is Schur-convex relative to the fiducial distribution on that system, can be used to construct resource monotones in a classical resource theory where those fiducial distributions are the free states.

In particular, in the context of the resource theory of athermality \cite{AthermalityTheory,1ShotAtherm2}, where the free states are thermal states, we can construct athermality monotones from families of functions that are Schur-convex relative to the thermal state (note that the thermal state of a system depends not only on the dimension of that system, but also on its Hamiltonian).  

\subsubsection{Schur-convexity relative to the uniform state}
\label{SubsubSchurConvexity}

By specializing to the case where the fiducial distribution for a system of dimension $d$ is just the uniform distribution of that dimension, Lemma~\ref{thm:Schurrelativeadjoining} provides a method of constructing a nonuniformity monotone from any convex function.
\begin{theorem}\label{thm:nonuniformitymonotones}
 For every convex function $g:\bbR_+\rightarrow\bbR$, the function $M_g:\rR\to\bbR$ defined by
\[
M_g(x)\assign G\left(x \Big\| m^{(d_x)}\right)
= \frac{1}{d_x}\sum\limits_{i=1}^{d_x} g\left(d_x x_i\right)
 \]
 is a nonuniformity monotone.
\end{theorem}

\begin{proof}
We must show that the two conditions of Proposition~\ref{prop:NUMon} hold. 
It is given that $M_g(x) = G\left(x \Big\| m^{(d_x)}\right)$, therefore it follows from Theorem~\ref{thm:Schurrelativeadjoining} that for any $d'\in \bbZ_+$,  $M_g\left(x \otimes m^{(d')}\right)=M_g(x)$, and it follows from Theorem~\ref{thm:Schurconvexrelativetoq} that for any $d \in \bbZ_+$, the restriction of $M_g$ to a $d$-dimensional space is Schur-convex relative to $m^{(d)}$, hence Schur-convex.  
\end{proof}

Note that if a Schur-convex function $G$ can be expressed in the form $G(x) \assign \sum\limits_{i=1}^{d_x} g(x_i)$ for some function $g:\bbR\rightarrow\bbR$, then $g$ is convex [Proposition C.1.c on p.\ 95 of MOA]. Therefore, the theorem above defines a nonuniformity monotone for every Schur-convex function that is expressible in this form.

Recall that a set of nonuniformity monotones that determines the noisy equivalence class of a state is called a \emph{complete} set of monotones.  The set of monotones $M_g$ that is defined by the set of \emph{all} convex functions $g$ is a complete set~\cite{hardy1952inequalities}. We have also seen that the monotones describing the height of the Lorenz curve at every point are a complete set.  For other examples, see Ch.\ 4, Sec.\ B of MOA, and \cite{RuchSchrannerSeligman,joe1990majorization}.

\subsubsection{Schur-convexity-preserving functions}

It is easy to see directly from Definition~\ref{DefSchurConvexity} that
the composition of any Schur-convex function $G_d:\bbR_+^d \to \bbR$ with a function $f:\mathbb{R}\to\mathbb{R}$ that is {\em increasing} on its domain
gives another Schur-convex function $f\circ G_d:\bbR_+^d \to \bbR$.  Similarly, the composition of any Schur-convex function relative to $q$, $G(\cdot\|q):\bbR_+^{d_q} \to \bbR$, with an increasing function $f$ gives another function that is Schur-convex relative to $q$, $f\circ G(\cdot\|q):\bbR_+^{d_q} \to \bbR$.


Schur-convex functions of $x$ can be understood as generalizations of the negative entropy, or ``negentropy'',  of $x$ and are therefore denoted by $-H_{(\cdot)}(x)$, with a subscript labelling the particular such function.
Similarly, functions of $x$ that are Schur-convex relative to $q$ can be naturally understood as generalizations of the relative entropy of $x$ to $q$ and are therefore denoted by $H_{(\cdot)}(x\|q)$. Finally, nonuniformity monotones of $x$ will be denoted $I_{(\cdot)}(x)$.  


\subsubsection{Examples, including order-$p$ R\'{e}nyi nonuniformities}
\label{SubsecRenyip}

Table~\ref{table:NUmonotones} provides examples of nonuniformity monotones, constructed from various convex functions using Theorem~\ref{thm:nonuniformitymonotones}. These are derived as monotones for the classical resource theory, but, by Corollary~\ref{corollary:classicalmonotonetoquantummonotone}, we can evaluate these monotones for the spectrum of a quantum state to obtain a monotone for the quantum resource theory. 

\begin{table*}[htdp]{A sampling of nonuniformity monotones derived from convex functions on the reals}
\begin{center}
\begin{tabular}{|c|c|c|c|}
\hline
Function pair & Schur-convex function on $\bbR^d_+$ & Schur-convex function relative to $q$ & Classical nonuniformity monotone \\
\hline
\hline
$g(a)$ (convex);& $-H_{(\cdot)}(x)\assign f \circ G_d (x)$ & $H_{(\cdot)}(x\|q) \assign f \circ G(x\|q)$  &  $I_{(\cdot)}(x) \assign  H_{(\cdot)} \left( x\Big\|m^{(d_x)}\right)$ \\
$f(t)$ (increasing)& $\displaystyle = f\left( \sum\limits_{i=1}^d g(x_i) \right)$ & $\displaystyle = f\left(\sum\limits_{i=1}^{d_q} q_i\, g\left(\frac{x_i}{q_i}\right) \right)$  & $= f \left(\sum\limits_{i=1}^{d_x} \frac{1}{d_x}\, g\left(d_x x_i\right) \right)$
 \\
\hline
\hline
$\begin{array}{c}g(a)=a^2 \\ f(t)=t\end{array}$ & $\displaystyle\sum_{i=1}^d x_i^2$ & $\displaystyle\sum_{i=1}^{d_q} \frac{x_i^2}{q_i}$ & $\displaystyle d_x\sum_{i=1}^{d_x} x_i^2$ \\
\hline

$\sqrt{1+a^2}$ &&&Amato index \\
$t$ & $ \displaystyle\sum\limits_{i=1}^d\sqrt{1+x_i^2}$ &$\displaystyle\sum\limits_{i=1}^{d_q}q_i \sqrt{1+ \frac{x_i^2}{q_i^2}}$ &$\displaystyle I_{\textnormal{Amato}}(x)\assign\sum\limits_{i=1}^{d_x}\sqrt{\frac1{d_x^2}+x_i^2}$ \\
\hline

$a \log a$ & Shannon negentropy & Relative Shannon entropy to $q$ & Shannon nonuniformity \\
$ t$ & $\displaystyle -H(x)\assign\sum\limits_{i=1}^d x_i \log x_i$ & $\displaystyle H(x\|q) := \sum\limits_{i=1}^{d_q} x_i \log\left(\fr{x_i}{q_i}\right)$& $\begin{array}{rl}I(x)\assign&H\left(x\Big\|m^{(d_x)}\right)\\=&\log d_x - H (x)\end{array} $
 \\
\hline

$- \log a$  & Burg negentropy & Relative Burg entropy to $q$ & Burg nonuniformity \\
$t$ 
& $\displaystyle -H_{\textnormal{Burg}} (x) := - \sum\limits_{i=1}^d \log x_i  $ 
& $ \begin{array}{rl}H_{\textnormal{Burg}} (x\| q)&:= \displaystyle \sum\limits_{i=1}^{d_q} q_i \log \left(\fr{q_i}{x_i}\right)\\&\left(=H(q\|x)\right)\end{array}$
& $\begin{array}{rl}I_{\textnormal{Burg}} (x)&:= H_{\textnormal{Burg}}\left(x\Big\|m^{(d_x)}\right)\\&=\log d_x-H_{\textnormal{Burg}}(x)\\\Big(&\left.=H\left(m^{(d_x)}\Big\|x\right)\right)\end{array}$\\
\hline

$\pm a^p$ & Order-$p$ Tsallis negentropy & Order-$p$ Relative Tsallis entropy to $q$ & Order-$p$ Tsallis nonuniformity \\
$\begin{array}{c}\displaystyle\frac{\sgn(p)}{p-1}\left(\strut\pm t-1\right) \\ {\scriptstyle p\in\mathbb{R}\setminus\{0,1\}}\end{array}$ & $\displaystyle -H^{\textnormal{Ts}}_p(x)\assign\frac{\sgn(p)}{p-1}\left(\sum_{i=1}^{d} x_i^p-1\right)$&$\displaystyle\begin{array}{rl}H^{\textnormal{Ts}}_p(x\|&q)\assign\\ \displaystyle\frac{\sgn(p)}{p-1}&\left(\sum\limits_{i=1}^{d_q}x_i^p q_i^{1-p}-1\right)\end{array}$&$\begin{array}{l}I^{\textnormal{Ts}}_p(x)\assign H^{\textnormal{Ts}}_p\left(x\Big\|m^{(d_x)}\right)\\= d_x^{p-1} \left( H^{\textnormal{Ts}}_p\left(m^{(d_x)}\right) - H^{\textnormal{Ts}}_p(x)\right)\end{array}$\\
\hline
\hline
 
$\pm a^p$ & Order-$p$ R\'enyi negentropy & Order-$p$ Relative R\'enyi entropy to $q$ & Order-$p$ R\'enyi nonuniformity \\

$\begin{array}{c}\displaystyle\frac{\sgn(p)}{p-1}\log\left(\strut \pm t\right) \\ {\scriptstyle p\in\mathbb{R}\setminus\{0,1\}}\end{array}$
& $\displaystyle -H_p(x)\assign\frac{\sgn(p)}{p-1}\log \sum_{i=1}^d x_i^p $ & $\displaystyle H_p(x\|q)\assign\frac{\sgn(p)}{p-1} \log \left(\sum\limits_{i=1}^{d_q} x_i^p q_i^{1-p}\right)$ & $\begin{array}{rl}I_p(x)\assign&H_p\left(x\Big\|m^{(d_x)}\right)\\=&\sgn(p)\log d_x - H_p(x)\end{array}$\\
\hline

$\delta_{a,0}-1\equiv -a^0$& Order-0 R\'enyi negentropy & Order-0 Relative R\'enyi entropy & Order-0 R\'enyi nonuniformity \\ 
$-\log(-t)$&$\begin{array}{c}\textnormal{(neg-max-entropy)}\\-H_0 (x)\assign -\log\left|\supp (x)\right|\end{array}$ &$\begin{array}{c}\textnormal{(relative max-entropy) to $q$}\\ \displaystyle H_0(x\|q)\assign - \log \left( \sum\limits_{i\in \supp(x)} q_i \right)\end{array}$ &$ \begin{array}{rl}I_0(x)\assign& H_0\left(x\Big\|m^{(d_x)}\right)\\=&\log d_x - H_0 (x)\end{array}$\\
\hline

& Order-$\infty$ R\'enyi negentropy & Order-$\infty$ Relative R\'enyi entropy & Order-$\infty$ R\'enyi nonuniformity \\ 
&$\begin{array}{c}\textnormal{(neg-min-entropy)}\\-H_\infty(x)\assign\log x_1^{\downarrow}\end{array}$ &$\begin{array}{c}\textnormal{(relative min-entropy) to $q$}\\ \displaystyle H_\infty(x\|q)\assign\log\max\limits_{i=1,\dots,d_q}\left(\fr{x_i}{q_i}\right)\end{array}$ &$ \begin{array}{rl}I_\infty(x)\assign& H_\infty\left(x\Big\|m^{(d_x)}\right)\\=&\log d_x - H_\infty (x)\end{array}$\\
\hline

& Order-($-\infty$) R\'enyi negentropy & Order-($-\infty$) Relative R\'enyi entropy & Order-($-\infty$) R\'enyi nonuniformity \\ 
 & $\begin{array}{rl}-H_{-\infty} (x)\assign&-\log\min\limits_{i=1,\dots,d} x_i\\=&-\log x_d^{\downarrow}\end{array}$ & $\begin{array}{rl}H_{-\infty}(x\|q)\assign& \displaystyle -\log\min\limits_{i=1,\dots,d_q}\left(\fr{x_i}{q_i}\right)\\\big(=& H_{\infty}(q\|x)\big)\end{array}$& $\begin{array}{rl}I_{-\infty}(x)\assign&H_{-\infty}\left(x\Big\|m^{(d_x)}\right)\\=&-\log d_x - H_{-\infty} (x) \\\Big(=&H_{\infty}\left(m^{(d_x)} \Big\| x\right)\Big)\end{array}$\\

\hline
\hline
\end{tabular}
\caption{The first column specifies a function $g$ which is convex and another function $f$ which is increasing (on the relevant domains of definition), and the second specifies the Schur-convex family that they define. In the third column, we define the associated relative Schur-convex function to a distribution $q$. Finally, in the fourth column, we define a nonuniformity monotone as the relative Schur-convex function to the uniform state. The last four rows concern the nonuniformity monotones based on R\'enyi entropies. The symbol $\pm$ denotes a sign that depends on $p$. It is $+$ for $p<0$ and $p>1$, and $-$ for $0<p<1$ (formally, it equals $\sgn(p(p-1))$).}
\label{table:NUmonotones}
\end{center}
\end{table*}

Suppose we take our convex function $g:\bbR_+ \to \bbR$ to be
\[
g(a)\assign a^2.
\]
We first construct from this function a family of Schur-convex functions $\{ G_d : d\in \bbZ_+\}$ where $G_d:\bbR^d_+\to\bbR$, using the method of Lemma~\ref{lemma:SchCon}:
\[
G_d(x)\assign\sum\limits_{i=1}^d x_i^2.
\]
We then derive, for each $G_d$, the associated Schur-convex function relative to some distribution $q\in\rR$, denoted $G(\cdot\|q)$:
\[
   G(x\|q)\assign\sum\limits_{i=1}^{d_q}\fr{x_i^2}{q_i}.
\]
Finally, we construct a nonuniformity monotone $M_g:\rR\to\bbR$ from this function by taking $q$ to be the uniform distribution of dimension equal to that of the argument:
\[
M_g(x)\assign G(x\|m^{(d_x)})
= d_x\sum\limits_{i=1}^{d_x}x_i^2.
\]

The reader can easily follow this process of definition for each of the examples provided in Table~\ref{table:NUmonotones}.  We here provide some comments on these.

The Amato index is the nonuniformity monotone defined from the convex function $g(a)\assign\sqrt{1+a^2}$.  It is a simple example of a monotone which is also easily justified by considering the geometry of the Lorenz curve, as we will demonstrate in Section~\ref{subs:Geo}.

Perhaps the most paradigmatic example of a Schur-convex function is the negative of the Shannon entropy, equivalently, the Shannon \emph{negentropy}, which is generated by Lemma~\ref{lemma:SchCon} from the convex function $g(a)=a \log a$ (in this article, all logarithms are base 2). Using the standard notational convention of $H(x)$ for the Shannon entropy, the Shannon negentropy is simply
$$-H(x)\assign\sum\limits_{i=1}^{d_x} x_i \log x_i.$$
The relative entropy that $g(a)=a \log a$ defines is 
\[
H(x\|q) \assign \sum\limits_{i=1}^{d_x} x_i \log \left(\frac{x_i}{q_i}\right),
\]
which we term the \emph{relative Shannon entropy} of $x$ to $q$.  It is also known as the Kullback--Leibler divergence of $x$ from $q$ \cite{kullback1951information}. 
This yields the nonuniformity monotone
\begin{eqnarray*}
I(x) &\assign& H(x\|m)= \sum\limits_{i=1}^{d_x} x_i \log \left(d_x x_i \right)\\
&=& \log d_x - H(x),
\end{eqnarray*}
which we term the \emph{Shannon nonuniformity} of $x$.

Starting from the convex function $g(a) = -\log a$, we obtain as the corresponding Schur-convex function the negative of the \emph{Burg entropy} \cite{Burg},
\[
-H_{\textnormal{Burg}} (x) := - \sum\limits_{i=1}^{d_x} \log x_i .
\]
 We find the relative Schur-convex function that this defines to be
\[
H_{\textnormal{Burg}} (x\| q) := \sum\limits_{i=1}^{d_x} q_i \log\left(\fr{q_i}{x_i}\right).
\]
It is straightforward to verify that 
\[
H_{\textnormal{Burg}}(x\|q) = H(q\|x).
\]
So the  relative Burg entropy of $x$ to $q$ is just the relative Shannon entropy of $q$ to $x$.  

In the context of athermality theory, where $q$ is the thermal state, the analogous quantity has been discussed in \cite{janzing2000thermodynamic}.

As the next example, consider the one-parameter family of functions
\begin{equation}
   g_p(a):=\pm a^p\qquad (p\in\mathbb{R}\setminus\{0,1\}),
   \label{eqGP}
\end{equation}
where the sign is chosen such that $g_p$ is convex (that is, $+$ for $p<0$ and $p>1$, and $-$ for $0<p<1$).
We first obtain that $G^{(p)}(x):=\pm\sum_{i=1}^{d_x} x_i^p$ is Schur-convex. 
Setting $f(t)\assign \sgn(p)(\pm t -1)/(p-1)$ yields the \emph{order-$p$ Tsallis negentropy}~\cite{Tsallis} $-H_p^{\rm Ts}:=f\circ G^{(p)}$; that is,
\[
   -H_p^{\rm Ts}(x)\assign \frac{\sgn(p)}{p-1}\left(\sum_{i=1}^{d_x} x_i^p - 1\right)
\]
which is therefore Schur-convex. Note that the order-$p$ Tsallis entropy is usually defined as $\frac1{p-1}\left(1-\sum_i x_i^p\right)$ \cite{Tsallis} so that our terminology only coincides with the standard one for $p\ge0$, and differs by a negative sign for $p<0$.  Nonetheless, we are here adopting the convention that the term `entropy'  (respectively, `negentropy') should be reserved for functions that are Schur-concave (respectively, Schur-convex).

Applying Theorem~\ref{thm:nonuniformitymonotones} (which remains valid under composition with $f$), we
obtain the order-$p$  relative Tsallis entropy to $q$~\cite{TsallisRel}:
\[
H^{\textnormal{Ts}}_p(x\|q) 
\assign\frac{\sgn(p)}{p-1}\left(\sum\limits_{i=1}^{d_x}x_i^p q_i^{1-p}-1\right).
\]
(Again, we deviate from the standard definition by a negative sign for $p<0$.)
Finally, taking $q$ to be the uniform distribution, we obtain what we call the order-$p$ Tsallis nonuniformity: 
\begin{align}
I^{\textnormal{Ts}}_p(x)
\assign& H^{\textnormal{Ts}}_p\left(x\Big\|m^{(d_x)}\right) \nonumber\\
=& d_x^{p-1} \left( H^{\textnormal{Ts}}_p\left(m^{(d_x)}\right) - H^{\textnormal{Ts}}_p(x)\right). \nonumber
\end{align}

An interesting fact about 
$g_p(a)$ is that in the limit $p\to1$, it converges to $a\ln a$, so that the order-$p$ Tsallis nonuniformity converges to the Shannon nonuniformity (up to the multiplicative factor $\ln2$) in the limit $p\to 1$.
The $p\to0^+$ and $p\to\infty$ limits are also important because they are parent quantities of nonuniformity monotones based on the max- and min-entropies, which will be discussed shortly.

Repeating the construction above, with $g_p$ as in~(\ref{eqGP}), but alternative choice of function
$f(t):=\sgn(p)\log(\pm t)/(p-1)$, we obtain the order-$p$ R\'enyi neg\-entropy
\[
   -H_p(x)\assign \frac{\sgn(p)}{p-1}\log\sum_{i=1}^{d_x} x_i^p,
\]
which is again Schur-convex. Moving on, we construct
the order-$p$ relative R\'enyi entropy of $x$ to $q$,
\[
   H_p(x\|q)\assign \frac{\sgn(p)}{p-1}\log\sum_{i=1}^{d_x} x_i^p q_i^{1-p},
\]
and the nonuniformity monotone
\[
   I_p(x)\assign H_p\left(x\left\|m^{(d_x)}\right.\right) = \sgn(p)\log d_x-H_p(x)
\]
which we call the order-$p$ R\'enyi nonuniformity of $x$.

Although the functions $-H_p$, $H_p(\cdot\|q)$ and $I_p$ are not defined at $p\in\{0,1\}$, they do converge in these limits and in the limits $p\to\pm\infty$. 

The $p\to1$ limit is the least interesting: it just yields the functions $H$, $H(\cdot\|q)$ and $I$ (the Shannon negentropy, Shannon relative entropy and Shannon nonuniformity).

The limits $p\to0^+$ and $p\to\infty$ are of special significance. They will reappear below in different contexts, and are related to the min- and max-entropies. Since there are different versions of these entropies in the literature, we spell out the details in formal definitions. At this point, we are using the definitions by Renner \cite{RennerThesis}, but we will later depart from his conventions in the case of smoothed entropies.

The limit $p\to0^+$ yields the Schur-convex function
\[
-H_0(x)\assign -\log\left|\supp (x)\right|,
\]
which is the negative of what is traditionally called the \emph{max-entropy of $x$} (because $H_0(x)$ attains the maximum value among the $H_p(x)$)~\cite{gorban2010entropy}.  Hence we call this function the \emph{neg-max-entropy}. The corresponding Schur-convex function relative to $q$ is  called the \emph{relative max-entropy of $x$ to $q$}\footnote{It has also been called the min-relative entropy in \cite{datta2009min}  because among relative entropies, it has the minimum value.}:
\[
H_0(x\|q)\assign-\log\left( \sum\limits_{i\in \supp(x)} q_i \right),
\]
whence the resulting nonuniformity monotone is
\begin{align}\label{eq:defnI0}
I_0(x)&\assign H_0\left(x\Big\|m^{(d_x)}\right)\nonumber \\
&=\log d_x-\log\left|\supp(x)\right| \nonumber \\
&=\log d_x- H_0(x). 
\end{align}

The limit $p\to\infty$, on the other hand, yields the Schur-convex function
\begin{align}\label{Hinfty}
-H_\infty(x)\assign&\log\max\limits_{i\in\left\{1,\dots,d_x\right\}}x_i\nonumber\\
=&\log\left(x^\downarrow_1\right),
\end{align}
which is the negative of what is traditionally known as the \emph{min-entropy of $x$}, and which we term the \emph{neg-min-entropy}. The corresponding Schur-convex function relative to $q$ is what is called the \emph{relative max-entropy of $x$ to $q$} \cite{datta2009min}:
\[
H_\infty(x\|q)\assign\log\max\limits_{i \in \left\{1, \dots,d_x\right\}}\left(\fr{x_i}{q_i}\right),
\]
whence the resulting nonuniformity monotone is
\begin{align}\label{Iinfty}
I_\infty(x)&\assign H_\infty\left(x\Big\|m^{(d_x)}\right)\nonumber\\
&=\log d_x+\log\left(x^\downarrow_1\right) \nonumber\\
&=\log d_x - H_{\infty}(x).
\end{align}

Finally, consider the limit $p\to - \infty$.  Here we find that

\begin{align}
-H_{-\infty}(x) &\assign-\log\min\limits_{i\in\left\{1,\dots,d_x\right\}} x_i, \nonumber \\
H_{-\infty}(x\|q) &\assign-\log\min\limits_{i\in\left\{1,\dots,d_x\right\}}\left(\tfrac{x_i}{q_i}\right), \nonumber \\
I_{-\infty}(x) &\assign H_{-\infty}\left(x\Big\| m^{(d_x)}\right) \nonumber \\
 &=-\log d_x - H_{-\infty} (x).\nonumber
\end{align}

It is straightforward to verify that 
\[
H_{-\infty}(x\|q) = H_{\infty}(q\|x).
\]

\begin{remark}\label{rem:MaxWgtMix}
$H_{-\infty}(x\|q)$ admits of a simple operational interpretation: $2^{-H_{-\infty}(x\|q)}$ is the maximum probability with which the state $q$ can appear in a mixture of states that yields $x$\footnote{The analogue of this property for the quantum version of this quantity has been noted in~\cite{DattaPItalk}.} .   The proof is straightforward:  if the state $q$ appears with weight $w$ in a mixture of states that yields $x$, then $x = w q +(1-w) z$ for some state $z$.  In this case, all the components of $x-wq$ must be positive.  Consequently,
$$w \le \min\limits_{i\in\left\{1,\dots,d_x\right\}} \left(\tfrac{x_i}{q_i}\right) = 2^{-H_{-\infty}(x\|q)}.$$
 
Because the order-($-\infty$) Renyi nonuniformity, $I_{-\infty}$, is obtained by setting $q=m$ in $H_{-\infty}(x\|q)$, it follows that this monotone quantifies the largest probability with which the uniform state $m$ can appear in a convex decomposition of $x$.\footnote{As such, it is analogous to the one from entanglement theory called ``the separability of a state'' and which is defined as the maximum probability with which a separable state appears in a convex decomposition of the state~\cite{Lewenstein1998}.}
\end{remark}
 
The R\'{e}nyi $p$-nonuniformities will have an important role to play in the discussion of state conversion with a catalyst.  In anticipation of those results, we note two facts about these monotones.  First, they are all additive, that is, 
\begin{equation}\label{eq:Renyiadditivity}
I_p(x \otimes y) = I_p(x) + I_p(y)
\end{equation}
for all $p \in \bbR$. This follows from the simple identity
\[
   \sum_{i,j}(x_i y_j)^p =\left(\sum_i x_i^p\right)\left(\sum_j y_j^p\right).
\]
Secondly, the different R\'{e}nyi entropies are ordered as follows:
\begin{align}
I_{-\infty}(x) \ge  I_{p'}(x) \ge I_{p}(x),\;\;&p'<p<0,\label{eq:order1}\\
I_0(x) \le I_{p}(x) \le I_{p'}(x) \le I_{\infty}(x),\;\;&0<p<p'. \label{eq:order2}
\end{align}
Inequalities~\eqref{eq:order2} follow from a similar order relation over the $p$-norms $\|x\|_p\equiv \left(\sum_i x_i^p\right)^{1/p}$,
while \eqref{eq:order1} follow from Lemma 17 of \cite{1ShotAtherm2}. There, it is proven that for any distribution $x$ with full support, $H_p(x)$, considered as a function of $p$, is nondecreasing over the range $p<0$, whereas for any $x$ whose support is not the full sample space, $H_p(x)\to-\infty$ for all $p<0$. Consider first the case of an $x$ with full support. We find that
\begin{align}
\lim\limits_{p\to0^-}H_p(x)&=-\log d_x;\nonumber\\
\lim\limits_{p\to0^+}H_p(x)=:&H_0(x)=\log d_x.\nonumber
\end{align}
Therefore,
\begin{align}
\lim\limits_{p\to0^-}I_p(x)=-\log d_x&-\lim\limits_{p\to0^-}H_p(x)=0;\nonumber\\
I_0(x)=\log d_x-&H_0(x)=0.\nonumber
\end{align}

On the other hand, for some $x$ whose support is a proper subset of the sample space,
\begin{align}
\lim\limits_{p\to0^-}I_p(x)=&-\log d_x-\lim\limits_{p\to0^-}H_p(x)=\infty;\nonumber\\
I_0(x)=\log d_x&-H_0(x)=\log\left(\frac{d_x}{|\supp(x)|}\right).\nonumber
\end{align}

It follows from the inequalities in \eqref{eq:order1}, \eqref{eq:order2} that for any $x$ (fully-supported or otherwise),
\begin{align}\label{eq:infsupRenyi}
I_0(x) &= \inf_{p\in \bbR} I_p(x)\nonumber \\
I_{\infty}(x) &= \sup_{p \in \bbR_+} I_p(x).
\end{align}

See also \cite{vEH2010} for a treatment of the connection of the R\'enyi relative entropies with majorization and Lorenz curves.

\subsection{Nonuniformity monotones arising from the geometry of the Lorenz curve}
\label{subs:Geo}
It is also interesting to consider nonuniformity monotones that are inspired by features of the Lorenz curve. Lorenz curves were originally introduced in economics to characterize income inequality (MOA). There are a number of measures of income inequality that were defined in terms of the Lorenz curve. These immediately yield interesting nonuniformity monotones. 

The \emph{Gini index} of $x$ is the area of the Lorenz curve of $x$ over the diagonal \cite{gini1912variability}\footnote{In fact, the Gini index is usually defined as twice this area.} (see
also F.4.a on p.\ 563 of MOA). This is clearly a monotone because if the Lorenz curve of $x$ is nowhere below the Lorenz curve of $y$, then the area over the diagonal of the Lorenz curve of $x$ is not less than that of $y$. A straightforward calculation yields the value of the Gini index of $x$ in terms of its components:
$$M_{\textnormal{Gini}}(x)=\frac{d_x-1}{2d_x}-\frac1{d_x}\sum\limits_{i=2}^{d_x}(i-1)x^\downarrow_i.$$

The \emph{Schutz index} of $x$ is the maximum vertical deviation between the Lorenz curve of $x$ and the line joining (0,0) to (1,1)\cite{Schutz1951} (see also F.4.f on p.\ 565 of MOA). If the Lorenz curve of $x$ is everywhere above the Lorenz curve of $y$, the vertical deviation for $y$ at any abscissa is no greater than that of $x$ at the same abscissa. Consequently, the same holds for the maximum vertical deviation. By virtue of the monotonicity of the vertical deviation over any of the linear segments of the Lorenz curve of $x$, the Schutz index of $x$ is given in terms of its components as
$$M_{\textnormal{Schutz}}(x)=\max\limits_{k\in 1,\dots,d_x}\left(\sum\limits_{i=1}^kx^\downarrow_i-\frac k{d_x}\right).$$

The length of the Lorenz curve of $x$ is also a nonuniformity monotone.  In the context of income inequality, it was proposed in  \cite{amato1968metodologia} and later in ~\cite{kakwani1980income} (see also F.4.h on p.\ 565 of MOA).  We will call it the \emph{Amato index}.  It is seen to be a monotone by virtue of the fact that the Lorenz curves of $x$ and $y$ have the same boundary points and are convex, so that if the Lorenz curve of $x$ is nowhere below the Lorenz curve of $y$, it cannot be shorter. The Amato index of $x$ is given by
$$M_{\textnormal{Amato}}(x)=\sum\limits_{i=1}^{d_x}\sqrt{\frac1{d_x^2}+x_i^2}.$$
It was noted in Section~\ref{subs:ConMon} that this monotone can be constructed from the convex function $g(a)=\sqrt{1+a^2}$.

\begin{table*}[htdp]{A sampling of nonuniformity monotones derived from the geometry of the Lorenz curve}
\begin{center}
\begin{tabular}{|c|c|c|}
\hline 
Geometric feature of Lorenz curve& Nonuniformity monotone & Comments \\
\hline \hline
Height at point $u$ & $M_{\textnormal{Lorenz},u}(x):=L_x(u)=S_{\lfloor d_x u \rfloor}(x)
$   &   \\
$u\in [ 0,1]$ & $+ \left[ S_{\lfloor d_x u \rfloor+1}(x)
- S_{\lfloor d_x u \rfloor}(x) \right](d_x u- \lfloor d_x u \rfloor)$ & \\
& where $S_{k}\left(  x\right): =\sum_{i=1}^{k}x_{i}^{\downarrow}$ &\\
\hline
Area above diagonal & $M_{\textnormal{Gini}}(x)=\frac{d_x-1}{2d_x}-\frac1{d_x}\sum\limits_{i=2}^{d_x}(i-1)x^\downarrow_i$ & Gini index \\
\hline
Maximum vertical deviation & $M_{\textnormal{Schutz}}(x)=\max\limits_{k\in \{1,\dots,d_x\}}\left(\sum\limits_{i=1}^kx^\downarrow_i-\frac k{d_x}\right)$ & Schutz index \\
from diagonal & & \\ 
\hline
Length of Lorenz curve &$M_{\textnormal{Amato}}(x)=\sum\limits_{i=1}^{d_x}\sqrt{\frac1{d_x^2}+x_i^2}$ & Amato index \\
\hline
Slope of on-ramp & $m^{\textnormal{on}}(x)=2^{I_\infty(x)}$ & $I_\infty(x)$ is the order-($\infty$) Renyi nonuniformity \\
\hline
Length of tail & $\ell(x) = 1-2^{-I_0(x)}$  & $I_0(x)$ is the order-($0$) Renyi nonuniformity \\
\hline 
Negative of slope of off-ramp & $-m^{\textnormal{off}}(x)= -2^{-I_{-\infty}\left(x\right)}$ & $I_{-\infty}(x)$ is the order-($-\infty$) Renyi nonuniformity \\
\hline
\end{tabular}

\caption{
The first column specifies a geometric property of the Lorenz curve of the classical distribution $x$. The second specifies the classical nonuniformity monotone that this property defines. }

\label{table:LCmonotones}
\end{center}
\end{table*}

The nonuniformity monotones corresponding to various limits of the family of Renyi nonuniformities, specifically $I_{\infty},I_0$ and $I_{-\infty}$, have very simple interpretations in terms of the geometry of the Lorenz curve.  Indeed, this connection allows us to deduce that they are nonuniformity monotones merely from the characterization of noisy-majorization in terms of one Lorenz curve being everywhere not less than another.

For any Lorenz curve $L_x$, define the \emph{on-ramp} to be the first segment of the curve.  Its slope is clearly determined by the largest eigenvalue of $x$, $x_1^{\downarrow}$. Specifically, the slope of the on-ramp is $m^{\textnormal{on}}(x):=x^{\downarrow}_1 / (1/d_x)$.  Recalling the definition of the order-($\infty$) Renyi nonuniformity, Eq.~\eqref{Iinfty}, we conclude that
\begin{equation}
 m^{\textnormal{on}}(x)
 =2^{I_\infty(x)}.
 \label{eqSlope}
\end{equation}
If $x$ is mapped to $y$ by noisy operations, the Lorenz curve $L_y$ must lie on or below $L_x$, and consequently, the slope of the on-ramp of $L_y$ must be less than or equal to that of $L_x$. This is an intuitive way to see why $I_\infty$ is a monotone.

Similarly, the \emph{tail} of a Lorenz curve $L_x$, i.e.\ the right-most part where the curve is flat and attains the value one,
is related to the number of zero eigenvalues in $x$. The length $\ell(x)$ of this tail is $\ell(x) = \frac 1 {d_x}\left( d_x - | \textnormal{supp}(x)| \right)$, which by Eq.~\eqref{eq:defnI0}, implies that
\begin{equation}
\ell(x) = 1-2^{-I_0(x)}.
 \label{eqTailLength}
\end{equation}
Again, if $x$ is mapped to $y$ by noisy operations, the Lorenz curve $L_y$ must lie on or below $L_x$, and this
is possible only if the tail lengths satisfy $\ell(y)\leq \ell(x)$. This is a simple way of seeing why $I_0$ is a nonuniformity monotone.

Finally, we mention another monotone that is given by the geometry of the Lorenz curve. Defining the \emph{off-ramp} of the curve to be the final segment --- the one that touches the (1,1) point --- the monotone is the slope of the off-ramp. This is zero if the Lorenz curve has a nontrivial tail and is nonzero otherwise. 

 This slope is given by $m^{\textnormal{off}}(x)=2^{\log d_x + \log x^{\downarrow}_{d_x}}$
where $x^{\downarrow}_{d_x}$ is its \emph{smallest} component of $x$
(recall that the min-entropy is defined in terms of the \emph{largest} component of $x$).
Recalling the definition of $I_{-\infty}\left(x\right)$, the order-($-\infty$) Renyi nonuniformity, Eq.~\eqref{Iinfty}, we conclude that
\begin{align}
m^{\textnormal{off}}(x)&=2^{-I_{-\infty}\left(x\right)}.
\end{align}
Clearly, if the Lorenz curve of $x$ is nowhere below that of $y$, the off-ramp slope of $x$ is less than or equal to that of $y$.  Hence the negative of the off-ramp slope is a nonuniformity monotone, which is a simple way of seeing that $I_{-\infty}(x)$ is a monotone.

\subsection{Nonuniformity monotones from other Schur-convex functions}
\label{subs:OtherSC}
Not all Schur-convex functions lend themselves to a decomposition as the sum of a single convex function evaluated on the components, but we can still construct nonuniformity monotones from some of these.

\begin{proposition}
Given a distribution, 
 $q \in \bbR_+^{d_q}$,
and a non-negative function $g:\bbR_+\rightarrow\bbR_+$, the function $\Gamma(\cdot\|q):\bbR_+^{d_q} \to \bbR$ defined by
 \begin{equation}\label{logmajorization}
 \Gamma(x\|q)=\prod_{j=1}^{d_q}\left[g\left(\frac{x_j}{q_j}\right)\right]^{q_j}
 \end{equation}
is Schur-convex (Schur-concave) relative to $q$ if $\log g$ is convex (concave).
\end{proposition}

\begin{proof}
By taking the log on both sides of~\eqref{logmajorization} and using Lemma~\ref{thm:Schurconvexrelativetoq} we get that 
$\log(\Gamma(x\|q))$ is Schur-convex. Since the $\log$ function is monotonically increasing we get that $\Gamma(x\|q)$ must also be Schur-convex. The concave case follows the same lines.
\end{proof}

As an example, consider the function $g:\bbR_+\rightarrow\bbR_+$ defined by $g(a)=a$. Since $\log a$ is concave, the proposition above implies that the function $\prod_{j=1}^{d_q}\left(\frac{x_j}{q_j}\right)^{q_j}$ is Schur-concave relative to $q$ and therefore
$$
\Gamma_{WG}(x\|q)=1 - \prod_{j=1}^{d_q}\left(\frac{x_j}{q_j}\right)^{q_j}
$$
is Schur-convex relative to $q$. Note that the product is the \emph{weighted} geometric mean (hence `WG') of the elements $x_j/q_j$ with weights $q_j$.

For $q=m^{(d_x)}$ this example yields the nonuniformity monotone 
$$M_G(x)\assign1-d_x\left( \prod\limits_{i=1}^{d_x}x_i \right)^{\frac1{d_x}}.$$
Because the weighted geometric mean is uniformly weighted in this case, it is just the geometric mean (hence the subscript `$G$').
It is the analogue for nonuniformity of the $G$-concurrence entanglement monotone~\cite{Gourconcurrence}. 

Note that the product of components, $\prod\limits_{i=1}^{d_x}x_i$, and the geometric mean of the components, $\left( \prod\limits_{i=1}^{d_x}x_i \right)^{\frac1{d_x}}$ are both Schur-concave functions, but they generalize to the same nonuniformity monotone, namely $M_G(x)$, because of the requirement of eq.~(\ref{eq:NUmonotoneAddm}).

In fact, the product of all the components is just one among a whole class of Schur-concave functions: the elementary symmetric polynomials in the components. The $k$th such polynomial is defined as $S_k^{(d)}:\mathbb R^d_+\to\mathbb R$,
\begin{align}
S_k^{(d)}(x):=&\sum\limits_{i_1<i_2\ldots<i_k}\prod\limits_{j=1}^kx_{i_j},
\label{eq:ElSym}
\end{align}
where the sum is over all $k$-tuples of distinct components.
We have not found a way to construct nonuniformity monotones from the Schur-concave function $S_{k}^{(d)}$ except for the special cases of $k=d$ and $k=2$ (the $k=1$ case is trivial).

Finally, we summarize in Table~\ref{table:OptlIntnsMonotones} the operational interpretations of some of these measures, interpretations that will be discussed in the rest of this article.

\begin{table*}[htdp]
\begin{center}
\begin{tabular}{|c|c|c|}
\hline 
Nonuniformity monotone & Operational Interpretation & Comments \\
\hline 
\hline
$I_{-\infty}(x)\assign-\log d_x - H_{-\infty}(x)$ & Quantifies maximum weight of uniform state in decomposition& See Remark~\ref{rem:MaxWgtMix} \\
\hline
$I_0(x)\assign \log d_x - H_0(x)$ & Quantifies single-shot distillable nonuniformity & See Lemma~\ref{lemma:distillation} \\
\hline
$I_\infty(x)\assign \log d_x - H_\infty(x)$ & Quantifies single-shot nonuniformity of formation & See Lemma~\ref{lemma:formation} \\
\hline
$I(x)\assign \log d_x - H(x)$& Quantifies asymptotic rate of interconversion & See Lemma~\ref{LemAsymptoticRate} \\
 & Describes total order of states under approximate catalysis & See Theorem~\ref{ThmApproxCatalysis} \\
\hline
\end{tabular}
\caption{Operational interpretation of some nonuniformity monotones}
\label{table:OptlIntnsMonotones}
\end{center}
\end{table*}
 
\section{Exact state conversion}\label{Sec:Exactstateconversion}

\subsection{A standard form of nonuniformity: sharp states}

It is useful to define a nonuniform state that can serve as a \emph{natural unit} for the resource in question.  For the resource of nonuniformity, the pure state of a two-level system, that is, the distribution $(1,0)$, is an obvious choice.  We refer to such a state as a \emph{bit of nonuniformity} or as a \emph{pure bit}.
 
The family of pure states of different dimensions also provides a natural standard relative to which one can quantify other states' nonuniformity.  A pure state of dimension $d$ corresponds to the distribution 
$$(1,\underbrace{0,\dots,0}_{d-1}).$$
Clearly, $k$ pure bits constitute a pure state of dimension $2^k$, however, pure states for dimensions that are not powers of $2$ cannot be represented as an integer number of pure bits.  

Pure states of different dimensions are themselves subsumed as a special case of a third family of nonuniform states, which we will call the \emph{sharp states}.  It is the latter set which is the most versatile and which we will take in this article to be the standard relative to which we judge the nonuniformity of states.  A sharp state is any distribution of the form
 \begin{equation}\label{eq:sharpstate}
(\underbrace{\underbrace{\tfrac{1}{d_u},\dots,\tfrac{1}{d_u}}_{d_u},0,\dots,0}_d),
\end{equation}
where $d,d_u \in \bbZ_+$ and $d_u \le d$, that is, a state of dimension $d$ that is uniformly distributed over $d_u$ physical states and assigns probability zero to the rest.

If we evaluate the Shannon nonuniformity of such a state, we find  $\log (d/d_u).$ Furthermore, one easily verifies that the order-0 and order-$\infty$ R\'{e}nyi nonuniformities are equal to the Shannon nonuniformity for this state.  It then follows from Eq.~\eqref{eq:infsupRenyi} that for all $p\in \bbR_+$ the order-$p$ R\'{e}nyi nonuniformity is equal to the Shannon nonuniformity for this state.
Indeed, the noisy equivalence class of this state is the set of all sharp states with the same ratio of $d$ to $d_u$ and therefore on the set of sharp states \emph{every} nonuniformity monotone can be expressed as a function of the Shannon nonuniformity.  We will therefore adopt the convention of refering to $\log (d/d_u)$ as simply \emph{the nonuniformity} of the sharp state.  Clearly, the possible values for the nonuniformity of a sharp state is just the image of the rationals greater than 1 under the logarithmic function.  

We will adopt as our canonical representative of each equivalence class of sharp states the one of lowest dimension within that class, and we will label it by its nonuniformity.  The canonical sharp state with nonuniformity $I$ will be denoted $s_I$.  Thus, we have
\begin{align}\label{eq:sharpstates}
s_0 &\assign (1) \\ 
s_{\log (3/2)} &\assign (1/2,1/2,0) \\
s_1 &\assign (1,0) \\
s_ {\log (7/3)} &\assign (1/3,1/3,1/3,0,0,0,0) \\
s_{\log 3} &\assign (1,0,0) \\
s_2 &\assign (1,0,0,0) 
\end{align}

Note that a uniform state is a sharp state with nonuniformity zero.  The canonical sharp state in this class, $s_0$, is the uniform state of smallest dimension, namely, dimension $1$ (which we could also denote by $m^{(1)}$) but it is of course noisy-equivalent with the uniform state $m^{(d)}$ of any dimension.  A pure bit is the canonical sharp state with nonuniformity 1, $s_1$.  A pure state of dimension $d$ is the canonical sharp state with nonuniformity $\log d$, $s_{\log d}$.

\begin{figure}[h]
\centering
\includegraphics[width=.48\textwidth,clip=true]{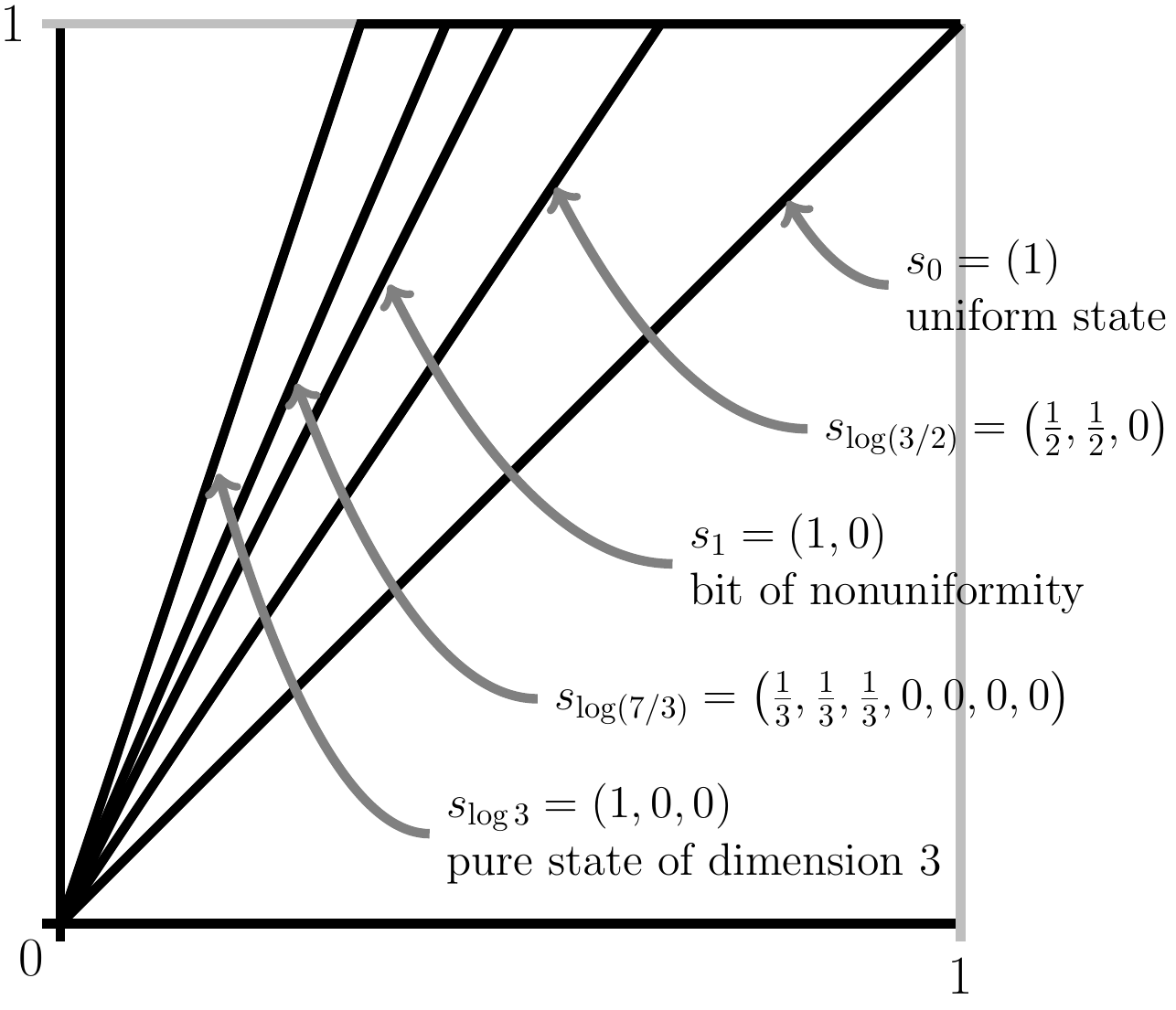}
\captionsetup{singlelinecheck=off,justification=raggedright}
\caption{The Lorenz curves of various sharp states.}
\label{fig:Sharp_states}
\end{figure}

The Lorenz curves representing $s_{I}$ for various values of $I$ are depicted in Fig.~\ref{fig:Sharp_states}.  It is easy to see that the set of sharp states includes all and only those states for which the Lorenz curve has the simple form of an on-ramp and a tail.  Recalling Eqs.~(\ref{eqSlope}) and (\ref{eqTailLength}), the on-ramp slope of a sharp state is $m^{\textnormal{on}}(s_{I}) = 2^I$, while the tail length is $\ell(s_{I})=1-2^{-I}$.  It is clear from the geometry of the Lorenz curves that the set of sharp states are totally ordered under noisy operations.  For any degree of nonuniformity that is equal to the logarithm of a real (possibly irrational) number, one can approximate this arbitrarily closely with a sequence of sharp states whose nonuniformity
converges to that degree of nonuniformity in the limit of arbitrarily large dimension.

Finally, note that a tensor product of two sharp states, one with nonuniformity $I_1$ and one with nonuniformity $I_2$ is another sharp state with nonuniformity $I = I_1 + I_2$, that is, $s_{I_1} \otimes s_{I_2}$ is noisy-equivalent with $s_{I_1 + I_2}$.

\subsection{Nonuniformity of formation and distillable nonuniformity}
\label{section:formationanddistillation}

Given this standard form of nonuniformity, two important questions arise:
\begin{itemize}
\item[(i)] What is the minimum nonuniformity of sharp state required to deterministically create (or ``form'') a single copy of state $x$ by noisy operations?  
\item[(ii)] Given a single copy of state $x$, what is the maximum nonuniformity of sharp state that can be deterministically extracted (or ``distilled'') from it by noisy operations?  
\end{itemize}
We refer to the answer to the first question as the \emph{single-shot nonuniformity of formation} and the answer to the second question as the \emph{single-shot distillable nonuniformity}.  
The term ``single-shot'' refers to our interest in forming or distilling only one copy of $x$.  

If we find that the minimum nonuniformity of any sharp state required to form $x$ is $I$, then 
the minimum number of pure bits required is $\lceil I \rceil$, where $\lceil a \rceil$ denotes the smallest integer not smaller than $a$ (the ceiling function).
Similarly, if the maximum nonuniformity of any sharp state that can be extracted is $I$, then the 
maximum number of pure bits that can be extracted is $ \lfloor I \rfloor$, where $\lfloor a \rfloor$ denotes the largest integer not larger than $a$ (i.e. the floor function).

Note that the possession of a pure bit is equivalent, as a resource, to having access to one implementation of a one-bit erasure operation\footnote{It is similarly evident that a sharp state with nonuniformity $\log d$ is equivalent, as a resource, to having access to an operation that takes an arbitrary state of dimenstion $d$ to a sharp state with nonuniformity $\log d$.}.  This follows from the fact that (a) given one pure bit as a resource, one can implement erasure on a system by swapping the system's state with the resource's state, and (b) given one implementation of a one-bit erasure operation, a uniform state (available for free) can be transformed into one bit of nonuniformity.  It is useful to think of this equivalence class of resources as the ability to do one bit of ``informational work''.  So distillation of nonuniformity is the analogue, within the resource theory of nonuniformity, of work extraction in the resource theory of athermality.  Indeed, work extraction often involves a two-step procedure wherein one first distills pure bits, and then uses them in a Szilard engine to do work \cite{WorkValOfInfo,LawsOfThermo}.  

Questions (i) and (ii) above are easily answered geometrically by means of Lorenz curves, as illustrated in Fig.~\ref{fig:NUDistillableAndCost}. This proof technique was first used in~\cite{FundLimitsNature}.

\begin{figure}[h!]
\centering
\includegraphics[width=.3\textwidth,clip=true]{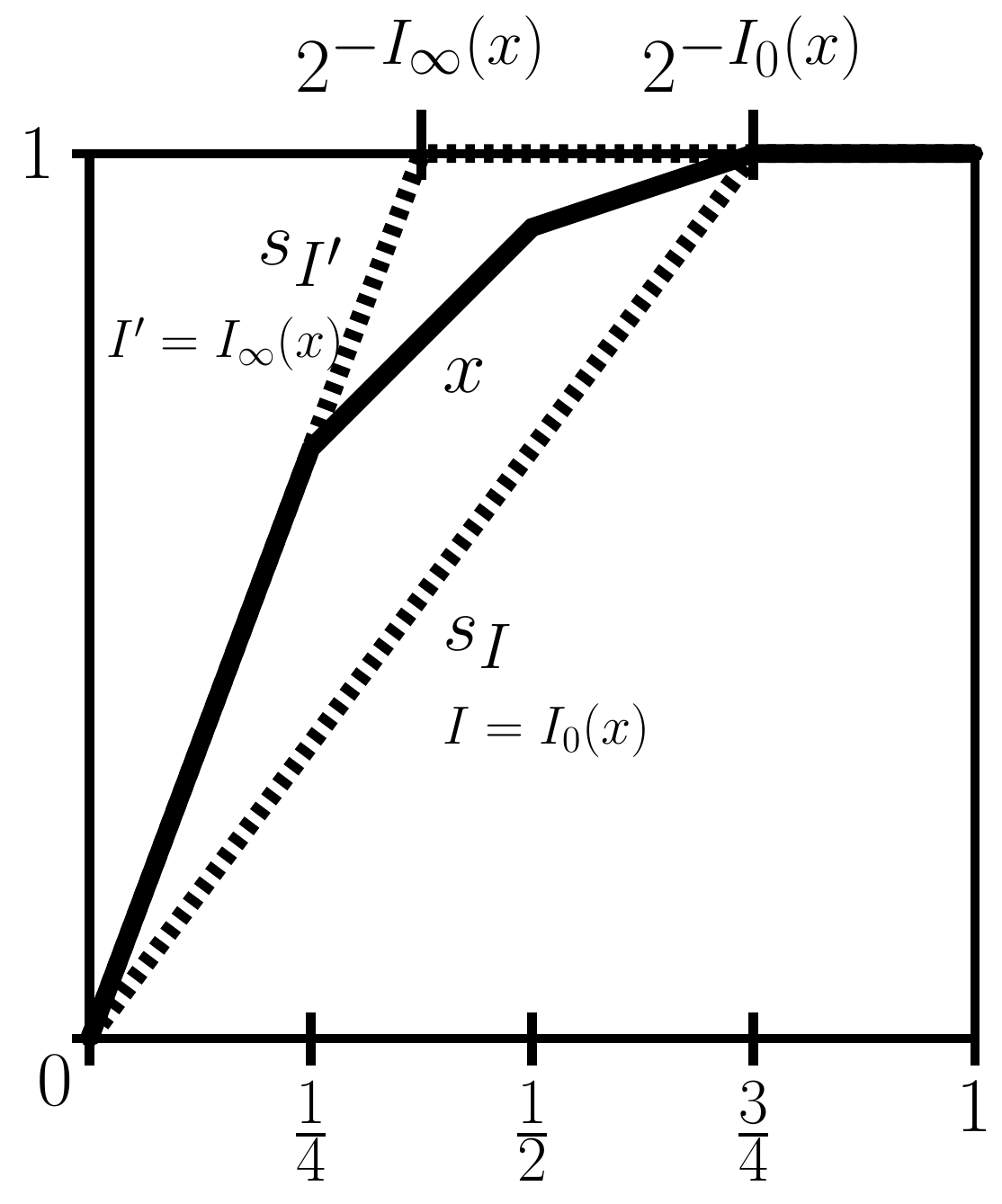}
\captionsetup{singlelinecheck=off,justification=raggedright}
\caption{From the Lorenz curve of $x$, one easily infers that $I_{0}(x)$ is the maximum nonuniformity of sharp state that can be distilled from $x$ and that  $I_{\infty}(x)$ is the minimum nonuniformity of sharp state that is required to form $x$. 
}
\label{fig:NUDistillableAndCost}
\end{figure}

For a given state $x$, its Lorenz curve $L_x$ lies above or on the Lorenz curve $L_{s_I}$ of a sharp state with nonuniformity $I$ if and only if it has at least as long a tail length, i.e., if and only if \ $\ell(x)\geq \ell(s_I)$.  Given the expression for the tail length of $x$, Eq.~(\ref{eqTailLength}), and the fact that the tail length of $s_I$ is simply $\ell(s_I) = 1- 2^{-I}$, the condition becomes $I \le I_0(x)$.  Since $I_0(x)$ is the logarithm of a rational, we can attain the bound using a uniform ancilla of finite dimension.

Similarly, there is a simple geometric condition that is necessary and sufficient to guarantee that the Lorenz curve $L_{s_I}$ of a sharp state with nonuniformity $I$ lies on or above the Lorenz curve $L_x$ of $x$: namely, the on-ramp slope
of $L_{s_I}$ must equal or exceed that of $L_x$, that is, $m^{\textnormal{on}}_{s_I}\geq m^{\textnormal{on}}_x$. According to Eq.~(\ref{eqSlope}) and the fact that the on-ramp slope of $L_{s_I}$ is $2^I$, this is equivalent to $I \geq I_{\infty}(x)$.  Here, we cannot always attain the bound with a finite-dimensional ancilla, since $I_{\infty}(x)$ can be the logarithm of an irrational number. However, we can approach equality arbitrarily closely by using large enough ancillary uniform states.

We can summarize what we have proven as follows.

\begin{proposition}\label{lemma:distillation}
The state conversion $x \cconv s_I$ is possible if and only if $$I \le I_0(x).$$ 
Consequently, the maximum nonuniformity of sharp state that can be distilled deterministically from a state $x$ by noisy operations is $I_0(x)$, the order-0 R\'{e}nyi nonuniformity of $x$.
\end{proposition}

\begin{proposition}\label{lemma:formation}
The state conversion $s_I \cconv x$ is possible if and only if $$I \ge I_{\infty}(x).$$ 
Consequently, the minimum nonuniformity of sharp state that is required to deterministically form a state $x$ by noisy operations is $I_{\infty}(x)$, the order-($\infty$) R\'{e}nyi nonuniformity of $x$.
\end{proposition}

These results yield an operational interpretation of $ I_\infty(x)$ ($ I_0(x)$) in terms of the single-shot nonuniformity cost (yield) of a state $x$\footnote{Note that it is not, strictly speaking, the single-shot cost (yield) of pure bits because the latter is the integer ceiling (floor) of this quantity.  Nonetheless, as long as one remembers this subtlety, pure bit cost (yield) is an accurate description and it is often described this way 
\cite{WorkValOfInfo,QuantLandauer,ThermoMeaningNegEntropy}.}.

Note that if the system is composite and the state is a product state, then by virtue of the additivity of $I_\infty$ and $I_0$, the distillable nonuniformity and the nonuniformity of formation of the composite state are just the sum of those of the components.  In other words, the maximum nonuniformity yield and the minimum nonuniformity cost are both achievable by processing  the components separately.  It also follows that if one cannot distill any nonuniformity from $x$ then one cannot do so from any number of copies of $x$ either.  The analogous fact was pointed out for athermality theory in \cite{janzing2000thermodynamic}.   

Also note that if the system consists of several components, correlations among the systems is a resource of nonuniformity.  For instance, a pair of bits which are perfectly correlated but with uniform marginals on each bit define a probability distribution $x \assign (1/2,0,0,1/2)$, corresponding to a sharp state with nonuniformity $1$.  This is noisy-equivalent to a single pure bit.  The point is that, in the resource theory of nonuniformity, all unitaries are free, including those that couple together distinct systems.  It follows that how nonuniformity is encoded within a composite system (locally or in correlations) is irrelevant.

In the resource theory of athermality, the problem that is the analogue of finding the maximum single-shot distillable nonuniformity is finding the maximum amount of work that can be extracted from a state in a single-shot protocol.  This was first solved 
in \cite{FundLimitsNature} and \cite{AbergWork}. That the single-shot distillable nonuniformity is $I_0(x)$ is simply the specialization to energy-degenerate systems of their result.  \cite{FundLimitsNature} also determined the work \emph{cost} of preparing a state in a single-shot protocol, which in the case of energy-degenerate systems is just the single-shot nonuniformity of formation, and their result reduces to  $I_{\infty}(x)$ in this case.

Finally, the results on single-shot formation and distillation of nonuniformity 
yield a simple sufficient condition for the possibility of a state conversion.

\begin{lemma}\label{lemma:gowitness}
If $x$ and $y$ are states such that
\[
I_0(x) \geq  I_\infty(y),
\]
then $x \cconv y$.

Equivalently, the function
\[
W(x,y)\assign I_0(x) -I_\infty(y)
\]
is a go witness for the state conversion $x \cconv y$, which is to say that if $W(x,y)\ge 0$ then $x \cconv y$.
\end{lemma}
\begin{proof}
The result follows from Propositions~\ref{lemma:distillation} and \ref{lemma:formation}.  The latter implies that a sharp state with nonuniformity $I_\infty(y)$ suffices to form $y$, and the former implies that starting from $x$, we can distill a sharp state with nonuniformity $I_0(x)$ (recall that $2^{I_0(x)}$ is always rational).   Consequently, if $I_0(x) \geq  I_\infty(y)$, we can distill from $x$ a sharp state that is
sufficient to form $y$.  
\end{proof}

\subsection{Nonuniformity cost and yield of state conversion}\label{sec:costyieldstateconv}

We have discussed the problem of whether or not $x$ can be mapped to $y$ by noisy operations.  We have also considered how much nonuniformity (in some standard form) may be extracted from a state, and how much nonuniformity is required to create the state.  There is an obvious way in which to combine these questions into another pair of questions:
\begin{itemize}
\item[(i)] If $x \cconv y$, what is the maximum nonuniformity of sharp state that can be distilled in addition to achieving the state conversion?
\item[(ii)] If it is not the case that $x \cconv y$, what is the minimum nonuniformity of sharp state that one requires to make the state conversion possible?
\end{itemize}
We call the former the \emph{nonuniformity yield} of the state conversion
and the latter the \emph{nonuniformity cost} of the state conversion.

The following lemma will be useful for answering these questions.
\begin{lemma}\label{lemma:adjoinsharp}
The operation of adjoining an ancillary system in sharp state $s_I$ to a system in state $x$, that is, $x \mapsto x \otimes s_I$, 
corresponds to the map  $L_x(u) \mapsto L_{x \otimes s_I}(u)$ on Lorenz curves, where
\begin{equation}\label{eq:LorenzxotimessR}
   L_{x \otimes s_I} (u)=\left\{
      \begin{array}{cl}
         &L_x (2^I u), \; \;   u \in [0, 2^{-I}] \\
         &1, \; \; \;\;\;\;\;\;\;\;\;\;\; u \in ( 2^{-I}, 1 ]
      \end{array}
   \right.
\end{equation}
This corresponds to a $2^I$-fold compression of the Lorenz curve of $x$ along the $u$-axis
while lengthening the tail an appropriate amount.
\end{lemma}

To see this, note that if $I=\log (d/d_u)$, then 
\begin{align}
&(y \otimes s_I)^{\downarrow}= \nonumber\\
&( \underbrace{y_1^{\downarrow}/d_u, \dots, y_1^{\downarrow}/d_u}_{d_u},\dots, \underbrace{ y_{d_y}^{\downarrow}/d_u, \dots , y_{d_y}^{\downarrow}/d_u}_{d_u}, \underbrace{0, \dots, 0}_{d_y(d-d_u)} ),
\end{align}
where the multiplicity of each nonzero coefficient is $d_u$ and the mulitplicity of the zeros is $d_y(d-d_u)$.  One then verifies that the Lorenz curve has the form described. 

The case of $I=0$ corresponds to adjoining an ancillary system in the uniform state, which leaves the Lorenz curve invariant.

It is also useful to define the following order relation on states.

\begin{definition}[$\lambda$-noisy-majorization]
We say that 
$x$ $\lambda$-noisy-majorizes $y$,
where $\lambda \in \bbR$, 
if
\begin{equation}\label{eq: LorenzmajorizationbyfactorR}
L_x(u) \ge L_y (2^{\lambda} u) \; \; \forall u \in [ 0, u_*],
\end{equation}
where $u_* \assign \min \{ 1, 2^{-\lambda} \}$, or equivalently,
\begin{equation}\label{eq: LorenzmajorizationbyfactorR2}
L_x(2^{-\lambda} u) \ge L_y ( u) \; \; \forall u \in [ 0, u_*^{-1}].
\end{equation}
\end{definition}

We will seek to understand this relation in terms of the geometry of Lorenz curves.  We consider the cases of $\lambda > 0$ and $\lambda <0$ in turn.

For $\lambda >0$, if $x$ $\lambda$-noisy-majorizes $y$ then \emph{not only} is the Lorenz curve of $x$ everywhere greater than or equal to the Lorenz curve of $y$, but it is also the case that one can implement a $2^\lambda$-fold compression of the Lorenz curve of $y$ along the $u$-axis (making it rise to $1$ more rapidly), and \emph{still} find that the Lorenz curve of $x$ is everywhere greater than or equal to it.  Equivalently, it asserts that one can implement a $2^\lambda$-fold \emph{stretching} of the Lorenz curve of $x$ along the $u$-axis (making it rise to $1$ less rapidly), and still find that the Lorenz curve of $x$ is everywhere greater than or equal to that of $y$.

From Lemma~\ref{lemma:adjoinsharp} one can also deduce the physical significance of the fact that  $x$ $\lambda$-noisy-majorizes $y$ when $\lambda > 0$: it implies that $x \cconv y\otimes s_{\lambda}$, that is, one can distill a sharp state of nonuniformity $\lambda$ in addition to achieving the conversion of $x$ to $y$.   To see this, note that Eq.~\eqref{eq: LorenzmajorizationbyfactorR} implies, via Eq.~\eqref{eq:LorenzxotimessR}, that $L_{x } (u) \ge L_{y\otimes s_{\lambda}}(u)$ in the region $u \in [ 0, u_*]$ where $u_*=2^{-\lambda}$.  The only subtlety then, is to prove that $L_x(u) \ge L_{y \otimes s_\lambda}(u)$ also in the region $u \in (u_*,1]$.   This follows from noting that at $u=u_*$, we have $L_{y\otimes s_{\lambda}}(u_*)= 1$,  but we also have $L_{x } (u_*) \ge L_{y\otimes s_I}(u_*)$, and therefore $L_{x } (u_*)=1$.  But because the Lorenz curve of $x$ is concave it follows that $L_{x } (u)=1$ for $u \in (u_*,1]$ and consequently it cannot be smaller than $L_y(u)$ in that region.

We now consider the case of $\lambda < 0$.
In this case, if $x$ $\lambda$-noisy-majorizes $y$ then the Lorenz curve of $x$ must somewhere be lower than that of $y$ (so that $x$ does \emph{not} noisy-majorize $y$), but if we implement a $2^{\lambda}$-fold compression of the Lorenz curve of $y$ along the $u$-axis, which, because $2^{\lambda}<1$ is a net \emph{stretching} of the Lorenz curve of $y$ along the $u$-axis (making it rise to $1$ less rapidly), then the Lorenz curve of $x$ becomes everywhere greater than or equal to that of $y$.  Equivalently, it asserts that a $2^{|\lambda|}$-fold net compression of the Lorenz curve of $x$ along the $u$-axis (making it rise to $1$ more rapidly) can make it everywhere greater than or equal to that of $y$.

Physically, if $x$ $\lambda$-noisy-majorizes $y$ for $\lambda < 0$, then $x \otimes s_{|\lambda|} \cconv y$, that is, a sharp state of nonuniformity $|\lambda|$ can make the conversion of $x$ to $y$ possible.  

This time, the proof begins with Eq.~\eqref{eq: LorenzmajorizationbyfactorR2}, which implies, via Eq.~\eqref{eq:LorenzxotimessR}, that $L_{x \otimes s_{|\lambda|}} (u) \ge L_{y }(u)$ in the region $u \in [ 0, u_*^{-1}]$ where $u_*^{-1}=2^{\lambda}$. Again, one can easily infer that this inequality also holds in the region $u \in (u_*^{-1},1]$, and so holds for all $u \in [0,1]$.

Standard noisy-majorization is $\lambda$-noisy-majorization with $\lambda=0$. 

\begin{figure*}
\centering
\subfloat[]{
\centering
\includegraphics[width=.3\textwidth, clip=true]{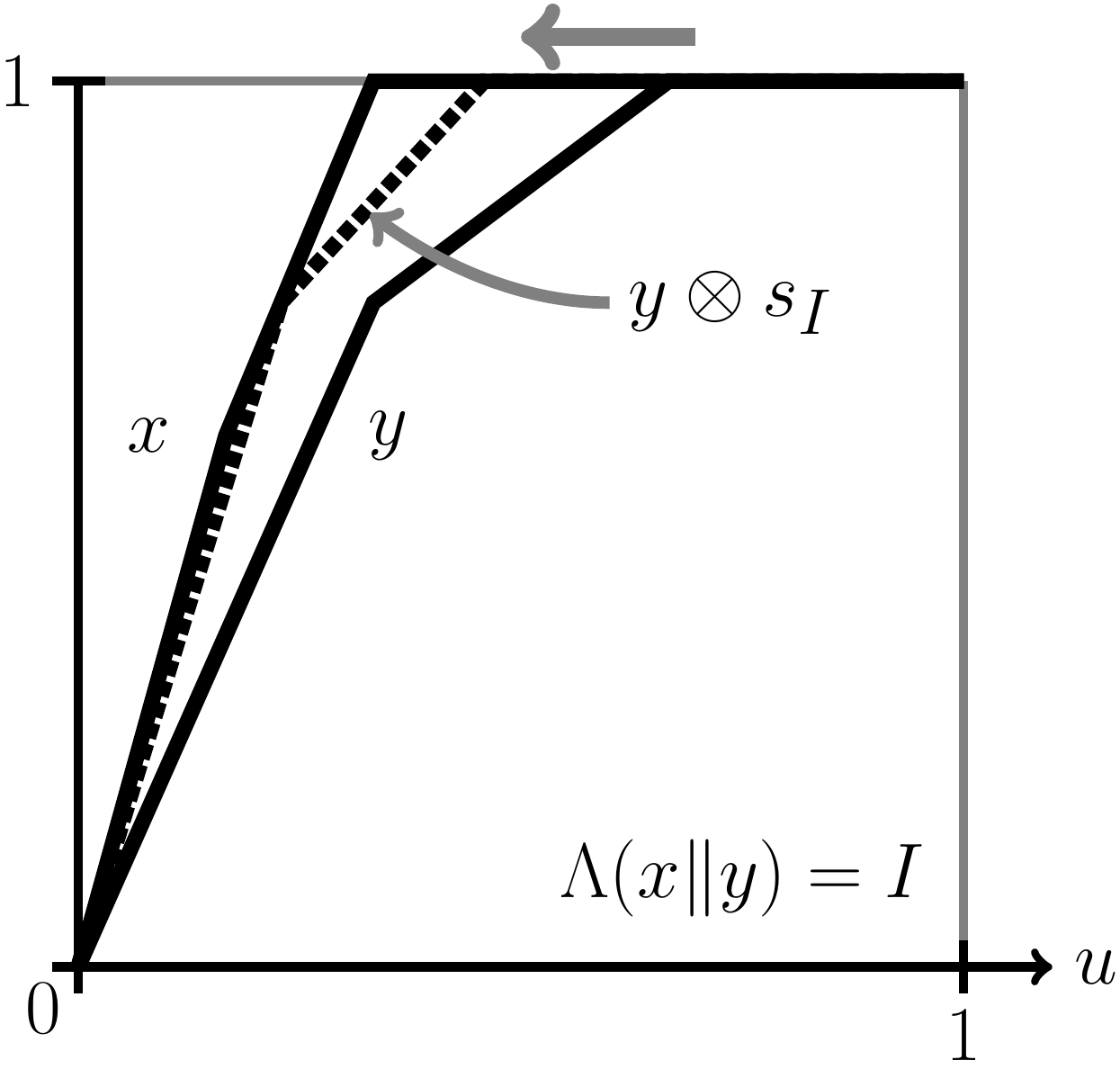}
\label{fig:s_R}
}
\subfloat[]{
\centering
\includegraphics[width=.3\textwidth, clip=true]{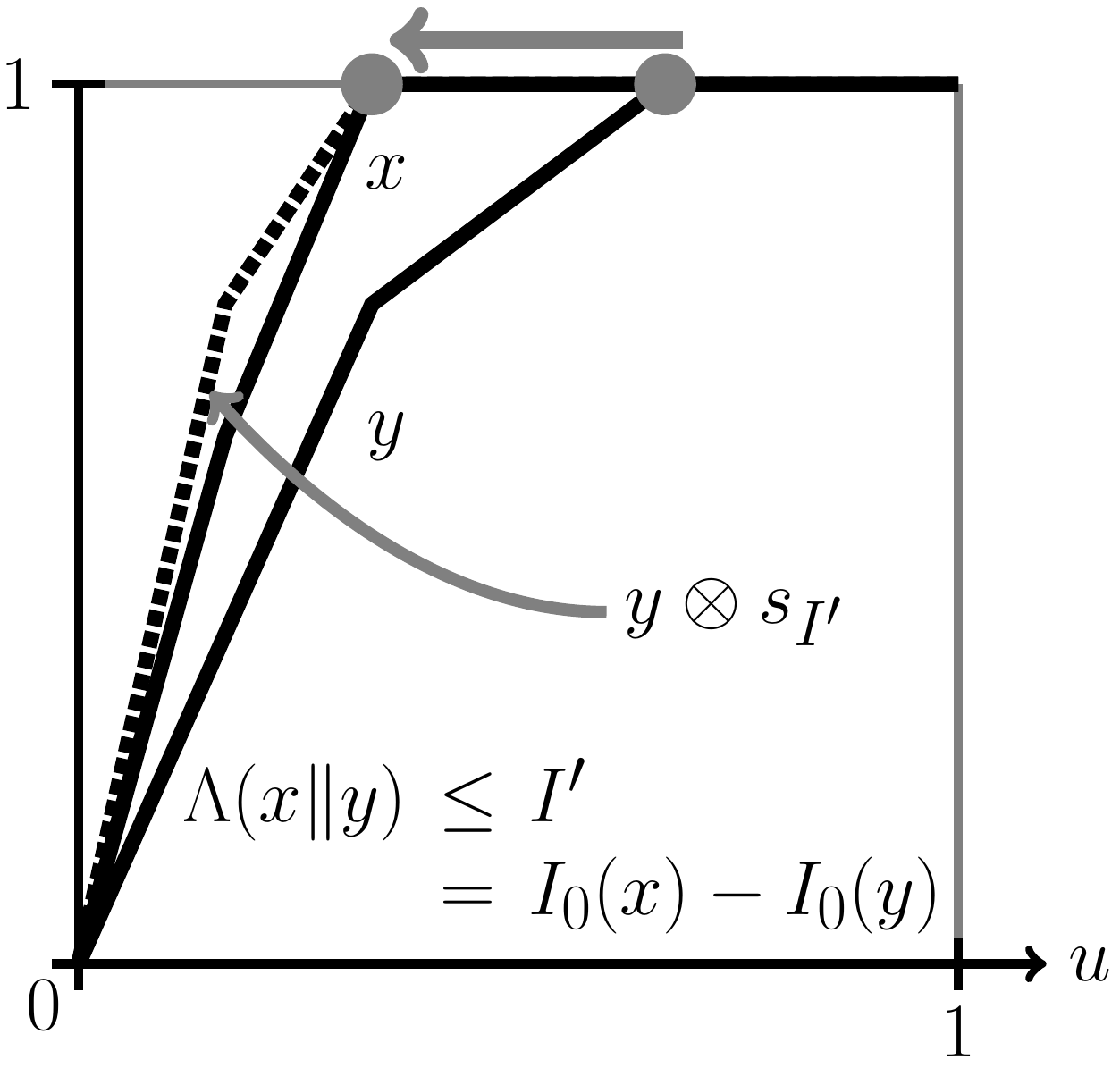}
\label{fig:s_R'}
}
\subfloat[]{
\centering
\includegraphics[width=.3\textwidth, clip=true]{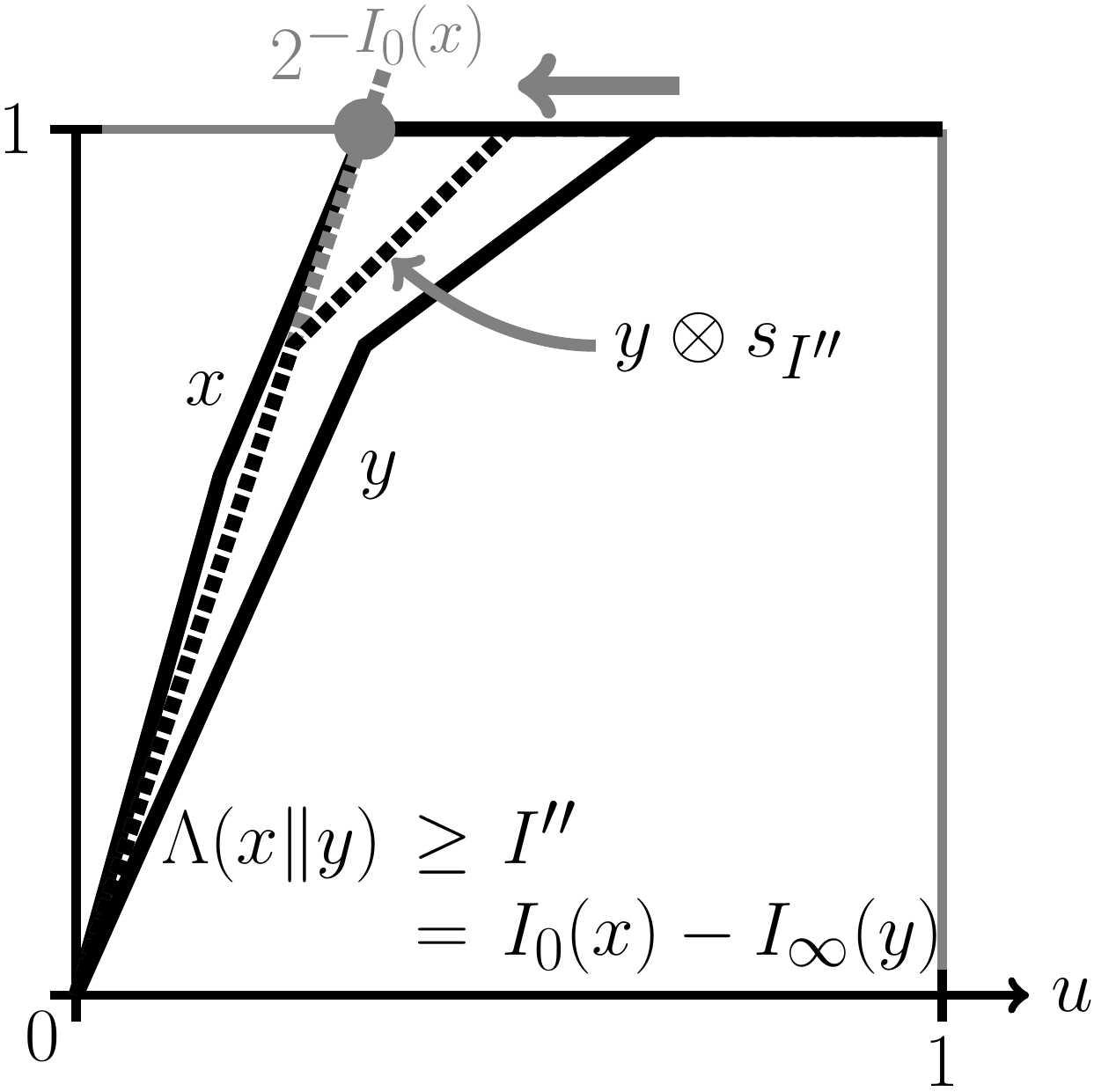}
\label{fig:s_R''}
}
\captionsetup{singlelinecheck=off,justification=raggedright}
\caption{An illustration of the nonuniformity yield of state conversion ($\Lambda(x\|y)$  when $\Lambda(x||y)\ge 0$) and bounds thereon.  The Lorenz curves of $x$ and $y$ are depicted, as well as
the Lorenz curve of $y \otimes s_I$ (dashed) for various values of $I$. According to~(\ref{eq:LorenzxotimessR}), the curve of $y\otimes s_I$ is obtained by linearly compressing
the curve of $y$ along the $u$-axis.
The nonuniformity yield is the largest value of $I$ such that $x \mapsto y \otimes s_I$, i.e.\ such that this linearly compressed curve is still below or on the curve of $x$.
This is illustrated in~(\ref{fig:s_R}): we compress the curve of $y$ until it touches the curve of $x$, and the resulting value of $I$ equals the nonuniformity yield
$I = \Lambda(x \| y)$.
In~(\ref{fig:s_R'}), it is shown how the upper bound $I':=I_0(x)-I_0(y)\geq \Lambda(x\|y)$ from Proposition~\ref{Proposition38}
can be derived graphically: compress the Lorenz curve of $y$ such that the leftmost point on its tail is brought to lie on top of the leftmost point on the tail of the Lorenz curve of $x$.
It is clear that the curve $L_y$ cannot be compressed
any further without crossing $L_x$, which gives on upper bound on $\Lambda(x\|y)$; in our example, the upper bound is strict, since $L_{y\otimes s_{I'}}$ already crosses $L_x$.
As a curve's tail is related to $I_0$ as shown in Fig.~\ref{fig:NUDistillableAndCost}, this graphical observation gives the desired upper bound by simple algebra.
By considering the amount of compression of $L_y$ that keeps the on-ramp slope of the resulting Lorenz curve less than that of $L_x$, we can obtain the 
upper bound $I_\infty(x)-I_\infty(y)$ (not shown).
Fig.~(\ref{fig:s_R''}) provides a graphical proof of the lower bound $I_0(x)-I_\infty(y)\leq\Lambda(x\|y)$.
The grey dotted line connects the origin with the leftmost point on the tail of $L_x$;
 it defines the on-ramp of a sharp state of nonuniformity $I_0(x)$ that can be distilled from $x$. Then $I_\infty(y)$ bits of nonuniformity
are spent to form $y$; graphically, the curve of $y$ is compressed until its on-ramp slope agrees with that of the sharp state.}
\label{fig:YieldOfConversion}
\end{figure*}

\begin{figure*}
\centering
\subfloat[]{
\centering
\includegraphics[width=.3\textwidth,clip=true]{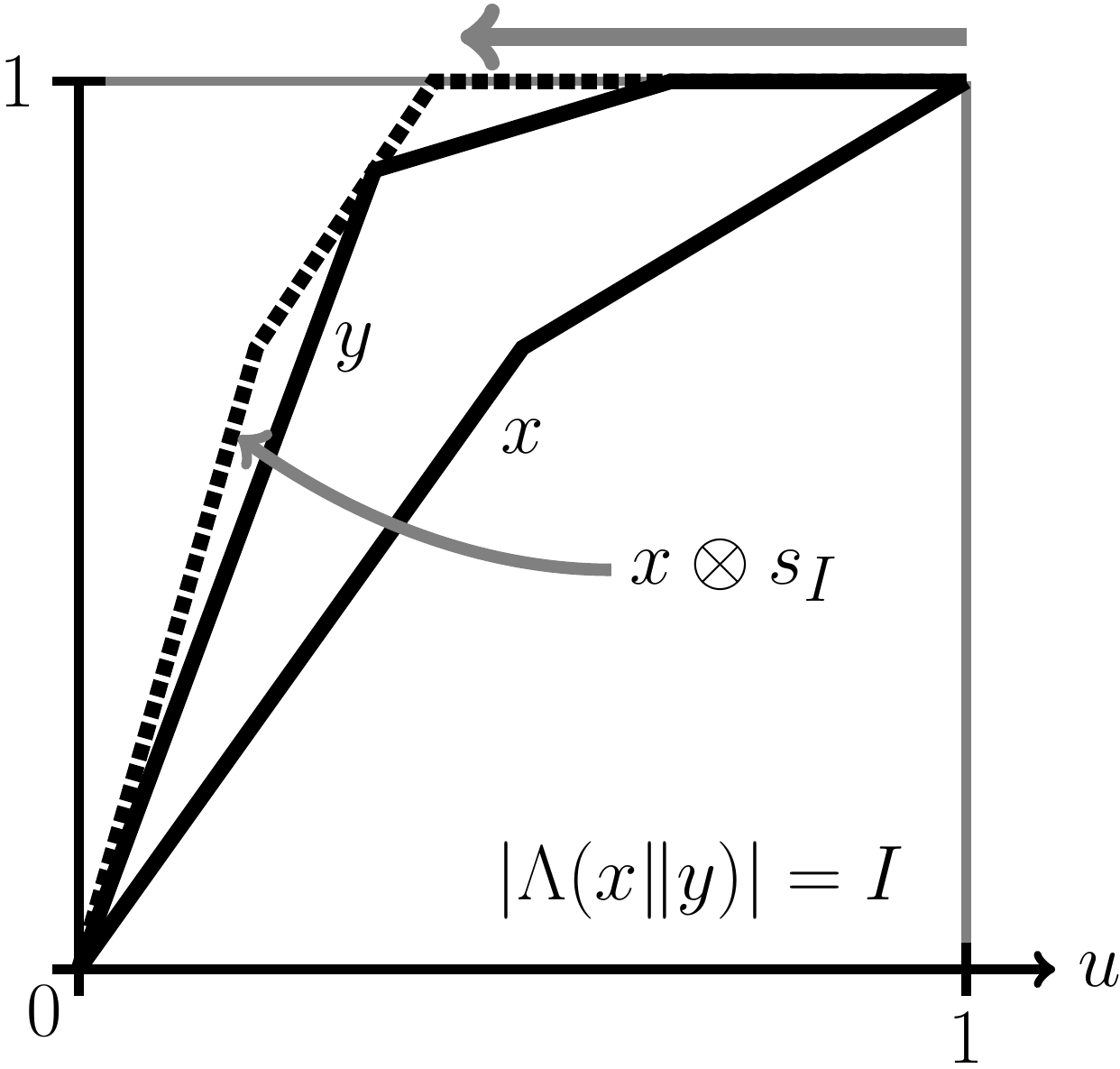}
\label{fig:Cost--s_I}
}
\subfloat[]{
\centering
\includegraphics[width=.3\textwidth,clip=true]{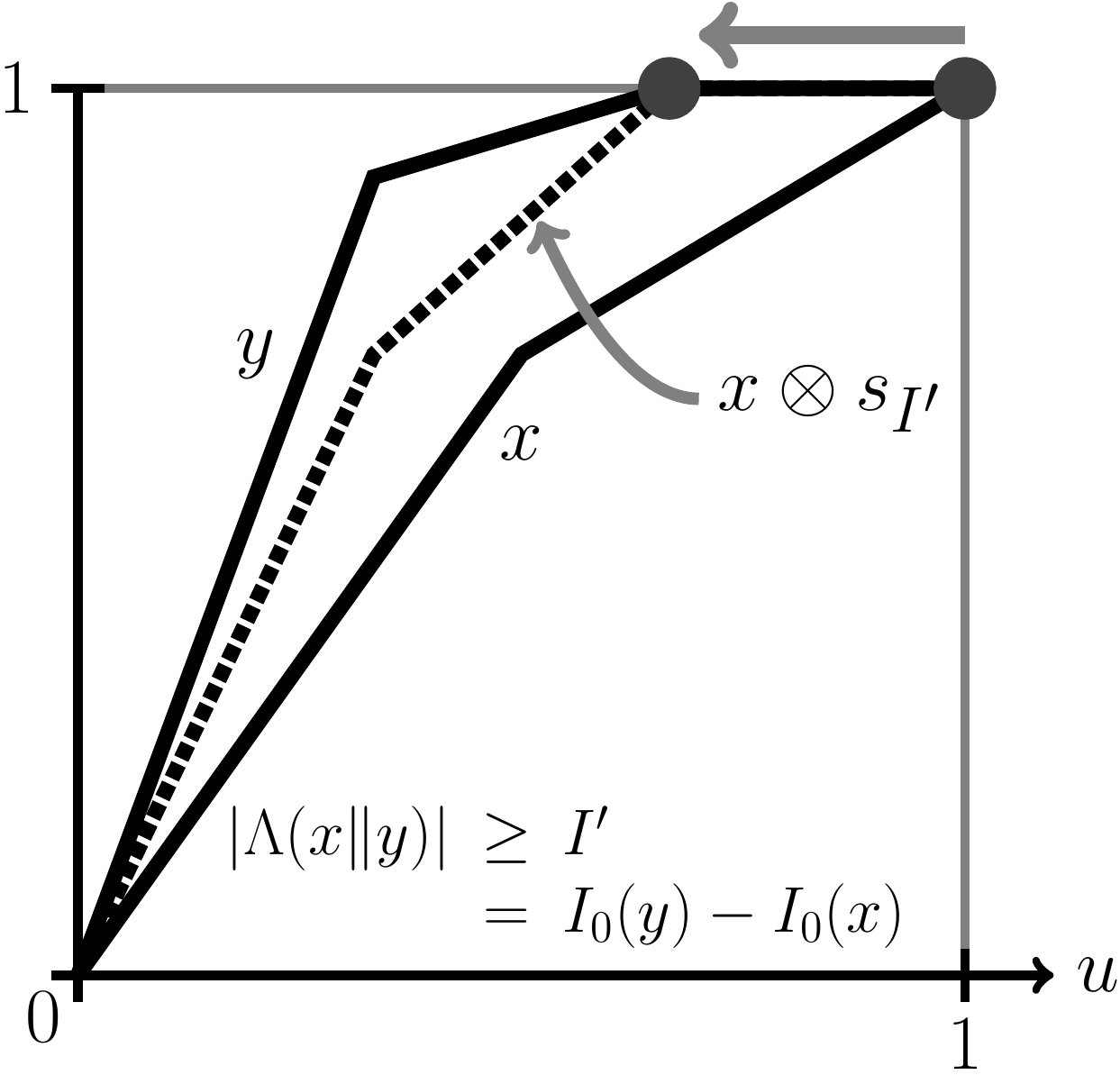}
\label{fig:Cost--s_I'}
}
\subfloat[]{
\centering
\includegraphics[width=.3\textwidth,clip=true]{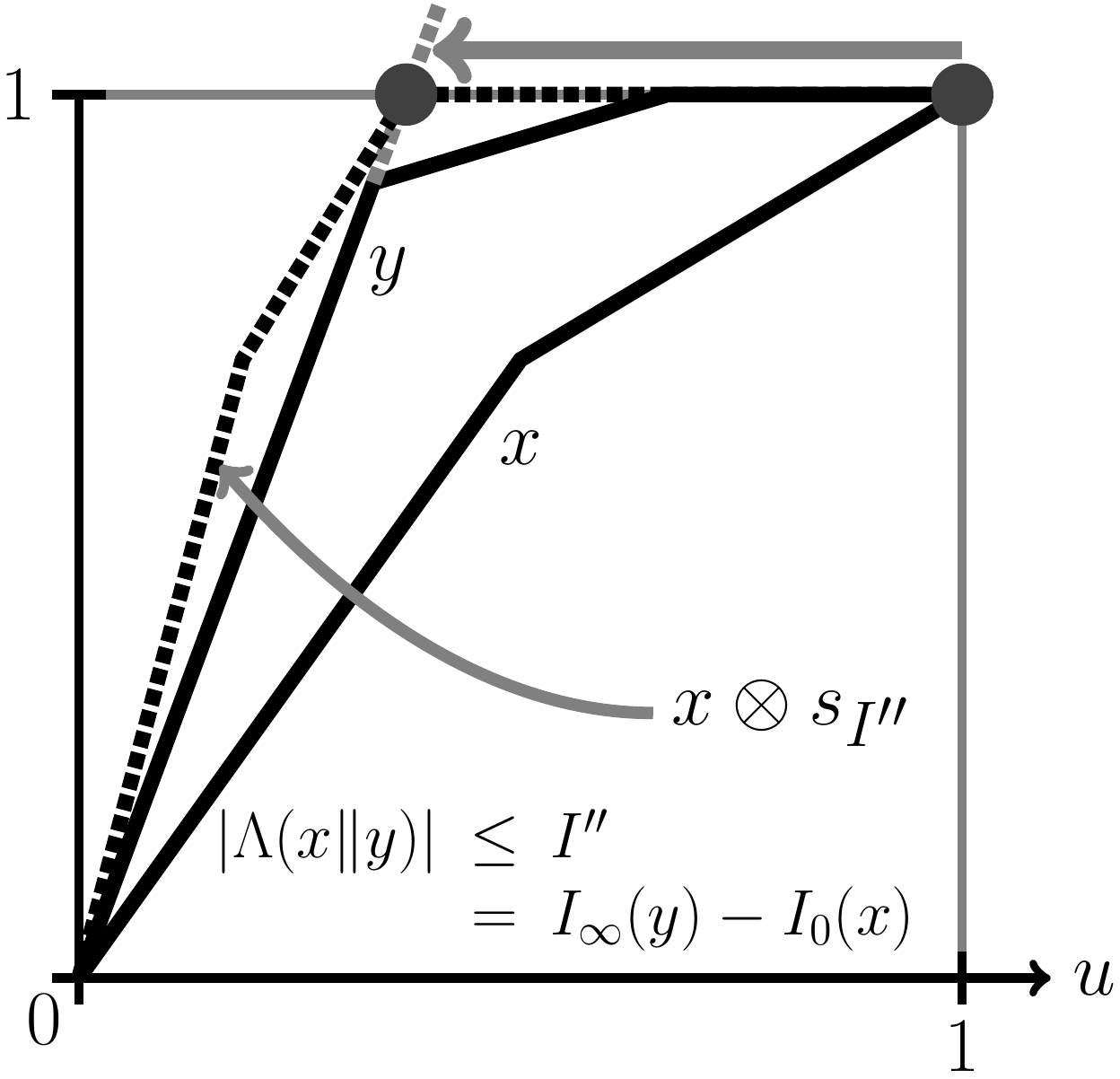}
\label{fig:Cost--s_I'''}
}
\captionsetup{singlelinecheck=off,justification=raggedright}
\caption{An illustration of the nonuniformity cost of state conversion ($|\Lambda(x||y)|$ when $\Lambda(x||y)<0$) and bounds thereon.
The Lorenz curves of $x$ and $y$ are depicted, as well as the Lorenz curve of $x \otimes s_I$ (dashed) for various values of $I$.
The nonuniformity cost is the least value of $I$ such that $x \otimes s_I \mapsto y$, which is graphically related to the least amount of compression in the $u$-direction
which is necessary such that the compressed curve of $x$ is on or above the curve of $y$. This is illustrated in Fig.~(\ref{fig:Cost--s_I}).  The value of $I$ corresponding to this least amount of compression defines the nonuniformity cost of converting $x$ to $y$, and the quantity $\Lambda(x \| y)$ is defined as the negative of this cost, $I = |\Lambda(x \| y)|$. 
Fig.~(\ref{fig:Cost--s_I'}) provides a graphical demonstration of the lower bound on the nonuniformity cost $|\Lambda(x \| y)| \ge I_0(y)- I_0(x)$.
The curve of $x$ has to be compressed by a factor at least equal to the one that makes the tail of the resulting Lorenz curve
as long as the tail of the Lorenz curve of $y$. In other words, the leftmost point of the tail of $L_x$ has to be brought to lie on top of the leftmost point of the tail of $L_y$ (these points are indicated by grey dots).  Considering the compression of $L_x$ required to make the on-ramp slope of the resulting Lorenz curve at least as great as that of $L_y$ yields the lower bound $I_\infty(y)-I_\infty(x)$ (not shown).
Fig.~(\ref{fig:Cost--s_I'''}) provides a graphical demonstration of the upper bound on the nonuniformity cost $|\Lambda(x \| y)| \le I_{\infty}(y)- I_0(x)$. The grey dotted line extends the on-ramp of $L_y$, and its intersection with the horizontal line of height one is indicated by a gray dot.  The leftmost point of the tail of $L_x$ is also indicated by a grey dot.  If the Lorenz curve of $x$ is compressed to such an extent that the second dot comes to lie on top of the first, then the new Lorenz curve $x\otimes s_I$ is clearly everywhere on or above the curve of $y$.}
\label{fig:CostOfConversion}
\end{figure*}

Note that the sign of $\lambda$ determines whether $\lambda$-noisy majorization is a strengthening or a weakening of noisy-majorization: if $x$ $\lambda$-noisy-majorize $y$ for a strictly positive $\lambda$, then it has \emph{more} nonuniformity than is required for it to noisy-majorize $y$, while if $x$ $\lambda$-noisy-majorizes $y$ for negative $\lambda$, then it has \emph{less} nonuniformity than is required for it to noisy-majorize $y$.

Finally, we require a definition before formalizing the answers to the questions that headed this section.

\begin{definition}\label{def:stateconversionwitness}
The \emph{maximum} factor $\lambda$ by which $x$ $\lambda$-noisy-majorizes $y$ is denoted $\Lambda(x\|y)$,
\begin{equation}
\Lambda(x\|y) \assign \max \left\{\lambda:  L_x(u) \ge L_y (2^\lambda u) \; \; \forall u \in ( 0, u_* ] \right\},
\label{eq:dualwitness2}
\end{equation}
where $u_* \assign \min \{1, 2^{-\lambda} \}$. 
\end{definition}

Given that $x \cconv y \otimes s_{I}$ if and only if $x$ ${\lambda}$-noisy-majorizes $y$ by a factor ${\lambda}=I$, it is clear that if the maximum such factor, $\Lambda(x\|y)$, is positive, then it quantifies the maximum nonuniformity of sharp state that can be distilled while achieving the state conversion $x \mapsto y$. 

Similarly, given that $x \otimes s_{I} \cconv y $ if and only if $x$ ${\lambda}$-noisy-majorizes $y$ by a factor ${\lambda}=-I$, it is clear that if the maximum such factor, $\Lambda(x\|y)$, is negative, then it quantifies the minimum nonuniformity of sharp state that is required to achieve the state conversion $x \mapsto y$. 

To summarize:

\begin{proposition}\label{lemma:distillationstateconversion}
If $\Lambda(x\| y) > 0$, then the state conversion $x \cconv y \otimes s_I$ is possible if and only if
$$I \le \Lambda(x\|y).$$
That is, one can distill a sharp state of nonuniformity at most $\Lambda(x\|y)$ in addition to achieving the conversion of $x$ to $y$.
\end{proposition}

\begin{proposition}\label{lemma:formationstateconversion}
If $\Lambda(x\| y) < 0$, then the state conversion $x \otimes s_I \cconv y $ is possible if and only if
$$I \ge -\Lambda(x\|y).$$
That is, in order to achieve the conversion of $x$ to $y$ it costs a sharp state of nonuniformity at least $|\Lambda( x\| y)|$.
\end{proposition}

In general, there is not a simple expression for $\Lambda(x\|y)$ because it depends on the details of the shape of the Lorenz curves of $x$ and $y$.  

Nonetheless, one can compute $\Lambda(x\|y)$ by finite means. For every height $h$ in the plot of the Lorenz curves of $x$ and $y$, one obtains a bound on $\Lambda( x\| y)$ in terms of the ratio $L_y^{-1}(h)/L_x^{-1}(h)$.  However, because $L_y(u)$ and $L_x(u)$ are concave, it suffices to compute these ratios only at the heights $h$ corresponding to the elbows of $L_x(u)$ and $L_y(u)$.  Therefore, there are at most $d_x+d_y-1$ such comparisons that need to be made. 
It follows that one can define $\Lambda(x\|y)$ by
 \begin{equation}\label{eq:Lambdafinitemeans}
\Lambda(x\|y) \assign \max \left\{\lambda:  L_x(u) \ge L_y (2^\lambda u) \; \; \forall u \in \mathcal{U} \right\},
\end{equation}
where $\mathcal{U} \assign \left( \{ \tfrac{1}{d_x},\tfrac{2}{d_x},\dots, 1\} \cup \{ \tfrac{1}{d_y},\tfrac{2}{d_y},\dots, 1\}\right)  \cap( 0, u_* ]$ with $u_* \assign \min \{1, 2^{-\lambda} \}$. 

In addition, even without solving the optimization problem, one can determine some nontrivial bounds on $\Lambda(x \| y)$.

\begin{proposition}
\label{Proposition38}
$\Lambda(x\|y)$ is bounded above and below as follows:
\begin{align}
I_0(x) - I_\infty(y) &\le \Lambda(x \| y), \\ 
\Lambda(x \| y) &\le \min \{ I_0(x)-I_0(y), I_\infty(x)-I_\infty(y) \} .
\end{align}
If $x$ is a sharp state then
$$\Lambda(x \| y)= I(x) -I_\infty(y).$$ 
If $y$ is a sharp state then
$$\Lambda(x\|y)=I_0(x) - I(y).$$
If both $x$ and $y$ are sharp states then 
$$\Lambda(x \| y)=I(x) - I(y).$$
\end{proposition}

To see the proof of these bounds, we consider the cases of $\Lambda(x\|y) \ge 0$ and $\Lambda(x\|y)<0$ separately, and use intuitions concerning nonuniformity cost and yield of state conversions.  See Figs.~\ref{fig:YieldOfConversion} and \ref{fig:CostOfConversion} for the associated Lorenz curves.

Consider the case $\Lambda(x\|y) \ge 0$ for instance.  The lower bound $\Lambda(x\|y)\geq I_0(x)-I_\infty(y)$ is evident from the fact that one can always achieve a yield which is at least the excess of the single-shot distillable nonuniformity of $x$, $I_0(x)$,  over the single-shot nonuniformity of formation of $y$, $I_{\infty}(y)$.   
We obtain the \emph{upper bounds} on the nonuniformity yield from the monotonicity and additivity of $I_\infty$ and of $I_0$. For instance, we have $I_0(x) \ge I_0(y \otimes s_I) = I_0(y)+ I$, which implies that $I\le I_0(x)-I_0(y)$.  Similar arguments hold for the case of $\Lambda(x\|y) < 0$.

Finally, recalling that for a sharp state $s$, $I(s) = I_0(s)=I_\infty(s)$, it follows that if either $x$ or $y$ or both are sharp states, the upper and lower bounds coincide and consequently we obtain an exact expression for $\Lambda$.

Note that if $y$ is a uniform state $m$ (which is a sharp state of nonuniformity 0), then  
$\Lambda(x\|y)=\Lambda(x\|m)=  I_0(x), $ which is positive, so the result predicts a nonuniformity yield of $I_0(x)$.
Given that one can prepare $m$ for free, the nonuniformity yield of this state conversion problem is just the nonuniformity that can be distilled from $x$.  This is indeed $I_0(x)$, as demonstrated in Lemma~\ref{lemma:distillation}. 

Similarly, if $x$ is a uniform state $m$, then $\Lambda(x\|y)=\Lambda(m\|y)=  -I_\infty(y)$, which is negative, so the result predicts a nonuniformity cost of  $I_\infty(y)$.  Given that $m$ is free, this state conversion problem is just the problem of forming the state $y$.
This reproduces the result of Lemma~\ref{lemma:formation} that the nonuniformity of formation is $I_\infty(y)$.

A state conversion process that has been of particular interest in the literature is that of erasure, which for a system of dimension $d$ is defined as the process which takes an arbitrary state $x$ to a pure state of dimension $d$.  Because such a pure state is a sharp state with nonuniformity $\log d$, we can immediately infer from our result that the nonuniformity cost of erasure of a state $x$ of a $d$-dimensional system is 
$$\log d - I_0(x) = H_0(x),$$
the R\'{e}nyi 0-entropy of $x$.   
Note that in the athermality theory, where the free states are thermal states at temperature $T$, the work cost of erasure is simply the nonuniformity cost multiplied by $kT \ln 2$.  And hence, the work cost of erasure of a state $x$ is $kT \ln 2 H_0(x)$. In particular, if $x$ is the uniform state of a binary variable, $m^{(2)}$, then the work cost is $kT \ln 2$. This is a form of Landauer's principle \cite{Landauer}.
Similarly, pure states can be used to do work in the athermality theory.   Bennett was perhaps the first to suggest that an initialized memory tape was capable of storing work \cite{BennettComp}.

If we imagine having a reservoir in a sharp state of arbitrarily large nonuniformity, then for any pair of states $x$ and $y$, one can ask: what is the minimum value of $\lambda_2 - \lambda_1$ such that $ x \otimes s_{\lambda_1} \cconv y \otimes s_{\lambda_2}$?  If this value is positive, there is a yield of nonuniformity, while if it is negative, there is a cost.  As it turns out, if this conversion is possible, then either it is possible with $\lambda_1=0$ or it is possible with $\lambda_2 = 0$.  The presence of the sharp state reservoir does not increase the yield or reduce the cost of any state conversion. This is a consequence of the uselessness of sharp states as catalysts, Proposition~\ref{prop:nocatwsharpstate} below.

For any given state $y$, if we consider $\Lambda(x\|y)$ as a function over states $x$, it is a nonuniformity monotone.
Similarly, for any given state $x$, if we consider $-\Lambda(x\|y)$ as a function over states $y$, it is a nonuniformity monotone.

The idea of quantifying ``how much'' one state majorizes another is from \cite{LawsOfThermo}. The maximum $\lambda$ for which $x$ $\lambda$-noisy-majorizes $y$, which we have denoted $\Lambda(x\|y)$, was there called the ``relative mixedness''.   The order relation that we have called $\lambda$-noisy-majorization was first introduced in \cite{QuantLandauer} as simply ``$\lambda$-majorization''. 
We use the term ``noisy-majorization'' because we are reserving the term ``majorization'' for the case where the states are of the same dimension, to accord with standard usage.  The identification of the nonuniformity yield of the state conversion $x \mapsto y$ as $\Lambda(x\|y)$ was also made in~\cite{QuantLandauer}.
This yield translates, in the athermality theory, to work extraction, featured in \cite{LawsOfThermo,SSP1,SSP2}, and to cooling, featured in \cite{janzing2000thermodynamic}.
That $\Lambda(x\|y)$ simplifies to the nonuniformity of formation when $x$ is a uniform state was noted in~\cite{LawsOfThermo}. That paper also presents the athermality cost of a particular erasure. The upper and lower bounds on the nonuniformity cost of state conversion, together with the exact expressions in the case where one of the states is sharp, have not previously been noted.

Finally, it is worth noting that these results 
provide necessary and sufficient conditions for a state conversion to be possible.
\begin{proposition}
The maximum factor $\lambda$ such that $x$ $\lambda$-noisy-majorizes $y$, denoted $\Lambda(x\|y)$, is a complete witness for the state conversion $x \cconv y$, that is,  $\Lambda(x\|y)\ge 0$ if and only if $x \cconv y$.
\end{proposition}

\begin{proof}
By Proposition~\ref{lemma:distillationstateconversion}, $\Lambda(x\| y)\ge 0$ implies that $x \cconv y \otimes s_I$ with $I\ge 0$ and because marginalizing over $s_I$ is a free operation, this implies that $x \cconv y$.  Similarly, by Proposition~\ref{lemma:formationstateconversion}, $\Lambda(x\| y) <  0$ implies that $x \not \mapsto y$ by noisy operations. 
\end{proof}
 
This criterion was first proposed as a necessary and sufficient condition for state conversion in \cite{LawsOfThermo}.  In fact, they provided the generalization of this condition for the resource theory of athermality. The criterion was rederived in \cite{QuantLandauer}.

\subsection{Catalysis}\label{Sec:Catalysis}
 
\begin{definition}
In any given resource theory, if $x \not \mapsto y$ under the free operations, but there exists a state $z$ such that $x \otimes z \mapsto y \otimes z$, then $z$ is said to be a \emph{catalyst} for the conversion of $x$ to $y$.
\end{definition}
  In the resource theory of nonuniformity, there are nontrivial examples of catalysis.   This is true both for state conversions between states of equal dimension and between states of unequal dimension.  In the latter case, we have the following example, adapted from \cite{JonathanPlenio}.
The two states are
\begin{align}
x &= (0.5,0.25,0.25,0), \\
y &= (0.8, 0.2),
\end{align}
and the catalyst is
\begin{equation}
z = (0.6, 0.4).
\end{equation}
Here, $x \not \mapsto y$ by noisy operations, but $x \otimes z \cconv y \otimes z$.  
Specifically, there is a permutation taking $x \otimes z$ to $y \otimes m^{(2)} \otimes z$ where $m^{(2)}$ is the uniform state on a bit.  One then marginalizes over $m^{(2)}$ to obtain $y \otimes z$.

Although we have seen that Lorenz curves are extremely useful tools for many notions of state conversion under noisy operations, they are of limited usefulness in understanding catalysis.  Nonetheless, they do make evident one simple result:
\begin{proposition}\label{prop:nocatwsharpstate}
A sharp state is useless as a catalyst, that is, for all finite $I\in \bbR_+$, $x \otimes s_I \cconv y \otimes s_I$ if and only if $x \cconv y$.
\footnote{This result is the analogue of the fact that maximally entangled states are useless as catalysts in entanglement theory~\cite{JonathanPlenio}.}
\end{proposition}
To see that this is the case, recall from Lemma~\ref{lemma:adjoinsharp} that the operation of joining a sharp state with nonuniformity $I$ to a state $x$, $x \mapsto x \otimes s_I$, is represented in terms of Lorenz curves as a $2^I$-fold compression along the $u$-axis towards the origin, and appending a tail of the appropriate length. As such, if the Lorenz curve of $x$ does not lie everywhere above or on the Lorenz curve of $y$, then the same relationship will hold between the Lorenz curve of $x \otimes s_I$ and that of $y \otimes s_I$. 

As an aside, this result has an interesting consequence for the role of ideal measurements in the resource theory of nonuniformity.  It is clear that the capacity to implement an ideal measurement requires a resource of nonuniformity because it requires the pointer to be initialized in a sharp state.  Nonetheless, one might wonder whether such an ideal measurement device might sometimes be used in order to make possible a state conversion that would otherwise be impossible, \emph{without} consuming any of the nonuniformity in the device.  In short, one might wonder whether an ideal measurement device might sometimes act as a catalyst.  However, insofar as such a device is modelled by a sharp state, Proposition~\ref{prop:nocatwsharpstate} implies that it is useless as a catalyst.  Whatever state conversion can be done with a measurement device that is used catalytically can also be achieved without it.

Proposition~\ref{prop:nocatwsharpstate} also implies that removing zero components from a state (or from another state that is noisy-equivalent to it) cannot change its effectiveness as a catalyst.  To express this claim compactly, it is useful to introduce a notation for removing zero components. 
\begin{definition}
\label{DefClipping}
For any state $x$ and $I \le I_0(x)$, define $x^{\langle I \rangle}$ to be the smallest-dimensional state such that  $x^{\langle I \rangle} \otimes s_I$ is noisy-equivalent to $x$.
\end{definition}
Clearly, the Lorenz curve of $x^{\langle I \rangle}$ is that of $x$ stretched by a factor of $2^I$ only the $u$-axis, thereby shortening its tail length.  If $I=I_0(x)$, then the Lorenz curve of $x^{\langle I \rangle}$ is that of $x$ stretched to the point where its tail length goes to zero.  
The map $x \mapsto x^{\langle I \rangle}$ will be referred to as a \emph{truncation} of $x$.
\begin{corollary}\label{truncatecatalyst}
For any $I \le I_0(z)$, $z$ catalyzes the conversion of $x$ to $y$ under noisy operations if and only if $z^{\langle I \rangle}$ does as well, that is, $x \otimes z \cconv y \otimes z$ if and only if $x \otimes z^{\langle I \rangle} \cconv y \otimes z^{\langle I \rangle}$.
\end{corollary}

To see that this is the case, it suffices to note that $z$ is noisy-equivalent to $z^{\langle I \rangle} \otimes s_I$, so that $x \otimes z \cconv y \otimes z$ if and only if $x \otimes z^{\langle I \rangle} \otimes s_I \cconv y \otimes z^{\langle I \rangle} \otimes s_I$, and then to apply Proposition~\ref{prop:nocatwsharpstate}. It follows that states which act as nontrivial catalysts do not need to contain any zeros.

\subsubsection{Quasi-order of states under noisy operations assisted by a catalyst}

A quasi-order over states that is induced by a set of free operations assisted by a catalyst has been called a \emph{trumping} quasi-order~\cite{NielsenIntroduction}.  We will refer to the trumping quasi-order induced by noisy operations as the \emph{noisy-trumping} quasi-order.

At first glance, one might expect that the definition of noisy-trumping ought to be that $x$ noisy-trumps $y$ if there exists a catalyst $z$ such that $x \otimes z \cconv y \otimes z$.
However, we will see that this definition leads to some 
mathematical complications that are not physically significant, and we will consequently be led to consider modifications of this definition wherein either the input or the output state is allowed to be arbitrarily well approximated by $x$ or $y$ respectively. 
This move is analogous to how we chose to define noisy operations (Definition \ref{DefNoisyQOperations}) as those that can be arbitrarily well approximated by maps of the form~(\ref{eqNoisyQuantum}) rather than those that are precisely of that form.

We start by recapitulating the part of the results of~\cite{klimesh} and~\cite{Turgut} which is the basis of all that follows.
\begin{lemma}[\textcite{klimesh}]
\label{LemKlimesh}
Let $x$ and $y$ be states of the same dimensionality (i.e. $d_x=d_y=:d$) that do not both contain zero components
and that satisfy $x^\downarrow\neq y^\downarrow$. Then there is a catalyst $z$ such that $x\otimes z \cconv y\otimes z$ if and only if
$f_r(x)>f_r(y)$ for all $r\in\mathbb{R}$, where
\[
   f_r(x)=\left\{
      \begin{array}{cl}
         \ln \sum_{i=1}^d x_i^r & \mbox{if }r>1,\\
         \sum_{i=1}^d x_i \ln x_i & \mbox{if }r=1,\\
         -\ln \sum_{i=1}^d x_i^r & \mbox{if }0<r<1,\\
         -\sum_{i=1}^d \ln x_i & \mbox{if }r=0,\\
         \ln \sum_{i=1}^d x_i^r & \mbox{if }r<0.
      \end{array}
   \right.
\]
\end{lemma}

This result resembles our notion of a \emph{complete set of monotones}, introduced in Subsection~\ref{SubsubSchurConvexity} as describing a set
of monotones that completely characterize state transformations: the lemma states that the functions $f_r$ \emph{almost} have this property -- but not
exactly. Namely, while $f_r(x)>f_r(y)$ for all $r$ implies that $x$ can be catalytically converted to $y$, we cannot infer anything in the case where
$f_r(x)\geq f_r(y)$ (and equality, for example, is attained at only one single value of $r$).

Not only is this at odds with the definition of a complete set of monotones; it also indicates a certain unphysical kind of discontinuity.
Imagine some fixed state $x$, and a sequence of states $y_n$, all of the same dimensionality, which converge to $y$ in the limit of large $n$, i.e.\ $\lim_{n\to\infty} y_n=y$.
It may well be the case that $f_r(x)>f_r(y_n)$ for all $r\in\mathbb{R}$ and $n\in\mathbb{N}$, but that only the non-strict inequality survives the limit,
i.e.\ $f_r(x)\geq f_r(y)$ with actual equality for some $r$. In this case, we could convert $x$ catalytically into \emph{any} of the states $y_n$, but \emph{not}
into the state $y$.

The physical interpretation would be odd: it would tell us that we cannot produce $y$ perfectly (catalytically from $x$), but we can produce
it to arbitrary accuracy. However, in actual physical situations, we can \emph{never} expect to produce any state perfectly; conversion to arbitrary
accuracy is the best we can hope for. Therefore, we would like to have a definition of trumping that takes this physical limitation into account, and that
includes situations where one can obtain $y$ from $x$ catalytically to arbitrary accuracy, even if one cannot obtain it perfectly.

The most straightforward way to obtain a definition like this would be to say that
 $x$ noisy-trumps $y$ if and only if $x$ can be catalytically converted
to an arbitrarily good approximation of $y$ by a suitable catalyst. In Appendix~\ref{Sec:appnoisytrumping}, we will analyze this definition, recovering and elaborating some results from~\cite{1ShotAtherm2}. However, we will see that this definition still contains a discontinuity that has no physical counterpart, namely, that input states which contain zero components behave differently from input states that have arbitrarily small components, even though the two could never be physically distinguished.

To achieve a physically sensible definition of noisy trumping, we must instead allow that the input state required to form $y$ catalytically be merely arbitrarily close to the noisy equivalence class of $x$, rather than a perfect copy of $x$.  
In order to avoid the mathematical subtleties of defining a metric over a set of states of differing dimensions, we will adopt a specific type of approximation to the noisy equivalence class of $x$, namely, the state that is the composite of $x$ and a sharp state that is arbitrarily close to the uniform state.  This sort of definition was also suggested in~\cite{1ShotAtherm2} where it was described as allowing some nonuniformity to be consumed in the process, as long as the amount of nonuniformity can be made arbitrarily small.

\begin{definition}[Noisy trumping]
\label{DefNoisyTrumping}
We say that $x$ noisy-trumps $y$ if for any $\delta>0$, no matter how small,
 there is a catalyst $z$ such that
\[
   x\otimes s_\delta \otimes z \cconv y\otimes z.
\]
\end{definition}
It turns out that the ability to consume an arbitrarily small amount of nonuniformity unlocks potential state transformations that are otherwise impossible.

Clearly, if $x\otimes z\cconv y\otimes z$ for some catalyst $z$, then $x$ noisy-trumps $y$ 
 (because for the input state $x\otimes s_\delta \otimes z$, one can simply marginalize over $s_\delta$ and then implement $x\otimes z\cconv y\otimes z$). 
However, the converse is not true: it is impossible to obtain $x\otimes z\cconv y\otimes z$ as a formal ``limit $\delta\to 0$'' of the definition of noisy trumping.
This will be established below, by showing that the conditions for noisy-trumping are different from the conditions in Lemma~\ref{LemKlimesh}.

It is no restriction to demand that the target state $y$ be produced perfectly: as we show in Corollary~\ref{CorTrumpingClosedness} below, it is an automatic consequence of Definition~\ref{DefNoisyTrumping} that the ability to produce arbitrarily good approximations of $y$ implies the ability to produce $y$ perfectly.

Note also that Definition~\ref{DefNoisyTrumping} requires that the final state of the catalyst be \emph{precisely} equal to its initial state.  This implies that the catalyst can be reused arbitrarily many times and consequently that any nonuniformity required to form $y$ must have been drawn from $x$, or a state arbitrarily close to it, and not from the catalyst.\footnote{Note that this verdict on the reusability of the catalyst would not change even if we interpreted the sharp state $s_{\delta}$ as being a modification of the initial state of the catalyst rather than a modification of the initial state of $x$.  For any given number $N$ of reuses of the catalyst, one could still take $N \delta$ to be arbitrarily small. }

Finally, note that in this article, we refer to state conversions between arbitrarily good approximations of states as \emph{exact}.  A state conversion is only described as \emph{approximate} if the degree of approximation is finite. Section~\ref{Sec:Approximatestateconversion} considers approximate state conversion.

The results by Klimesh and Turgut, Lemma~\ref{LemKlimesh}, imply the following characterization of noisy-trumping:

\begin{lemma}[Conditions for noisy-trumping]\label{lemma:trumping}
Let $x$ and $y$ be any pair of states. Then $x$ noisy-trumps $y$ if and only if
\begin{equation}
   I_p(x)\geq I_p(y)\mbox{ for all }p\geq 0.
   \label{eqTrumpingIp}
\end{equation}
In other words, the set of R\'enyi $p$-nonuniformities for nonnegative $p$, $\{I_p: p \geq 0 \}$ defined in Table~\ref{table:NUmonotones} (the $p=1$ case corresponds to the Shannon nonuniformity), is a complete set of monotones for the noisy trumping relation.
\end{lemma}

\begin{proof}
Necessity of these conditions for noisy-trumping is straightforward to prove. Applying the fact that every $I_p$ is a nonuniformity monotone (cf.\ Subsection~\ref{SubsecRenyip}) to
Definition~\ref{DefNoisyTrumping}, we obtain
\[
   I_p(x\otimes s_\delta\otimes z)\geq I_p(y\otimes z)
\]
for every $\delta>0$, and using additivity of $I_p$ as well as $I_p(s_\delta)=\delta$, this becomes
\[
   I_p(x)+\delta\geq I_p(y)\qquad\mbox{for all }\delta>0,
\]
which implies that $I_p(x)\geq I_p(y)$.

Sufficiency of the conditions for noisy-trumping are much more difficult to prove. We will do so by appealing to the results of~\cite{klimesh},
summarized in Lemma~\ref{LemKlimesh}.

We begin by noting that if the states under consideration are of unequal dimensions, $d_x \ne d_y$, then we can simply replace $x$ by $x\otimes m^{(d_y)}$ and replace $y$ by $y\otimes m^{(d_x)}$, because the replacement does not alter the values of any of the $I_p$'s, the property of containing zeros or not, or the noisy-trumping property.
Thus, we may assume that $d_x=d_y$.

So suppose that $I_p(x)\geq I_p(y)$ for all $p\geq 0$. Fix any $\delta>0$ which is the logarithm of a rational number. We start by removing common zeros from $x$ and $y$ in the sense
of Definition~\ref{DefClipping}. Set $y':=y^{\langle I_0(y)\rangle}$. 
Since $I_0(x)\geq I_0(y)$, the state $x':=x^{\langle I_0(y)\rangle}$ is well-defined.
 We will be considering a state conversion with initial state $x'\otimes s_\delta$ and a final state in the noisy equivalence class of $y'$.  To apply Lemma~\ref{LemKlimesh}, we need to ensure that the dimensions of the initial and final states are equal, so we take the final state to be $y'\otimes m^{(d_{s_\delta})}$ where $d_{s_\delta}$ is the dimension of the space on which $s_\delta$ is defined.
 We now show that $ I_p(x'\otimes s_\delta) > I_p(y'\otimes m^{(d_{s_\delta})})$ for all $p \geq 0$.
 
By definition of $x'$, $x'\otimes s_{I_0(y)}$ is noisy-equivalent to $x$. Recalling the additivity of the $I_p$'s and the fact that $I_p(s_I)=I$, we have that $I_p(x')+I_0(y)=I_p(x).$  For similar reasons, $I_p(y') + I_0(y)= I_p(y)$.  This implies that
\begin{eqnarray*}
   I_p(x'\otimes s_\delta)&=&I_p(x)-I_0(y)+\delta \\
   &>& I_p(y)-I_0(y) = I_p(y'\otimes m^{(d_{s_\delta})})
\end{eqnarray*}
for all $p>0$. We note that the monotonicity of the family of functions $\{ f_p \,\,|\,\, p>0 \}$ is equivalent to the monotonicity of the family $\{ I_p \,\,|\,\, p>0 \}$ of
R\'{e}nyi $p$-nonuniformities. Therefore, we obtain $f_p(x'\otimes s_\delta)>f_p(y'\otimes m^{(d_{s_\delta})})$ for all $p>0$.
Since $x'\otimes s_\delta$ contains zeros, but $y'\otimes m^{(d_{s_\delta})}$ does not,
it follows that for all $p \leq 0$ we have $\infty=f_p(x'\otimes s_\delta)>f_p(y'\otimes m^{(s_\delta)})$. Thus, Lemma~\ref{LemKlimesh} shows that there is a catalyst $z$ such that
\[
   x'\otimes s_\delta\otimes z \cconv y'\otimes m^{(d_{s_\delta})}\otimes z.
\]
By adding on a sharp state $s_{I_0(y)}$ on both sides (which is not touched during the state conversion), and using again the noisy-equivalence
of $x'\otimes s_{I_0(y)}$ and $x$ (and similarly for $y$), and the fact that one can marginalize over the uniform state,
we see that $x\otimes s_\delta\otimes z \cconv y\otimes z$. Since $\delta>0$ was chosen arbitrarily, this proves that $x$ noisy-trumps $y$.
\end{proof}

An immediate consequence of Lemma~\ref{lemma:trumping} is a certain continuity (or closedness) property which is desirable from the point of view
of physical interpretation.
\begin{corollary}
\label{CorTrumpingClosedness}
Suppose that $x$ is any state, and $(y_n)_{n\in\mathbb{N}}$ is any sequence of states of common dimensionality which converges to some state $y$, i.e.\ $\lim_{n\to\infty} y_n=y$.
If $x$ noisy-trumps every $y_n$ then $x$ noisy-trumps $y$.
\end{corollary}
\begin{proof}
According to Lemma~\ref{lemma:trumping}, if $x$ noisy-trumps every $y_n$ then $I_p(x)\geq I_p(y_n)$ for every $n\in\mathbb{N}$ and $p\geq 0$.
Since every $I_p$ is continuous if $p>0$, it follows that $I_p(x)\geq I_p(y)$ for $p>0$. Even though $I_0$ is not continuous, we can express it as $I_0=\lim_{p\searrow 0} I_p$,
and so we also obtain $I_0(x)\geq I_0(y)$ by taking the limit $p\searrow 0$. Invoking Lemma~\ref{lemma:trumping} again shows that $x$ noisy-trumps $y$.
\end{proof}

Lemma \ref{lemma:trumping} can also be expressed in terms of the existence of a complete witness for noisy-trumping, as follows.
\begin{definition}
\label{DefLambda}
For an arbitrary pair of states $x$ and $y$, define
\[
   \Lambda_{\rm cat}(x\|y) := \inf_{p\geq 0}\left(\strut I_p(x)-I_p(y) \right).
\]
\end{definition}
\begin{proposition}
\label{PropTrumpingLambda}
For an arbitrary pair of states $x$ and $y$, $x$ noisy-trumps $y$ if and only if
\begin{equation}\label{eq:LamCat}
  \Lambda_{\textnormal{cat}}(x\| y)\ge0.
\end{equation}
Equivalently, the function $\Lambda_{\textnormal{cat}}(x\| y)$ is a complete witness for noisy-trumping.
\end{proposition}

The necessary and sufficient conditions for noisy-trumping are a strict subset of those for noisy-majorization because noisy-majorization implies noisy-trumping but not vice-versa.  It follows that one must add conditions to those of Lemma~\ref{lemma:trumping} in order to characterize noisy-majorization.  In other words, R\'{e}nyi $p$-nonuniformities alone cannot decide whether a given noncatalytic state conversion is possible. This is true even if we take the R\'enyi $p$-nonuniformities of orders $p<0$ into account.  The reason is as follows.
Every R\'{e}nyi $p$-nonuniformity is additive and therefore is nonincreasing under $x \cconv y$ if and only if it is nonincreasing under $x \otimes z \cconv y \otimes z$ for any state $z$.  
To do justice to the distinction between noisy-majorization and noisy-trumping, therefore, one must consider one or more monotones that are \emph{not} additive, hence one must go beyond the R\'{e}nyi nonuniformities. 

\subsubsection{Nonuniformity of formation and distillable nonuniformity in the presence of a catalyst}

The nonuniformity of formation and the distillable nonuniformity (as determined in Proposition~\ref{lemma:distillation}
and~\ref{lemma:formation}) are not changed by having access to a catalyst.  

\begin{corollary}\label{corollary:formationcatalysis}
$x$ noisy-trumps $s_I$ if and only if
$$I \le I_0(x),$$
that is, the maximum nonuniformity of sharp state that can be distilled from $x$ under noisy operations assisted by a catalyst is $I_0(x)$.
\end{corollary}
\begin{proof}
Sufficiency is trivial because if $I \le I_0(x)$, then by the result on distillable nonuniformity (Proposition~\ref{lemma:distillation}) it follows that $x \cconv s_I$, so that one does not even require a catalyst to achieve the conversion. 
It suffices, therefore, to prove necessity. But this follows from Lemma~\ref{lemma:trumping}, which shows that $I_0$ is a noisy-trumping monotone, hence
$I_0(x)\geq I_0(s_I)=I$.
\end{proof}

\begin{corollary}\label{corollary:distilationcatalysis}
$s_I$ noisy-trumps $x$ if and only if
$$I \ge I_{\infty}(x),$$
that is, the minimum nonuniformity of sharp state that is required to form $x$ under noisy operations assisted by a catalyst is $I_{\infty}(x)$.
\end{corollary}
\begin{proof}
Again, sufficiency is straightforward to establish because if  $I\geq I_\infty(x)$ then by the result on the nonuniformity of formation (Proposition~\ref{lemma:formation}) we can infer that  $s_I \cconv x$, so that the conversion can be achieved even without a catalyst.  To establish necessity, we use again Lemma~\ref{lemma:trumping}: the limit
$p\to\infty$ shows that $I_\infty$ is a noisy-trumping monotone, hence $I=I_\infty(s_I)\geq I_\infty(x)$.
\end{proof}

\subsubsection{Nonuniformity cost and yield of catalytic state conversion}

In catalytic state conversion, we also have a result akin to Propositions~\ref{lemma:distillationstateconversion} and \ref{lemma:formationstateconversion}.  This was described in the context of the theory of athermality in ~\cite{1ShotAtherm2}.

\begin{proposition}\label{lemma:distillationcatstateconversion}
If $\Lambda_{\rm cat}(x\|y) > 0$, then $x$ noisy-trumps $y \otimes s_I$ if and only if
$$I \le \Lambda_{\textnormal{cat}}(x\|y).$$
That is, if $\Lambda_{\rm cat}(x\|y)$ is positive, then
using a catalyst one can distill a sharp state of nonuniformity at most $\Lambda_{\textnormal{cat}}(x\|y)$ in addition to achieving the conversion of $x$ to $y$.
\end{proposition}
\begin{proof}
We note that
\begin{eqnarray*}
   \Lambda_{\rm cat}(x\|y\otimes s_I) &=&\inf_{p\geq 0}\left(\strut I_p(x)-I_p(y\otimes s_I)\right)\\
   &=&\inf_{p\geq 0}\left(\strut I_p(x)-I_p(y)\right)-I \\
   &=&\Lambda_{\rm cat}(x\|y)-I,
\end{eqnarray*}
where we have made use of the additivity of $I_p$ and the fact that $I_p(s_I) = I$ for $p \geq 0$.  According to Proposition~\ref{PropTrumpingLambda}, $x$ noisy-trumps $y\otimes s_I$ if and only if $\Lambda_{\rm cat}(x\|y\otimes s_I)$ is non-negative, hence
if any only if $\Lambda_{\rm cat}(x\|y)\geq I$.
\end{proof}

\begin{proposition}\label{lemma:formationcatstateconversion}
If $\Lambda_{\rm cat}(x\|y) < 0$, then
$x \otimes s_I$ noisy-trumps $y$ if and only if
$$I \ge -\Lambda_{\textnormal{cat}}(x\|y).$$
That is, if $\Lambda_{\rm cat}(x\|y)$ is negative, then in order to achieve the conversion of $x$ to $y$ with the aid of a catalyst, it costs a sharp state of nonuniformity
at least $|\Lambda_{\textnormal{cat}}( x\| y)|$.
\end{proposition}

\begin{proof}
We have
\begin{eqnarray*}
   \Lambda_{\rm cat}(x\otimes s_I\| y) &=& \inf_{p\geq 0}\left(\strut I_p(x\otimes s_I)-I_p(y)\right) \\
   &=& I+\inf_{p\geq 0}\left(\strut I_p(x)-I_p(y)\right) \\
   &=& I+\Lambda_{\rm cat}(x\|y).
\end{eqnarray*}
According to Proposition~\ref{PropTrumpingLambda}, $x\otimes s_I$ noisy-trumps $y$ if and only if $\Lambda_{\rm cat}(x\otimes s_I\| y)$ is non-negative, hence if and only if $I \geq -\Lambda_{\rm cat}(x\|y)$.
\end{proof}

We again recover the results about formation and distilation as special cases:  the distillable nonuniformity  of $x$ under catalyzed noisy operations is simply the nonuniformity yield of the conversion of $x$ to the uniform state $m$, and $\Lambda_{\textnormal{cat}}(x\| m) =  \inf_{p\in \bbR} I_p(x) =I_0(x)$; similarly, the nonuniformity of formation of $x$ in the presence of a catalyst is recovered from $-\Lambda_{\textnormal{cat}}(m\| x) =  \sup_{p\geq 0} I_p(x)=I_\infty(x)$.

\subsection{Inadequacy of the second law as a criterion for state conversion}\label{sec:secondlaw}

In the context of the resource theory of nonuniformity, we might imagine that the role of the second law of thermodynamics is to provide a criterion under which a state conversion $x \mapsto y$ is possible by noisy operations.  The usual statement of the second law is, however, inadequate to this task.

The second law asserts that the Shannon entropy must be nondecreasing in the evolution of a state.  In the context of states of equal dimension, this is equivalent to saying that the Shannon negentropy or the Shannon nonuniformity must be nonincreasing.  

The first inadequacy is that for conversions between states of \emph{unequal} dimension, it is insufficient to look at the Shannon entropy (or the Shannon negentropy).  One must look instead at the Shannon \emph{nonuniformity}.  
For instance, for the pair of states $y= (1/2,1/2)$ and $x=(1/3,1/3,1/3,0)$, the entropies are $H(y)=1$ and $H(x)=\log 3$, so that $H(y) < H(x)$.
And yet the state conversion $x \mapsto y$ is possible by noisy operations, despite involving a \emph{decrease} in entropy.
 This is to be expected, because marginalization is an instance of a noisy operation and it can easily reduce the entropy.  
The Shannon nonuniformity, on the other hand, is nonincreasing under noisy operations.  In particular, in our example $I(x) = \log(4/3)$ while $I(y)=0$. 

The second inadequacy regards sufficiency of the condition.  Even if we consider a conversion between states of equal dimension, so that the Shannon nonuniformity is nonincreasing if and only if the Shannon entropy is nondecreasing, the nondecrease of the Shannon entropy is not a sufficient condition for the state conversion to be possible.  
This is because one can easily construct examples wherein $H(x) \le H(y)$ (equivalently $I(x)\ge I(y)$) but it is not the case that $x \cconv y$.  For instance, for states $x=(2/3,1/6,1/6)$ and $y=(1/2,1/2,0)$ one can easily verify (for instance, via the Lorenz curves) that $x \not\mapsto y$ by noisy operations.  Yet, $H(x)  \simeq 0.811$ and $H(y)= 1$.

Furthermore, one can never rehabilitate the second law by simply replacing entropy with some other nonuniformity monotone because the set of states forms a quasi-order under noisy operations, and therefore no single nonuniformity monotone can decide every state conversion problem.
To provide necessary and sufficient conditions for state conversion, one must look at the values of a \emph{complete set} of nonuniformity monotones.  
To our knowledge, this point was first made in \cite{ruch1978information}.

Note, however, that if a complete set of monotones is of infinite cardinality, then such a criterion will not be of much practical use.  Thus, we would like to also require that our criterion for state conversion be decidable by finite means.  
In many cases this can be accomplished by 
considering \emph{state conversion witnesses}, 
defined in Section~\ref{Freeopsmonotoneswitnesses}.  

In this article, we have seen two examples of complete witnesses for state conversion under noisy operations.  These were defined in Eqs.~\eqref{eq:dualwitness} and~\eqref{eq:dualwitness2}. It follows that a rehabilitated second law of thermodynamics, in the context of the resource theory of nonuniformity, can be stated thus:
$x \cconv y$ if and only if either of the following equivalent conditions hold:
\begin{itemize}
\item $\Delta(x\| y) \ge 0$ where
$$\Delta(x\| y) \assign \min_{k \in \{ 1,\dots,d_y\}} \left( L_x(k/d_y)-L_y(k/d_y)\right)$$
\item
$\Lambda(x\|y) \ge 0$ where 
$$\Lambda(x\|y) \assign  \max \left\{\lambda:  L_x(u) \ge L_y (2^\lambda u) \; \; \forall u \in \mathcal{U} \right\},$$
with 
$$\mathcal{U} \assign \left( \{ \tfrac{1}{d_x},\tfrac{2}{d_x},\dots, 1\} \cup \{ \tfrac{1}{d_y},\tfrac{2}{d_y},\dots, 1\}\right)  \cap( 0, u_* ]$$
and $u_* \assign \min \{1, 2^{-\lambda} \}$. 
\end{itemize}
Evaluating these functions requires an optimization over a finite number of expressions, so that the conversion problem becomes decidable by finite means.

Note that in the resource theory of athermality, the free energy plays the role of the Shannon nonuniformity, but because a single athermality monotone cannot capture the quasi-order of athermal states, the free energy is also inadequate to the task of providing necessary and sufficient conditions for state conversion.

One might think that the role of the second law is to provide a criterion for \emph{catalytic} state conversion rather than the regular variety.
Indeed, many statements of the second law allow for the environment of the system to undergo a change, as long as the process is cyclic, which is to say that in the end, the environment is returned to its initial state.  In the language of resource theories, the environment is allowed to be a catalyst for the state transition~\cite{1ShotAtherm2}.  
However, even under this interpretation, the standard formulation of the second law is inadequate.  For states of equal dimension, the nondecrease of entropy remains only a necessary, but not a sufficient, condition for the state conversion to be possible.  The quasi-order of states under catalytic noisy operations cannot be captured by any single nonuniformity monotone. 

However, a complete witness for state conversion under catalytic noisy operations is implied by Propositions~\ref{lemma:distillationcatstateconversion} and \ref{lemma:formationcatstateconversion}, namely:
$x$ noisy-trumps $y$ if and only if
\begin{itemize}
\item $\Lambda_{\textnormal{cat}}(x\|y) \ge 0$ where
$$\Lambda_{\textnormal{cat}}(x\|y)  \assign  \inf_{p\geq 0} \left(\strut I_p(x) - I_p(y)\right).$$
\end{itemize}
It remains an open problem to provide a criterion for catalytic state conversion that is decidable by finite means.

\section{Approximate state conversion}\label{Sec:Approximatestateconversion}

\subsection{Metrics over state space}

So far, we have discussed the question under what conditions noisy operations can transform a given state  into another state \emph{exactly} or to arbitrary precision.
However, in most physically relevant cases, this is too stringent a requirement, and a different problem becomes important: \emph{Under what
conditions can noisy operations transform a given state into another state which is $\varepsilon$-close to the desired target state for some fixed $\varepsilon$?}

Two immediate questions arise in the study of this problem. First, what kind of distance measure on the set of states should we use? Second, can the
quantum problem be reduced to the classical problem in the same way as in the case of exact state conversion?

The answer to the latter question turns out to be yes, given that we restrict our attention to distance measures---that is, metrics---on the quantum
states and the classical probability distributions that have the crucual property of \emph{contractivity}:

\begin{definition}
\label{DefContractiveMetric}
A \emph{contractive metric} on the quantum states is a map $\mathcal{D}$ that assigns to two quantum states $\rho,\sigma$ on the same underlying
Hilbert space of dimension $d_\rho=d_\sigma$ a non-negative real number $\mathcal{D}(\rho,\sigma)$ satisfying the axioms of a metric
\begin{itemize}
\item $\mathcal{D}(\rho,\sigma)=\mathcal{D}(\sigma,\rho)$,
\item $\mathcal{D}(\rho,\sigma)=0\Leftrightarrow \rho=\sigma$,
\item $\mathcal{D}(\rho,\sigma)\leq \mathcal{D}(\rho,\tau)+\mathcal{D}(\tau,\sigma)$ for all quantum states $\tau$ with $d_\tau=d_\sigma=d_\rho$,
\end{itemize}
such that every completely positive, trace-preserving map $\Phi$ is a contraction, i.e.
\[
   \mathcal{D}\left(\strut\Phi(\rho),\Phi(\sigma)\right)\leq \mathcal{D}(\rho,\sigma).
\]
Analogously, a contractive metric on the classical (discrete) probability distributions is a map $\mathcal{D}$ that assigns to two probability distributions
$x,y$ with the same dimension of the underlying sample space $d_x=d_y$ a non-negative real number $\mathcal{D}(x,y)$ that satisfies the classical
analogs of the axioms above.
\end{definition}
In this definition, the classical analog of a completely positive trace-preserving map $\Phi$ is a channel $\Phi$, i.e.\ a linear map that is described
by a finite stochastic matrix, or equivalently, a linear map that maps probability distributions to probability distributions.

Note that $\mathcal{D}(\rho,\sigma)$, respectively $\mathcal{D}(x,y)$, is undefined if $d_\rho\neq d_\sigma$, respectively $d_x\neq d_y$; however, this does
not mean that the values of $\mathcal{D}$ for different dimensions are completely unrelated. In fact, contractivity implies some relations between
the distances of states of different dimensionalities. For example, if $\tau$ is an arbitrary fixed finite-dimensional quantum state, then the map
$\rho\mapsto\rho\otimes\tau$ is a quantum operation, hence $\mathcal{D}(\rho\otimes\tau,\sigma\otimes\tau)\leq \mathcal{D}(\rho,\sigma)$.
On the other hand, taking the partial trace over the second system is a quantum operation too, implying the converse inequality, such that
we get
\[
   \mathcal{D}(\rho\otimes\tau,\sigma\otimes\tau)=\mathcal{D}(\rho,\sigma).
\]

If we are given a contractive metric $\mathcal{D}$ on the quantum states, then we can construct a contractive metric on the classical probability
distributions by restricting $\mathcal{D}$ to diagonal matrices; we will call this the classical metric that \emph{corresponds to} the respective quantum metric.
For example, the trace distance on quantum states
\[
   \mathcal{D}_\tr(\rho,\sigma):=\frac 1 2 \tr|\rho-\sigma|=\frac 1 2 \|\rho-\sigma\|_1
\]
has the corresponding classical metric
\begin{equation}
   \mathcal{D}_\tr(x,y)=\frac 1 2 \sum_{i=1}^{d_x} |x_i-y_i|
   \label{eqTraceDistance}
\end{equation}
(note that $d_x=d_y$, otherwise the expression $\mathcal{D}(x,y)$ is undefined). Another contractive metric that satisfies the conditions of our
definition is the \emph{purified distance}~\cite{Tomamichel}
\[
   \mathcal{D}_p(\rho,\sigma):=\sqrt{1-F(\rho,\sigma)^2}
\]
with $F(\rho,\sigma):=\|\sqrt{\rho}\sqrt{\sigma}\|_1$ the fidelity, which has the corresponding classical contractive metric
\begin{equation}
   \mathcal{D}_p(x,y)=\sqrt{1-\left(\sum_i \sqrt{x_i y_i}\right)^2}.
   \label{EqPurified}
\end{equation}
A consequence of contractivity is unitary equivalence: if $U$ is unitary, then
\[
   \mathcal{D}(U\rho U^\dagger, U \sigma U^\dagger)=\mathcal{D}(\rho,\sigma).
\]
This is because the unitary map as well as its inverse are quantum operations, and therefore contractivity implies inequalities in both directions. 
It follows that classical contractive metrics are left invariant if we permute the entries of the two probability vectors in its argument in the same way.

In the case of contractive metrics, approximate state transformation in the quantum case can be reduced to the classical problem:
\begin{lemma}
\label{LemApproximateQuantumClassical}
Let $\varepsilon>0$, let $\rho$ and $\sigma$ be two quantum states, and let $\mathcal{D}$ be any contractive metric on the quantum states. Then the following
two statements are equivalent:
\begin{itemize}
\item[(i)] There exists a quantum state $\tilde\sigma$ with $\mathcal{D}(\sigma,\tilde\sigma)\leq\varepsilon$ such that noisy quantum operations
can transform $\rho$ to $\tilde\sigma$.
\item[(ii)] There exists a distribution $\tilde s$ with $\mathcal{D}(\lambda(\sigma),\tilde s)\leq\varepsilon$ such that noisy classical operations can transform $\lambda(\rho)$
to $\tilde s$.
\end{itemize}
\end{lemma}
Here, the metric $\mathcal{D}$ in (ii) is the classical contractive metric corresponding to the quantum contractive metric $\mathcal{D}$.  Note that in general $\tilde s\neq\lambda(\tilde\sigma)$.

\begin{proof}
(ii)$\Rightarrow$(i): with $\sigma=\sum_k\lambda_k(\sigma)|k\rangle\langle k|$, define $\tilde\sigma:=\sum_k \tilde s_k |k\rangle\langle k|$,
then $\mathcal{D}(\sigma,\tilde\sigma)\leq \varepsilon$. Since noisy classical operations can transform $\lambda(\rho)$ into $\lambda(\tilde\sigma)$,
Lemma~\ref{lemma:quantumNOtoclassicalNO} applies and proves that noisy operations can transform $\rho$ into $\tilde\sigma$.

To see (i)$\Rightarrow$(ii), let $\Phi$ be the quantum operation that dephases in the eigenbasis of $\sigma$; i.e.\ $\Phi(\tau)
:=\sum_k \langle k|\tau|k\rangle |k\rangle\langle k|$, and set $\tilde s_k:=\langle k|\tilde\sigma|k\rangle$, then $\Phi(\tilde\sigma)=\sum_k \tilde s_k |k\rangle\langle k|$, and
\[
   \mathcal{D}(\lambda(\sigma),\tilde s)=\mathcal{D}(\sigma,\Phi(\tilde\sigma))=\mathcal{D}(\Phi(\sigma),\Phi(\tilde\sigma))\leq \mathcal{D}(\sigma,\tilde\sigma)\leq\varepsilon.
\]
Again, due to Lemma~\ref{lemma:quantumNOtoclassicalNO}, the assumption that $\rho \cconv \tilde\sigma$ implies that $\lambda(\rho) \cconv \lambda(\tilde\sigma)$. 
Furthermore,
according to Exercise II.1.12, ``Schur's Theorem'' in \cite{Bhatia}, the eigenvalues of $\tilde\sigma$ majorize the diagonal elements of $\tilde\sigma$ in the eigenbasis of $\sigma$,
hence $\lambda(\tilde\sigma)\succ \tilde s$, and consequently $\lambda(\tilde\sigma) \cconv \tilde s$.
It follows that $\lambda(\rho)\cconv \tilde s$.
\end{proof}

This lemma allows us to study approximate state conversion in purely classical terms. Hence, in the following, when we refer to a contractive metric $\mathcal{D}$,
we always assume that it is a contractive metric on the \emph{classical probability distributions}.

For convenience, we will make use of the following notational convention:
\begin{definition}[$\varepsilon$-approximate state conversion]
We write $x \epsiloncconv y$ if there exists a noisy classical operation taking $x$ to a state that is $\varepsilon$-close to $y$ relative to a contractive metric $\mathcal{D}$.
\end{definition}

The following transitivity property will be used several times: if we have both
\[
x \epsiloncconv y \qquad  \mbox{and}\qquad y \deltacconv z
\]
then it follows that
\[
x \epsplusdeltacconv z.
\]

\subsection{Smoothed entropies}

As we have seen in Subsection~\ref{section:formationanddistillation},
the order-$\infty$ and order-0 R\'{e}nyi nonuniformities determine the single-shot nonuniformity of formation and distillable nonuniformity of a state.  When the state preparation is allowed to be approximate, these quantities turn out to be given by smoothed versions of these R\'{e}nyi nonuniformities.  We begin, therefore, with a discussion of smoothing.
\begin{definition}
\label{DefSmoothEntropies}
Fix any contractive metric $\mathcal{D}$.
For probability distributions $x$ and $\varepsilon\geq 0$, define the smooth order-0 and order-$\infty$ R\'{e}nyi entropies as
\begin{eqnarray*}
   H_\infty^\varepsilon(x)&:=& \max_{x':\, \mathcal{D}(x,x')\leq\varepsilon} H_\infty(x'),\\
   H_0^\varepsilon(x)&:=& \min_{x':\, \mathcal{D}(x,x')\leq\varepsilon} H_0(x').
\end{eqnarray*}
Similarly, we define the smooth order-0 and order-$\infty$ R\'{e}nyi nonuniformities as
\begin{eqnarray*}
   I_\infty^\varepsilon(x)&:=& \log d_x - H_\infty^\varepsilon(x)  =   \min_{x':\, \mathcal{D}(x,x')\leq\varepsilon}  I_\infty(x'),\\
   I_0^\varepsilon(x)&:=& \log d_x - H_0^\varepsilon(x) =  \max_{x':\, \mathcal{D}(x,x')\leq\varepsilon} I_0(x').
\end{eqnarray*}
\end{definition}

It is clear that the values of these quantities depend on the choice of metric $\mathcal{D}$. In the following, we will always assume that the reader
is aware of this freedom of choice, and we will not always repeat that the statements are contingent on the choice of $\mathcal{D}$.

Note that the optimizations in Definition~\ref{DefSmoothEntropies} are over \emph{normalized} classical probability distributions $x'$.
This differs from some places in the literature where subnormalized (quantum) states are taken into account, for example in~\cite{Tomamichel}.

In Lemma~\ref{LemApproximateQuantumClassical} above, we have shown that we can without loss of generality restrict our analysis to classical probability distributions.
However, some readers may be interested in translating some results of this section directly into the quantum case, in order to apply
them in other contexts. This turns out to be possible whenever we speak about the smooth \emph{R\'{e}nyi $\infty$-entropy}: defining the quantum version of the smooth R\'{e}nyi $\infty$--entropy
in the usual way, the resulting quantity agrees with the classical version applied to the spectrum:
\begin{equation}
   S_\infty^\varepsilon(\rho):=\max_{\rho':\,\, \mathcal{D}(\rho,\rho')\leq\varepsilon} S_\infty(\rho')=H_\infty^\varepsilon(\lambda(\rho)),
   \label{eqHMinQC}
\end{equation}
where $S_\infty(\sigma):=-\log \|\sigma\|_\infty$.
[This definition resembles that in \cite{RennerThesis} and \cite{Tomamichel} but involves any contractive metric and involves smoothing over only normalized states.]
However, we do not know whether the analogous relation holds for the smooth R\'{e}nyi 0-entropy $H_0^\varepsilon$ in general; only that
\[
   S_0^\varepsilon(\rho):=\min_{\rho':\,\, \mathcal{D}(\rho,\rho')\leq\varepsilon} S_0(\rho')\leq H_0^\varepsilon(\lambda(\rho)).
\]
Note that, unlike eq.~\eqref{eqHMinQC}, we have an \emph{inequality} here.  We do get equality in the special case
where the contractive metric is chosen to be the trace distance $\mathcal{D}_\tr$ from~(\ref{eqTraceDistance}), as the following lemma asserts (proof in Appendix~\ref{app:proofs}):
\begin{lemma}
\label{LemMaxQC}
Let $\mathcal{D}$ be any contractive metric on the quantum states for which there exists a norm $\|\cdot\|$ on the
self-adjoint matrices such that $\mathcal{D}(\rho,\sigma)=\|\rho-\sigma\|$ (for example the trace distance,
$\mathcal{D}=\mathcal{D}_\tr$ with $\|\cdot\|=\frac 1 2 \|\cdot\|_1$).
Then the quantum version of the smooth max-entropy,
\[
   S_0^\varepsilon(\rho):=\min_{\bar\rho:\,\,\mathcal{D}(\rho,\bar\rho)\leq\varepsilon} S_0(\bar\rho),
\]
with $S_0(\bar\rho):=\log \rank(\bar\rho)$, agrees with the classical smooth max-entropy on the spectrum of $\rho$:
\[
   S_0^\varepsilon(\rho)=H_0^\varepsilon(\lambda(\rho)).
\]
\end{lemma}

The fact that this restriction is necessary is a first manifestation of the fact that $H_0^\varepsilon$ behaves in some ways ``less nicely'' than $H_\infty^\varepsilon$,
as we will see several times below, for example in the following lemma:
\begin{lemma}
\label{LemMinMaxMonotonicity}
For every $\varepsilon>0$, the quantity $I_\infty^\varepsilon$ is a nonuniformity monotone, but $I_0^\varepsilon$ is not, in general.
\end{lemma}
\begin{proof}
To prove that $I_\infty^\varepsilon$ is a nonuniformity monotone, suppose that $x$ and $y$ are distributions such that
$y=N(x)$ for some noisy operation $N$. Then, due to contractivity, the set of states $y'$ with $\mathcal{D}(y,y')\leq\varepsilon$
contains the set of distributions $N(x')$ with $\mathcal{D}(x,x')\leq\varepsilon$. Moreover, we know that $I_\infty$ (without the $\varepsilon$) is
a nonuniformity monotone. Thus
\begin{eqnarray*}
   I_\infty^\varepsilon(N(x))&=& \min_{y':\, \mathcal{D}(y,y')\leq\varepsilon} I_\infty(y')\\
   &\leq& \min_{x'\, \mathcal{D}(x,x')\leq\varepsilon} I_\infty(N(x'))\\
   &\leq&  \min_{x':\, \mathcal{D}(x,x')\leq\varepsilon} I_\infty(x')= I_\infty^\varepsilon(x).
\end{eqnarray*}

Now we consider $I_0^\varepsilon$.  Any nonuniformity monotone $M$ must satisfy $M(x\otimes m^{(d)})=M(x)$ for $m^{(d)}=(1/d,\ldots,1/d)$ the uniform state on
$\mathbb{R}^d$ (as mentioned in Eq.~(\ref{eq:NUmonotoneAddm})).
But this equation is not generally true for $I_0^\varepsilon$.

As a concrete counterexample, let $0<\varepsilon<\frac 1 2$, $d\geq\max\left\{ 1/(2\varepsilon),3\right\}$, and consider the trace distance $\mathcal{D}_{\tr}$.
If $x=m^{(2)}$ is the maximally
mixed state on one bit, then $I_0^\varepsilon(x)=0$. Let $y$ be the distribution
$\left(\frac 1 d, \frac 1 {2d},\frac 1 {2d},\ldots,\frac 1 {2d},0\right)\in\mathbb{R}^{2d}$. Then $\mathcal{D}_{\tr}(y,x\otimes m^{(d)})=1/(2d)\leq\varepsilon$,
so $H_0^\varepsilon(x\otimes m^{(d)})\leq H_0(y)=\log(2d-1)$, and $I_0^\varepsilon(x\otimes m^{(d)})\geq \log(2d)-\log(2d-1)>0.$
\end{proof}

In the case that the reference metric is chosen to be the trace distance $\mathcal{D}_\tr$, the smooth R\'enyi $0$-entropy $H_0^\varepsilon$ and
the corresponding nonuniformity $I_0^\varepsilon$ can be evaluated explicitly.
\begin{lemma}
\label{LemH0eps}
If the trace distance $\mathcal{D}_\tr$ is chosen as reference metric, then any distribution $x$ has smooth R\'enyi $0$-entropy
\[
   H_0^\varepsilon(x)=\log k,
\]
where $k\geq 1$ is chosen as the smallest integer such that $\sum_{i=1}^k x_i^\downarrow\geq 1-\varepsilon$, with $x_1^\downarrow\geq x_2^\downarrow\geq\ldots$
denoting the components of $x$ in non-increasing order.
\end{lemma}
The proof is given in Appendix~\ref{app:proofs}.

This lemma shows how $I_0^\varepsilon(x)$ can be determined graphically in terms of the Lorenz curve of $x$, if the reference
metric is the trace distance: denoting the two coordinates for the plot of the Lorenz curve by $u$ and $v$, draw the line $v=1-\varepsilon$
parallel to the $u$-axis, and determine the left-most elbow $u_0$ of the Lorenz curve that lies on or above that line.
Then $I_0^\varepsilon(x)=-\log u_0$. This procedure is depicted in Fig.~\ref{fig:I_0_eps_Lorenz}.

\begin{figure}[h!]
\centering
\includegraphics[width=.3\textwidth, clip=true]{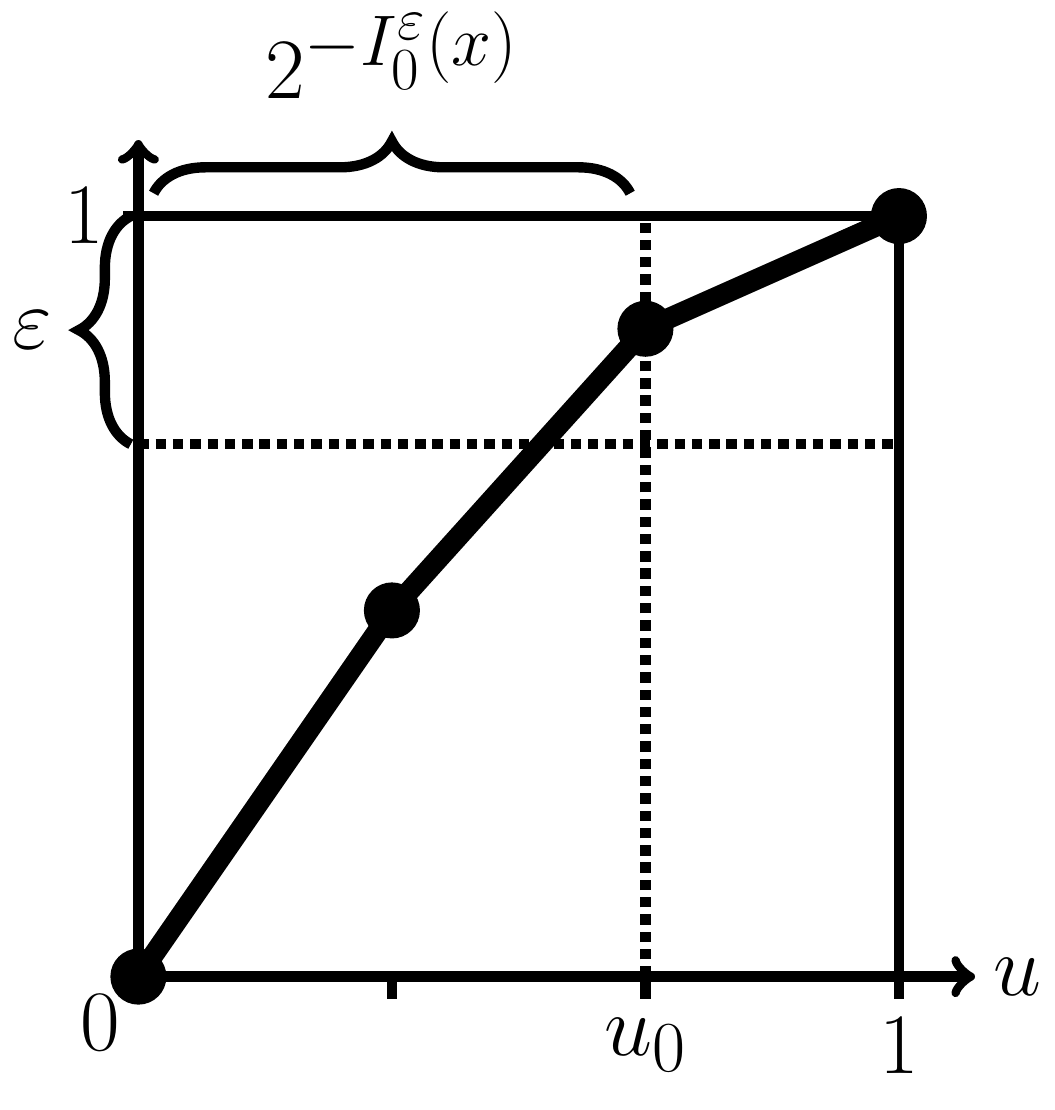}
\captionsetup{singlelinecheck=off,justification=raggedright}
\caption{If the trace distance is used as the reference metric, the smooth order-$0$ R\'enyi nonuniformity $I_0^{\varepsilon}(x)$ can be determined graphically
in terms of the Lorenz curve of $x$: draw the horizontal line $v=1-\varepsilon$, and mark the left-most elbow that is on or above that line.
Denote its $u$-coordinate by $u_0$, then $I_0^\varepsilon(x)=-\log u_0$. Since this quantity does not only depend on the Lorenz curve, but
also on the number and location of its elbows, it is clear that $I_0^\varepsilon$ cannot be a nonuniformity monotone.}
\label{fig:I_0_eps_Lorenz}
\end{figure}

This gives a simple geometric explanation why $I_0^\varepsilon$ is not a nonuniformity monotone, as demonstrated in Lemma~\ref{LemMinMaxMonotonicity}: this quantity
does not only depend on the Lorenz curve of the state, but also on its elbows. The map $x\mapsto x\otimes m^{(d)}$
changes the set of elbows (but not the Lorenz curve), and thus changes the value of $I_0^\varepsilon$ (but not that of any nonuniformity monotone).

\subsection{Nonuniformity of formation and distillable nonuniformity for approximate state preparation}\label{SubsecApproxFormationDistillation}

The quantity $I_\infty^\varepsilon$ turns out to quantify the \emph{single-shot nonuniformity of formation} if the formation is allowed to be achieved only approximately.

\begin{lemma}\label{lemma:approxformation}
The minimum $I$ such that $s_I \epsiloncconv x$, is given by the $\varepsilon$-smoothed order-($\infty$) R\'{e}nyi nonuniformity,
\[
I^{\varepsilon}_\infty(x).
\]
\end{lemma}

\begin{proof}
According to Lemma~\ref{lemma:formation}, in order to create a state $x'$ with $\mathcal{D}(x,x')\leq\varepsilon$, one needs
a sharp state of nonuniformity at least $I_\infty(x')$.
Optimizing over all possible $x'$ yields the minimal nonuniformity of any sharp state
needed to create \emph{some} state $x'$ which is $\varepsilon$-close to $x$:
\[
   \min_{x':\, \mathcal{D}(x,x')\leq\varepsilon} I_\infty(x') =  I_\infty^\varepsilon(x).
\]
\end{proof}

As it turns out, it is not so simple to prove an analogous statement for the distillable nonuniformity. While one can achieve
distillation of a sharp state of nonuniformity $I_0^\varepsilon(x)$ from $x$, this is not in general the optimal result.
\begin{lemma}\label{LemDistillSufficient}
It is possible to achieve $ x \epsiloncconv s_I$, for $I$ equal to the $\varepsilon$-smoothed order-$0$ R\'{e}nyi nonuniformity,
\[
I^{\varepsilon}_0(x).
\]
\end{lemma}
\begin{proof}
Choose $x'$ such that $I_0^\varepsilon(x)=I_0(x')$ and $\mathcal{D}(x,x')\leq\varepsilon$. According to Lemma~\ref{lemma:distillation},
one can distill a sharp state $s_I$ of nonuniformity $I=I_0(x')$ by noisy operations from $x'$. Denote this noisy operation by $N$,
then $N(x')=s_I$. Set $y:=N(x)$, then
\[
   \mathcal{D}(y,s_I)=\mathcal{D}(N(x),N(x'))\leq \mathcal{D}(x,x')\leq\varepsilon,
\]
proving the claim.
\end{proof}

If the expression in this lemma were optimal -- that is, if $I_0^\varepsilon(x)$ gave the \emph{maximal} nonuniformity of any sharp state
that can be extracted from $x$ -- then $I_0^\varepsilon$ would have to be a nonuniformity monotone. However, we know from Lemma~\ref{LemMinMaxMonotonicity}
that it is not. Therefore, the optimal expression for distillable uniformity must be different. In fact, as a simple corollary of Lemma~\ref{LemDistillSufficient},
we see immediately that nonuniformity
\begin{equation}
   I_0^\varepsilon(x\otimes m^{(d)})
   \label{eqAddMaxMix}
\end{equation}
for the $d$-dimensional uniform state $m^{(d)}$ (with $d$ arbitrary) can be achieved as well: we can add on the uniform state $m^{(d)}$ for free,
and apply Lemma~\ref{LemMinMaxMonotonicity} to $x\otimes m^{(d)}$. As shown in Lemma~\ref{LemMinMaxMonotonicity}, this quantity can be larger than
the corresponding one for $x$ alone.
In fact, we will show in the appendix that expression~(\ref{eqAddMaxMix}) is increasing in $d$, and the \emph{maximal nonuniformity of any sharp state that
can be extracted from $x$} is given by
\begin{equation}
   J_0^\varepsilon(x):=\lim_{d\to\infty} I_0^\varepsilon(x\otimes m^{(d)}).
   \label{eqOptimalMainText}
\end{equation}
The result from Lemma~\ref{LemDistillSufficient}, though not optimal, will be sufficient for all further applications: it will allow
us to obtain conditions for state conversion and compute the optimal asymptotic transition rate. Therefore, we defer the formal proof of~(\ref{eqOptimalMainText})
to the appendix.

However, we note that in the case that the reference metric is the trace distance, the expression~(\ref{eqOptimalMainText})
can be evaluated directly in terms of the Lorenz curve of $x$: adding on the uniform state $m^{(d)}$ to $x$ with increasing $d$ does not change the Lorenz curve,
but adds more and more elbows. In the limit of $d\to\infty$, the set of elbows becomes dense on the Lorenz curve. Applying the prescription
from Fig.~\ref{fig:I_0_eps_Lorenz} to determine $I_0^\varepsilon(x\otimes m^{(d)})$ means  -- in the limit of large $d$ -- to simply determine
the point $u_0$ where the line $v=1-\varepsilon$ intersects the Lorenz curve. Then $J_0^\varepsilon(x)=-\log u_0$.

\begin{figure}[h!]
\centering
\includegraphics[width=.3\textwidth, clip=true]{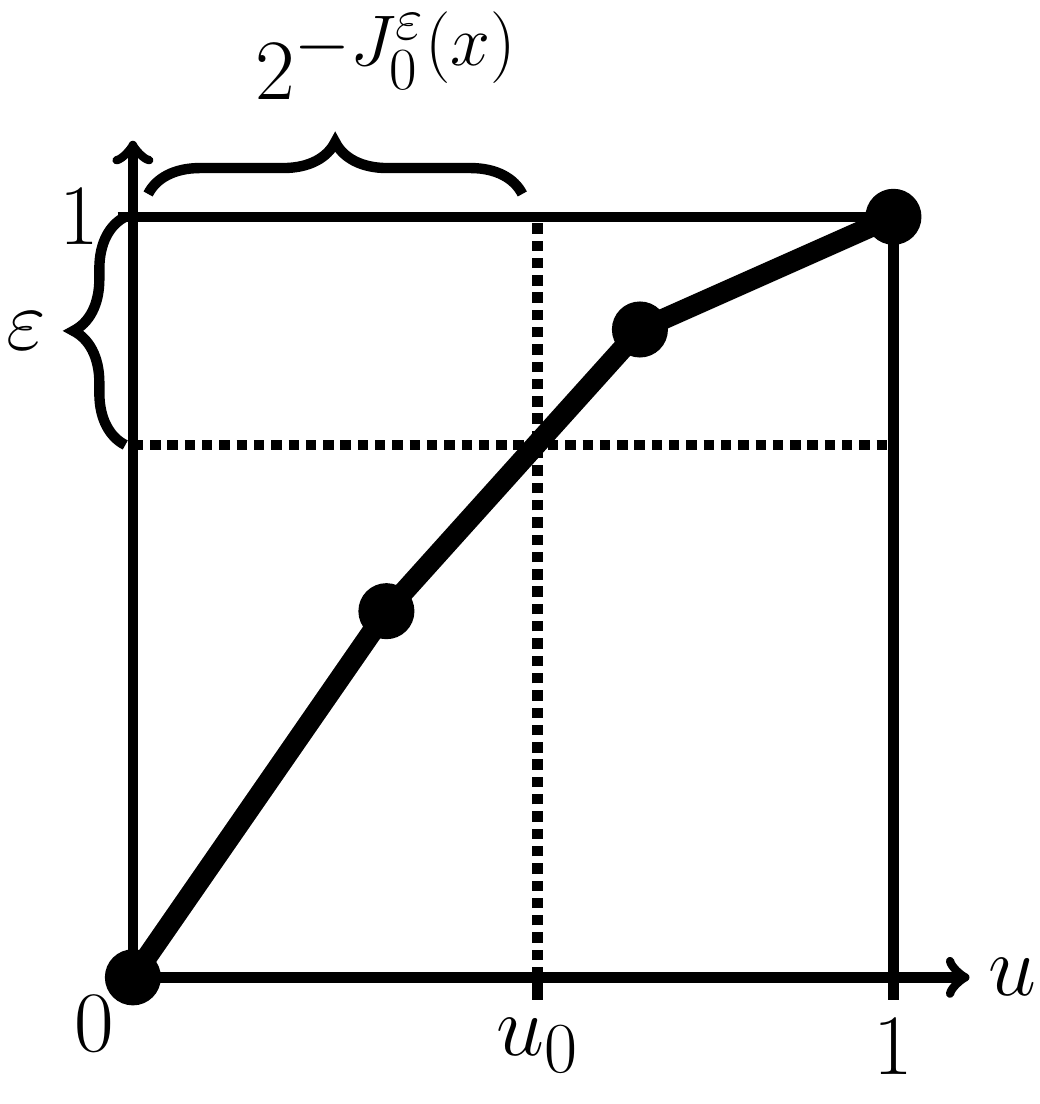}
\captionsetup{singlelinecheck=off,justification=raggedright}
\caption{If the trace distance is used as the reference metric, the distillable nonuniformity $J_0^\varepsilon(x)=\lim_{d\to\infty} I_0^\varepsilon(x\otimes m^{(d)})$ of
any state $x$ can be determined graphically by means of the Lorenz curve of $x$: determine the (smallest) value $u_0$ such that $L_x(u_0)=1-\varepsilon$,
i.e.\ the intersection of the horizontal line $v=1-\varepsilon$ and the Lorenz curve. Then $J_0^\varepsilon(x)=-\log u_0$.}
\label{fig:funnysmoothing}
\end{figure}

This reproduces the smoothing procedure for the ``min-free energy'' in~\cite{FundLimitsNature} in the case of trivial Hamiltonians $H=0$ (cf.\ their
Supplementary Figure S2).
It is depicted in Fig.~\ref{fig:funnysmoothing}. It is intriguing that the Lorenz curve again turns out to be a natural and useful tool also in the case of
approximate distillation of nonuniformity. However, note that the Lorenz curve procedure to determine $J_0^\varepsilon$ is only valid if the trace distance is used as the
reference metric. Expression~(\ref{eqOptimalMainText}) is more general: it is valid for all contractive metrics.

Our results above also give a qualitative confirmation of the statements in \cite{WorkValOfInfo}: $I_0^\varepsilon$ quantifies
the extractable nonuniformity (and thus work) in their Theorem I. Since $I_\infty^\varepsilon$ quantifies the nonuniformity needed to construct a state,
any attempt to obtain more nonuniformity (and thus work) back from a given state will only be possible with very small probability of success, coming from
mere chance by guessing bits correctly. This is reflected in Theorem II in their paper. However, note that the scenario in their paper differs from ours.
For example, while the resource theory of nonuniformity allows one to add on uniform states for free, the number of bits (or Szilard engines)
in \cite{WorkValOfInfo} is fixed at some number $n$.

In Section~\ref{section:formationanddistillation}, a sufficient condition for exact state conversion was derived from the expressions for the nonuniformity of formation and distillable nonuniformity.
In this section, we determine an analogous condition for \emph{approximate} state conversion, based on the results of the previous section. 
As in Subsection~\ref{SubsecApproxFormationDistillation}, we will assume that $\mathcal{D}$ is an arbitrary contractive metric as defined in Definition~\ref{DefContractiveMetric}.
The two quantities $I_\infty^\varepsilon$ and $I_0^\varepsilon$ yield a sufficient condition for approximate state interconvertibility:
\begin{lemma}
\label{LemConvert}
If $x$ and $y$ are states such that
\[
   I_0^{\varepsilon/2}(x)\geq I_\infty^{\varepsilon/2}(y),
\]
then $x\epsiloncconv y$.
\end{lemma}

\begin{proof}
We have already provided the proof for the $\varepsilon=0$ case in Lemma~\ref{lemma:gowitness}.
To extend the result to $\varepsilon>0$, note that
\begin{eqnarray*}
   I_0^{\varepsilon/2}(x)&=&\max_{\bar x:\, \mathcal{D}(\bar x,x)\leq\varepsilon/2} I_0(\bar x),\\
   I_\infty^{\varepsilon/2}(y)&=&\min_{\bar y:\, \mathcal{D}(\bar y, y)\leq \varepsilon/2} I_\infty(\bar y).
\end{eqnarray*}
Let $\bar x$ and $\bar y$ be states that achieve the maximum respectively minimum in these optimizations. Then
\[
   I_0(\bar x)=I_0^{\varepsilon/2}(x)\geq I_\infty^{\varepsilon/2}(y)=I_\infty(\bar y).
\]
Due to Lemma~\ref{lemma:gowitness}, there is a noisy operation $N$ such that $N(\bar x)=\bar y$.
Now let $y':=N(x)$. Then
\begin{eqnarray*}
    \mathcal{D}(y',y)&\leq& \mathcal{D}(y',\bar y)+\mathcal{D}(\bar y,y)\\
    &=& \mathcal{D}(N(x),N(\bar x))+\mathcal{D}(\bar y,y) \\
    &\leq& \mathcal{D}(x,\bar x)+\varepsilon/2\leq \varepsilon.
\end{eqnarray*}
This proves the claim.
\end{proof}

The previous lemma gives a sufficient condition for approximate convertibility. Similarly, we can obtain a necessary condition:
\begin{lemma}
\label{LemNoConvert}
If $x$ and $y$ are states such that
\[
   x\epsiloncconv y,
\]
then $I_\infty^{\varepsilon+\delta}(y)\leq I_\infty^\delta(x)$ for every $\delta\geq 0$.
\end{lemma}
\begin{proof}
Denote the state which is $\varepsilon$-close to $y$ and generated from $x$ by $y'$.
Lemma~\ref{LemMinMaxMonotonicity} implies that $I_\infty^\delta(x)\geq I_\infty^\delta(y')$. Let $\bar y$ be any state
such that $I_\infty^\delta(y')=I_\infty(\bar y)$ and $\mathcal{D}(y',\bar y)\leq\delta$. Then
$\mathcal{D}(y,\bar y)\leq \mathcal{D}(y,y')+\mathcal{D}(y',\bar y)\leq \varepsilon+\delta$,
and we obtain $I_\infty^{\varepsilon+\delta}(y)\leq I_\infty(\bar y)=I_\infty^\delta(y')\leq I_\infty^\delta(x)$.
\end{proof}

\subsection{Asymptotic state conversion}
\label{sec:Asym}

Our previous results on single-shot $\varepsilon$-noisy nonuniformity distillation and dilution, Lemmas~\ref{LemConvert} and~\ref{LemNoConvert}, allow us to recover (and slightly
generalize) the result in~\cite{HHOShort,HHOLong} on the asymptotic conversion rate between two nonuniform states.
Due to Lemma~\ref{LemApproximateQuantumClassical}, it is again sufficient to consider conversion between classical probability distributions.

The question we are interested in is the following. Suppose we are given $n$ copies of some state $x$, which we would like to use to create as many
copies of another state $y$ as possible, using only noisy operations. That is, we would like to transform $x^{\otimes n}$ into $y^{\otimes m}$, with $m$ as large as possible.
However, we do not demand that the transformation is perfect---instead, we allow a small error $\varepsilon$. That is, for every $n$, we ask for the maximal integer $m_n$
such that
\[
      x^{\otimes n} \conv\limits^{\varepsilon_n\textnormal{-noisy}}  y^{\otimes m_n},
\]
and $\varepsilon_n\to 0$ as $n\to\infty$. In particular, we are interested in how $m_n$ scales with $n$. We will first answer this question
for the case of constant error, i.e.\ $\varepsilon_n=\varepsilon$ for all $n$; afterwards, we will obtain the answer for asymptotically vanishing error $\varepsilon_n\to 0$
as a simple corollary.

We have previously obtained necessary and sufficient conditions for approximate state conversion in terms of smooth order-0 and order-$\infty$ R\'{e}nyi entropies. It is a well-known fact that these entropies converge to the Shannon entropy in the asymptotic limit of many copies of a state:
\begin{lemma}
\label{LemAEP}
If the reference metric is taken to be either the \emph{purified distance} $\mathcal{D}_p$ from~(\ref{EqPurified})
or the \emph{trace distance} $\mathcal{D}_\tr$ from~(\ref{eqTraceDistance}), then we have for every $0<\varepsilon<1$
\begin{eqnarray}
   \label{HMinAsympt}
   \lim_{n\to\infty}\frac 1 n H_\infty^\varepsilon(x^{\otimes n}) &=& H(x)\\ 
   \lim_{n\to\infty}\frac 1 n H_0^\varepsilon(x^{\otimes n}) &=& H(x) \label{HMaxAsympt}
\end{eqnarray}
for all finite probability distributions $x$, where $H(x)=-\sum_i x_i\log x_i$ is the Shannon entropy.
\end{lemma}
Several versions of this statement can be found in the literature; for example, in~\cite{Tomamichel}. Some care has to be taken, however, with respect
to the slightly different definitions of the smooth entropies. For example, in contrast to~\cite{Tomamichel}, we are here only considering classical
probability distributions, we define smooth entropies in terms of optimization over \emph{normalized} states (not subnormalized ones), and
we define the smooth max-entropy as $H_0^\varepsilon$ and not as $H_{1/2}^\varepsilon$, with $H_\alpha$ the R\'enyi entropy.
For completeness, we give a proof of Lemma~\ref{LemAEP} in Appendix~\ref{app:proofs}.

In contrast with Section~\ref{SubsecApproxFormationDistillation}, where we allowed an arbitrary choice of contractive metric,
the previous lemma is only proven for the trace distance and the purified distance. This is not completely unexpected: it simply
cannot be true for all contractive metrics; for example, it fails for the discrete metric
\[
   \mathcal{D}(x,y):=\left\{
      \begin{array}{cl}
         0 & \mbox{if }x=y,\\
         1 & \mbox{if }x\neq y.
      \end{array}
   \right.
\]
Furthermore, the trace distance and purified distance are the most frequently used metrics with clear operational meaning. We leave it open whether the lemma above -- or
the following one -- can be proven in greater generality.

Now we recover the result in~\cite{HHOShort,HHOLong} on the asymptotic conversion rate. Our proof
turns out to be considerably simpler than the original one: it suffices to combine
the single-shot results, Lemmas~\ref{LemConvert} and~\ref{LemNoConvert}, with the asymptotic equipartition property, Lemma~\ref{LemAEP}, to obtain the rate directly.

\begin{lemma}
\label{LemAsymptoticRate}
Let $x$ and $y$ be two states of possibly different dimensionalities that are not both uniform, let $0<\varepsilon<1$,
and choose either the trace distance or the purified distance as reference metric.
For every $n\in\mathbb{N}$, let $m_n$ be the largest integer such that
\[
      x^{\otimes n}\epsiloncconv y^{\otimes m_n}.
\]
Then
\[
   \lim_{n\to\infty}\frac{m_n} n =\frac{I(x)}{I(y)},
\]
where $I(z):=\log d_z - H(z)$, with $H$ the Shannon entropy. Similarly, let $k_n$ be the smallest integer such that
\[
   x^{\otimes k_n}\epsiloncconv y^{\otimes n}.
\]
Then
\[
   \lim_{n\to\infty}\frac{k_n} n =\frac{I(y)}{I(x)}.
\]
\end{lemma}
\textbf{Remark.} Due to Lemma~\ref{LemApproximateQuantumClassical}, the analogous statement for quantum states follows immediately.
Thus, we fully recover the result in~\cite{HHOShort,HHOLong}, generalized to constant error $\varepsilon$.

\begin{proof}
A moment's thought shows that the statement of the lemma is trivially true if either $x$ or $y$ is uniform (such that $I(x)=0$ or $I(y)=0$); hence
we only have to prove the case that $I(x)>0$ and $I(y)>0$.
We only prove the first claim; the proof of the second claim is analogous.
Clearly $m_n$ is increasing in $n$. Since for every $m\in\mathbb{N}$, we have
\[
   I_0^{\varepsilon/2}(x^{\otimes n})=n\cdot I(x)+o(n) > I_\infty^{\varepsilon/2}(y^{\otimes m})
\]
for $n$ large enough, Lemma~\ref{LemConvert} implies that $m_n\to\infty$ as $n\to\infty$. After some rearranging, Lemma~\ref{LemConvert} also says that if
the inequality
\[
   \frac m n \leq \frac{\frac 1 n I_0^{\varepsilon/2}(x^{\otimes n})}{\frac 1 m I_\infty^{\varepsilon/2}(y^{\otimes m})}
\]
is satisfied, then noisy operations can transform $x^{\otimes n}$ into a state that is $\varepsilon$-close to $y^{\otimes m}$.
By construction, this is impossible for $m:=m_n+1$, so
\[
   \frac {m_n+1} n > \frac{\frac 1 n I_0^{\varepsilon/2}(x^{\otimes n})}{\frac 1 {m_n+1} I_\infty^{\varepsilon/2}(y^{\otimes (m_n+1)})}.
\]
Thus, we obtain
\begin{eqnarray*}
   \liminf_{n\to\infty} \frac{m_n} n &\geq& \liminf_{n\to\infty} \frac{\frac 1 n I_0^{\varepsilon/2}(x^{\otimes n})}{\frac 1 {m_n+1} I_\infty^{\varepsilon/2}(y^{\otimes (m_n+1)})}\\
   &=&\frac{\log d_x -\lim_{n\to\infty}\frac 1 n H_0^{\varepsilon/2}(x^{\otimes n})}{\log d_y-\lim_{m\to\infty}\frac 1 m H_\infty^{\varepsilon/2}(y^{\otimes m})} \\
   &=& \frac{\log d_x-H(x)}{\log d_y-H(y)}.
\end{eqnarray*}
Conversely, choose $\delta\in (0,1-\varepsilon)$ arbitrarily (such that $\varepsilon+\delta<1$), then Lemma~\ref{LemNoConvert} shows that the inequality
\[
   \frac {m_n} n \leq \frac{\frac 1 n I_\infty^\delta(x^{\otimes n})}{\frac 1 {m_n}I_\infty^{\varepsilon+\delta}(y^{\otimes m_n})}.
\]
holds for every $n$, from which we can analogously infer by computing the limit of the right-hand side that
\[
   \limsup_{n\to\infty} \frac{m_n} n \leq \frac{\log d_x-H(x)}{\log d_y-H(y)}.
\]
This proves that the limit $\lim_{n\to\infty} m_n/n$ exists and equals $I(x)/I(y)$.
\end{proof}

As a simple corollary, we recover the result for \emph{asymptotically vanishing error}, i.e.\ for the case that $\varepsilon_n\to 0$ as $n\to\infty$ (however slowly).
The proof idea is simple: use the method from Lemma~\ref{LemAsymptoticRate} to achieve fixed error $\varepsilon$; then, as $n$ grows, decrease $\varepsilon$
to zero, but do this slowly enough such that the fraction $m_n/n$ still has the same limit as in Lemma~\ref{LemAsymptoticRate}. In the appendix, we give a proof
that formally shows that this idea works.

By a \emph{protocol}, we mean any choice of integers $(m_n)_{n\in\mathbb{N}}$ and noisy operations that transform $x^{\otimes n}$ into $y^{\otimes m_n}$ up
to error $\varepsilon_n$ for every $n$. With this definition, we can formalize the above idea into
\begin{lemma}
\label{LemAsymptoticRateVanishing}
There is a protocol for transforming $n$ copies of $x$ into $m_n$ copies of $y$ with asymptotically vanishing error at rate
\[
   r:=\lim_{n\to\infty}\frac {m_n} n =\frac{I(x)}{I(y)};
\]
however, no higher rate is achievable by any protocol of this kind.
\end{lemma}
The proof is provided in Appendix~\ref{app:proofs}.

\subsection{Nonuniformity cost and yield of asymptotic state conversion}

Now we return to the problem posed in Subsection~\ref{sec:costyieldstateconv}, which is to quantify the nonuniformity cost and yield of state conversion.
We can analyze this problem in the asymptotic case, and the answer turns out to be simple. Suppose that $I(x)< I(y)$,
then -- for $n$ large enough -- the approximate state conversion
\begin{equation}
   x^{\otimes n}\epsiloncconv y^{\otimes n}
   \label{eqApproxNN}
\end{equation}
is impossible -- otherwise we would have $I(x)\geq I(y)$ according to Lemma~\ref{LemAsymptoticRate} above. However, we may be able to create
$y^{\otimes n}$ if, in addition to consuming $x^{\otimes n}$, we can consume a certain number, $m_n$,  of pure bits (i.e. sharp states of nonuniformity 1), that is,
\[
   x^{\otimes n}\otimes s_1^{\otimes m_n} \epsiloncconv y^{\otimes n}.
\]
It is natural to try and determine the minimum value of $m_n$ such that one can still achieve this conversion. 

Similarly, if $I(x)> I(y)$, then the conversion~(\ref{eqApproxNN})
is possible, but ``inefficient'' -- in addition to $y^{\otimes n}$, it will be possible to extract a certain number $m_n$ of pure bits, such that
\[
   x^{\otimes n} \epsiloncconv y^{\otimes n}\otimes s_1^{\otimes m_n}.
\]
It is also natural to ask what the maximal value of $m_n$ is such that this conversion can be achieved.

In Subsection~\ref{sec:costyieldstateconv}, we answered the analogous questions for the case of single states $x$ and $y$ (not $x^{\otimes n}$ or $y^{\otimes n}$) and for exact conversion. We provide the answer in the asymptotic case, for approximate conversion, in Lemma~\ref{LemAsymptCost} below.

There are two natural ways to prove Lemma~\ref{LemAsymptCost}. One is to directly generalize the proof of Lemma~\ref{LemAsymptoticRate},
and to use the inequalities (for $\varepsilon,\delta\geq 0$)
\begin{eqnarray*}
   I_\infty^{\varepsilon+\delta}(x\otimes y)&\leq& I_\infty^\varepsilon(x)+I_\infty^\delta(y),\\
   I_0^{\varepsilon+\delta}(x\otimes y)&\geq& I_0^\varepsilon(x)+I_0^\delta(y)
\end{eqnarray*}
(following from analogous inequalities for smooth entropies) to split information measures of tensor products of states into parts.
However, there is a simpler proof which establishes Lemma~\ref{LemAsymptCost} as a corollary of Lemma~\ref{LemAsymptoticRate} directly.
It will be given below.

\begin{lemma}[Asymptotic yield/cost of conversion]
\label{LemAsymptCost}
Let $x$ and $y$ be two states of possibly different dimensionalities, let $0<\varepsilon<1$,
and choose either the trace distance or the purified distance as reference metric. Consider the following two cases,
characterized by the nonuniformity monotone $I(z):=\log d_z-H(z)$, with $H$ the Shannon entropy:
\begin{itemize}
\item Suppose that $I(x)\leq I(y)$. Let $m_n$ be the smallest integer such that
\[
   x^{\otimes n}\otimes s_1^{\otimes m_n}\epsiloncconv y^{\otimes n}.
\]
Then $\displaystyle\lim_{n\to\infty}\frac{m_n} n = I(y)-I(x)$.
\item Suppose that $I(x)\geq I(y)$. Let $m_n$ be the largest integer such that
\[
   x^{\otimes n}\epsiloncconv y^{\otimes n}\otimes s_1^{\otimes m_n}.
\]
Then $\displaystyle\lim_{n\to\infty}\frac{m_n} n = I(x)-I(y)$.
\end{itemize}
\end{lemma}
\begin{proof}
If $x$ or $y$ is uniform, the statement of the lemma follows directly from Lemma~\ref{LemAsymptoticRate}; thus, we will
assume that $I(x)>0$ and $I(y)>0$.
We give the proof of the first case, $I(x)\leq I(y)$, only; the proof of the second case is analogous.
Denote by $k_n$ the largest integer such that
\begin{equation}
   x^{\otimes n} \conv\limits^{\textnormal{$\varepsilon/2$-noisy}}
   s_1^{\otimes k_n}.
   \label{eqEps2Noisy}
\end{equation}
According to  Lemma~\ref{LemAsymptoticRate}, we have
$\lim_{n\to\infty}k_n/n = I(x)$. Now let $l_n$ be the smallest integer such that
\begin{equation}
   s_1^{\otimes l_n} \conv\limits^{\textnormal{$\varepsilon/2$-noisy}}
 y^{\otimes n}.
   \label{eqEp2Noisy2}
\end{equation}
According to  Lemma~\ref{LemAsymptoticRate}, we have $\lim_{n\to\infty} l_n/n = I(y)$.
Set $m:=\max\{l_n-k_n,0\}$. From~(\ref{eqEps2Noisy}), we obtain
\[
   x^{\otimes n}\otimes s_1^{\otimes m}\conv\limits^{\textnormal{$\varepsilon/2$-noisy}}
 s_1^{\otimes (k_n+m)}.
\]
Since $k_n+m\geq l_n$, we can first perform this conversion (and possibly discard some pure bits) and then perform conversion~(\ref{eqEp2Noisy2}), in
total yielding the conversion
\[
   x^{\otimes n}\otimes s_1^{\otimes m}\epsiloncconv y^{\otimes n}.
\]
Thus $m_n\leq m\leq |l_n-k_n|$. Dividing by $n$ and taking the lim sup of both sides yields
\[
   \limsup_{n\to\infty} \frac {m_n} n \leq I(y)-I(x).
\]
If $I(x)=I(y)$ then
\begin{equation}
   \liminf_{n\to\infty} \frac{m_n} n \geq I(y)-I(x)
   \label{eqLimInfCost}
\end{equation}
is trivially true. Thus, assume $I(x)<I(y)$, then it is clear that $m_n\to\infty$ as $n\to\infty$. Choose $\delta>0$ such that $\varepsilon+\delta<1$,
and let $i_n$ be the smallest integer such that
\[
   x^{\otimes i_n} \conv\limits^{\textnormal{$\delta$-noisy}}
 s_1^{\otimes n}.
\]
According to Lemma~\ref{LemAsymptoticRate}, we have $\lim_{n\to\infty} i_n/n=1/I(x)$.
Furthermore, let $j_n$ be the smallest integer such that
\[
   x^{\otimes j_n}\conv\limits^{\textnormal{$(\varepsilon+\delta)$-noisy}}
 y^{\otimes n}.
\]
Due to Lemma~\ref{LemAsymptoticRate}, we have $\lim_{n\to\infty} j_n/n=I(y)/I(x)$. Moreover, we have the chain of conversions
\begin{eqnarray*}
   x^{\otimes(n+i_{m_n})} &=& x^{\otimes n}\otimes x^{\otimes i_{m_n}} \conv\limits^{\textnormal{$\delta$-noisy}} x^{\otimes n}\otimes s_1^{\otimes m_n}\\
   &\epsiloncconv & y^{\otimes n},
\end{eqnarray*}
hence
\[
   x^{\otimes(n+i_{m_n})}\conv\limits^{\textnormal{$(\varepsilon+\delta)$-noisy}}
 y^{\otimes n}.
\]
By definition of $j_n$, it follows that $n+i_{m_n}\geq j_n$, and so
\[
   1+\underbrace{\frac{i_{m_n}}{m_n}}_{\to  1/I(x)}\cdot\frac{m_n}n\geq\underbrace{\frac{j_n}n}_{\to I(y)/I(x)},
\]
where the given limits are for $n\to\infty$. This establishes~(\ref{eqLimInfCost}).
\end{proof}

\subsection{Approximate catalysis}
\label{sec:ApproxCatalysis}
Generalizing the notion of exact catalysis (discussed in Subsection~\ref{Sec:Catalysis}), one may ask whether
a given state $x$ can be converted into another state $y$ \emph{approximately} via the help of a catalyst in state $z$, where the notion of approximation is the following:  the final state is $\varepsilon$-close to $y \otimes z$.  It follows that the marginal on the system of this final state might not be precisely $y$, the marginal on the catalyst of this final state might not be precisely $z$, and the system and catalyst might in fact have some correlation in the final state.
We ask under what conditions the ``approximate noisy-trumping''
\begin{equation}
      x\otimes z \epsiloncconv y\otimes z
   \label{eqEpsTrumping}
\end{equation}
is possible. This problem has recently been studied in~\cite{1ShotAtherm2}; here we give a brief summary of their results, translated into our
notation. In this subsection, the reference metric will always taken to be the trace distance, $\mathcal{D}=\mathcal{D}_\tr$, as defined in~(\ref{eqTraceDistance}).

The first insight from~\cite{1ShotAtherm2} is related to the \emph{embezzling} phenomenon~\cite{vanDamHayden}.
Suppose we fix $\varepsilon>0$, and ask for what pairs of states $x,y$ there exists some $z$ such that the conversion~(\ref{eqEpsTrumping}) is possible. The answer turns out to be:
\emph{for \textbf{all} pairs of states}. That is, the approximate noisy-trumping relation in~(\ref{eqEpsTrumping}) is trivial (every state can be converted to every other) if $\varepsilon$ is a fixed positive number which is independent
of the states and catalyst.

The trick to show this is to use a huge-dimensional, unphysical catalyst $z$ (an ``embezzling state'') which has the property that one can distill an
arbitrarily large amount of nonuniformity from it, while at the same time leaving it almost unmodified in trace distance. The intuitive reason why this is
possible is as follows: \emph{if two states $z,z'$ are
very close in the sense that $\mathcal{D}(z,z')\leq\varepsilon$, their nonuniformity content can still differ by arbitrarily large amounts.}
This is easy to see, for example, in the case of Shannon nonuniformity $I$. If $d_z=d_{z'}$, we have by the Fannes inequality~\cite{Fannes,NielsenAndChuang}, as improved by
Audenaert~\cite{Audenaert},
\[
   I(z')-I(z)=H(z)-H(z')\leq \varepsilon \log(d_z-1)+h(\varepsilon)
\]
if $\mathcal{D}(z',z)\leq \varepsilon$, where $h(\varepsilon):=-\varepsilon\log \varepsilon-(1-\varepsilon)\log (1-\varepsilon)$.
The sharpness of this bound shows that for large $d_z$, states that are close can still have significantly different amounts of nonuniformity.

However, this inequality also gives a hint on how to obtain a more interesting notion of approximate noisy-trumping: \emph{make sure that $\varepsilon\log(d_z-1)$ is bounded,
by decreasing $\varepsilon$ with $d_z$}. This prescription, interpreted as quantifying the ``error per particle'' in~\cite{1ShotAtherm2}, leads to a notion of
approximate noisy-trumping that is fully characterized by the Shannon nonuniformity $I$: 
\begin{theorem}[Thm.\ 14 in~\cite{1ShotAtherm2}]
\label{ThmApproxCatalysis}
Let $x,y$ be distributions with $d_x=d_y=:d$. If there exists a catalyst $z$ such that
\[
   x\otimes z \conv\limits^{(\varepsilon/\log d_z)\textnormal{-noisy}} y\otimes z,
\]
then
\[
   I(x)\geq I(y)-2\varepsilon - \frac{2\varepsilon\log d}{\log d_z}-h\left(\frac{2\varepsilon}{\log d_z}\right).
\]
Conversely if $I(x)>I(y)$, then for all sufficiently large $d_z$ there exists a catalyst $z$ such that
\[
      x\otimes z \deltacconv y\otimes z,\qquad \mbox{where }\delta=\exp\left(-\Omega(\sqrt{\log d_z})\right).
\]
\end{theorem}

Roughly speaking, the theorem above states that $x$ can be converted to $y$ under approximate catalytic noisy operations if and only if the Shannon nonuniformity of $x$ is bigger than that of $y$.
On the other hand, we know from Lemma~\ref{lemma:trumping} that  $x$ can be converted to $y$ under \emph{exact} catalytic noisy operations if and only if for all $p \geq 0$ the R\'enyi $p$-nonuniformity of $x$ is bigger than that of $y$.
The reason there is no conflict between these two results is that all the R\'enyi $p$-nonuniformities with $p\neq 1$ are \emph{not} asymptotically continuous.

\section{Conclusions}\label{Sec:Conclusions}

We have focussed in this article on problems of state conversion under noisy operations.  In the single-copy regime, we have studied exact and approximate state conversion, catalytic and noncatalytic. 
In each case, we have discussed the necessary and sufficient conditions for the conversion to be possible.  
These results have interesting consequences for the status of the second law of thermodynamics.

The standard formulation of the second law is that entropy does not decrease.  This is clearly inadequate as a criterion for the possibility of exact state conversion, with or without a catalyst, because evaluating the value of a single monotone can only generate a \emph{total order} over states while noisy operations, catalytic or noncatalytic, induce a quasi-order.

Fortunately, one can rehabilitate the second law as a criterion for exact state conversion.  To decide on the possibility of noncatalytic state conversion, one can compare the Lorenz curves of the states.  If one wants to make the decision by finite means, one can use either one of the two state conversion witnesses described in Section~\ref{sec:secondlaw}.  To decide on the possibility of state conversion in the presence of a catalyst, one can compare the set of order-$p$ R\'{e}nyi nonuniformities for the states.

The question of whether the second law is adequate as a criterion for  \emph{approximate} state conversion is more subtle. Theorem~\ref{ThmApproxCatalysis} shows that under a particular notion of approximate catalysis, wherein the final state of the system and the catalyst must be $\varepsilon/\log d_z$-close to the target (where $d_z$ is the dimension of the catalyst), a state $x$ can be converted to $y$ (of equal dimension) if and only if the entropy of $y$ is greater than that of $x$.  Leaving aside the case of unequal dimensions, this appears to be a vindication of the standard formulation of the second law.  This notion of approximate catalytic state conversion, however, (like the one wherein the final state is required to be $\varepsilon$-close to the target state for some fixed $\varepsilon$) 
resembles the phenomenon of embezzlement insofar as a significant amount of nonuniformity is drawn from the catalyst.

A proper notion of approximate catalytic conversion ought to have the feature that the nonuniformity required to form $y$ comes from $x$ and not from the catalyst.  The catalyst might still degrade with use, but one could require that the number of times one can reuse the catalyst is independent of the nature of the state conversions that it has facilitated.  One can, for instance, define a notion of approximate catalytic conversion as follows: the final state must be such that by consuming some additional small amount of nonuniformity, one can convert it to the desired target state of the system together with the catalyst.  Essentially, one is then using the state conversion witness  $\Lambda$ from Definition~\ref{def:stateconversionwitness} to define a metric over the states.  A similar definition can be made with the catalytic state conversion witness $\Lambda_{\rm cat}$ from Definition~\ref{DefLambda}.  This sort of notion was also proposed in~\cite{1ShotAtherm2}.  For any given finite amount of nonuniformity $\varepsilon$ that one uses to define the degree of approximation, one can always find a pair of incomparable states $x$ and $y$ such that $|\Lambda(x\|y)|\geq \varepsilon$ and $|\Lambda(y\|x)|\geq \varepsilon$, and therefore under this notion of approximate catalytic conversion, the order over states remains a quasi-order.  Necessarily then, the criterion for state conversion cannot be expressed in terms of the value of a single measure of nonuniformity, and the standard formulation of the second law is again inadequate.

Another set of results that we have described concerns the cost and yield of nonuniformity for state preparations and state conversions, both single-shot and asymptotic, as well as catalytic and noncatalytic.
We have introduced a 1-parameter family of states, the sharp states, and shown that these can provide a ``gold standard'' of nonuniformity, in terms of which we can measure the costs and yields.
Sharp states can be used to simulate erasure operations, so the nonuniformity yield of some process determines the amount of erasure, hence ``informational work'', that can be extracted from a state or a state conversion process.

Lemma~\ref{LemAsymptoticRate} implies that the rate at which pure bits must be consumed to generate $x$ is equal to the rate at which pure bits can be distilled from $x$, in the asymptotic limit, namely $I(x)$ (because the nonunifority of a pure bit is 1).  Lemmas~\ref{lemma:distillation} and \ref{lemma:formation}, on the other hand, show that in the single-shot case there is a gap:  the nonuniformity of sharp state that is needed to form a single copy of $x$ is $I_{\infty}(x)$, while one can only distill a sharp state of nonuniformity $I_0(x)$, and in general, $I_0(x) < I_{\infty}(x)$. Furthermore, Corollaries~\ref{corollary:formationcatalysis} and \ref{corollary:distilationcatalysis} demonstrate that the amount of nonuniformity one requires to form a state and the amount one can distill from a state do not change in the presence of a catalyst.

Turning to the nonuniformity costs and yields of state conversion,  Lemma~\ref{LemAsymptCost} shows that in the asymptotic limit,
 the cost per copy of taking $x$ to $y$ and the yield per copy of taking $y$ back to $x$ are both equal to $I(y)-I(x)$, hence the process is reversible.  In the single-shot case, however, there is a gap between cost and yield.  Propositions~\ref{lemma:distillationstateconversion} and \ref{lemma:formationstateconversion} imply that the nonuniformity cost of the conversion of $x$ to $y$ is $\Lambda(x\|y)$, while the nonuniformity yield of the conversion of $y$ back to $x$ is $-\Lambda(y\|x)$ (which is positive whenever $\Lambda(x\|y)$ is negative).  In general $|\Lambda(x\|y)| \ne  |\Lambda(y\|x)|$, so the conversion is not reversible.  
 
To summarize, what emerges from this analysis is that attempting state preparations or state conversions one copy at a time is inefficient.  It is only by processing asymptotically many copies at once that one can achieve perfect efficiency.

We now consider some open questions that remain in the resource theory of nonuniformity.

Although we have talked about simulating erasure operations using sharp states, there is a more general question concerning the precise resource requirements for simulating \emph{any} given non-noisy operation, or relatedly, what non-noisy operations can be simulated from a given nonuniform state.  An example of such a problem is to find the set of one-qubit quantum channels that can be implemented using an ancilla that has at most one bit of nonuniformity.   This problem has been considered by several authors~\cite{terhal1999simulating,zalka2002quantum} and a complete characterization of the class of accessible one-qubit channels is now known~\cite{Narang07}.  It remains to solve the problem beyond the case of one-qubit systems and one-qubit ancillas.

Another set of open questions concerns the situation wherein the restriction to noisy operations is combined with a locality restriction, such that an agent only has access to part of a composite system.  In this case, quantum state conversion problems \emph{do not} reduce to classical state conversion problems: entanglement between the accessible and inaccesible parts makes the problem inherently quantum.
There is some work already in this direction~\cite{QuantLandauer}.

For almost every question about nonuniformity discussed in this article, there is a simple analogue within the resource theory of athermality.  For one, it is the case that for quantum states that are block-diagonal across energy eigenspaces, every state conversion problem reduces to the corresponding problem for classical statistical states.  Furthermore, state conversion problems under classical thermal operations are simple generalizations of those same problems under noisy classical operations.  The equivalence classes of states under classical thermal operations are associated with \emph{Gibbs-rescaled} histograms~\cite{LawsOfThermo}, which plot the ratios $x_i/q_i$ in decreasing order, with $q$ the thermal state for the associated system. They can alternatively be associated with the cumulative integral of these functions, which are the athermal analogues of Lorenz curves. The curve of $x$ being everywhere greater than or equal to that of $y$ is the condition for $x \mapsto y$ under thermal operations, as shown in~\cite{FundLimitsNature}.  Functions over states that are nonincreasing under thermal operations, i.e.\ \emph{athermality monotones}, can be defined from geometric features of these curves or from Schur-convex functions relative to the thermal state.  As a consequence, for every cousin of the Shannon entropy in the zoo of entropies, there is an analogous cousin of the thermodynamic free energy.

The fact that the reduction of quantum state conversion problems to classical ones holds only for states that are block-diagonal in the energy eigenspaces is perhaps the most significant manner in which the theory of athermality must go beyond the theory of nonuniformity.  
Indeed, it remains an open problem to identify the thermal quasi-order over \emph{all} quantum states.  Note, however, that because a state that has coherence between energy eigenspaces is asymmetric relative to time translations, these questions are informed by results in the resource theory of asymmetry~\cite{Asymy1,Asymy2,Asymy3}.
 Indeed, if one has access to a reference frame for time (a clock), then by applying energy-conserving unitaries to the composite of system and clock, it is possible to simulate unitaries on the system that can rotate between the energy eigenbasis and other bases.  In this sense, a clock can act as a kind of catalyst for these state conversions~\cite{AthermalityTheory}.

We end with some general comments on the use of the framework of resource theories in the study of thermodynamics.

The standard tradition of enquiry in physics is \emph{dynamicist}: The physicist's job is to describe the natural dynamical behaviour of a system, without reference to human agents or their purposes.  There is, however, a complementary \emph{agent-centric} approach which focuses instead on characterizing limitations on an agent's control
of the behaviour of a system, or questions about what sorts of inferences an agent can make about the system~\cite{wallace2013inferential}. Thermodynamics partakes in both traditions.  One can see this easily by comparing various statements of the second law.  One that is clearly in the dynamicist tradition is Clausius's original statement: ``Heat can never pass from a colder to a warmer body without some other change, connected therewith, occurring at the same time'' \cite{Clausius}. On the other hand, the version of the Kelvin-Planck statement that is found in most textbooks is clearly agent-centric: ``it is impossible to devise a cyclically operating device, the sole effect of which is to absorb energy in the form of heat from a single thermal reservoir and to deliver an equivalent amount of work'' \cite{Rao}.

The resource-theoretic approach to thermodynamics that has been explored in this article is clearly agent-centric. The question of whether some state conversion which can be achieved by the free operations also arises under \emph{natural dynamics} (that is, in the absence of agents bringing them about) is not one that a resource theory seeks to answer. Nonetheless, the resource-theoretic approach might be adaptable to such questions. Specifically, if one wishes to characterize a certain subset of the free operations as more difficult to realize than others (on the grounds that they \emph{do not} arise naturally), then one ought to identify the features that distinguish these operations and redefine the set of free operations to exclude them.

\section{Acknowledgements}
The authors would like to thank Fernando Brand\~ao, Oscar Dahlsten, Nilanjana Datta, L\'idia del Rio, Jonathan Oppenheim and Joe Renes for discussions.   Research at Perimeter Institute is supported in part by the Government of Canada through NSERC and by the Province of Ontario through MRI.  GG and VN are supported by the Government of Canada through NSERC.
This work has been supported by the COST network.

\appendix
 
\section{Proofs}\label{app:proofs}
In the main text, the proofs of some lemmas have been deferred to the appendix. We provide those proofs here.

\vspace{2ex}\noindent\textbf{Lemma \ref{lemma:ClassifyingNO}} \textit{Noisy quantum operations are a strict subset of the unital operations and, in the case of equal dimension of input and output space, a strict superset of the mixtures of unitaries.}

\begin{proof}  We first demonstrate the inclusions, then show that they are strict. Noisy operations are necessarily unital because if the input state is completely mixed, then after adjoining an ancilla in a completely mixed state, implementing a unitary, and taking a partial trace, one is necessarily left with a completely mixed state in the output.  Mixtures of unitaries are necessarily noisy operations because one can implement them as follows. Suppose the ensemble of unitaries is $\{ p_i , U_i\}$.  One prepares an ancilla of arbitrarily large dimension in the completely mixed state, partitions its Hilbert space into subspaces, the relative dimensions of which are described by the distribution $\{ p_i \}$. Next, one implements a controlled unitary with the ancilla as the control and the system as the target, where, if the ancilla is found in the subspace $i$, the unitary $U_i$ is implemented on the system.  Such a controlled unitary is, of course, itself a unitary on the composite of system and ancilla. Finally, one discards the ancilla. 

Showing each inclusion's strictness is much harder.  An example of a unital operation that is not a noisy operation has been provided in~ \cite{HaagerupMusat}.  Meanwhile, the fact that not every noisy operation is a mixture of unitaries has been shown by Shor in~\cite{Shortalk}, making use of a result in~\cite{MendlWolf}.  We refer the reader to \cite{Shortalk} for details.
\end{proof}

\vspace{2ex}\noindent\textbf{Lemma \ref{lemma:NCO}} \textit{The set of noisy classical operations coincides with the set of uniform-preserving stochastic matrices and, in the case of equal dimension of input and output spaces (where the uniform-preserving stochastic matrices are the doubly-stochastic matrices), it coincides with the set of mixtures of permutations.}

\begin{proof} It is straightforward to see that noisy classical operations are stochastic matrices.  The proof that they are uniform-preserving proceeds in precise analogy to the proof that the noisy quantum operations are unital: if the input distribution is uniform, then by adjoining an ancilla in a uniform state and implementing a permutation, one creates a uniform state on the whole space, and every marginal is then also a uniform state.  

The converse direction, that every uniform-preserving stochastic matrix can be realized as a noisy operation, is the nontrivial one. When the input and output vector spaces are of equal dimension, the stochastic uniform-preserving matrices are doubly stochastic, and a famous result due to Birkhoff \cite{Birkhoff} establishes that every doubly stochastic matrix is achievable as a mixture of permutations.
Therefore, it suffices to show that every mixture of permutations can be realized as a noisy classical operation.  To implement permutation $P_i$ with probability $p_i$, prepare an arbitrarily large ancilla in the uniform state, and partition its sample space into subsets of relative size $p_i$. Next, implement a controlled permutation with the ancilla as the control and the system as the target, where, if the ancilla is found in the subset $i$, the permutation $P_i$ is implemented on the system.  Such a controlled permutation is, of course, itself a permutation on the composite of system and ancilla. Finally, discard the ancilla.

Finally, we can show that even if the input and output spaces have different dimensions, every uniform-preserving stochastic matrix can be achieved as a noisy classical operation.

Let $D$ be a uniform-preserving stochastic matrix from a $d_{\textnormal{in}}$-dimensional probability space to a $d_{\textnormal{out}}$-dimensional one.   We now append an ancillary system to the input and one to the output such that the two composites are of equal dimension.  That is, we define ancillary systems of dimensions $d_1$ and $d_2$ such that $d_{\textnormal{in}}d_1 = d_{\textnormal{out}}d_2$. Next, we define $D_0$ to be the $d_2 \times d_1$ matrix all of whose entries are $1/d_2$.  This is clearly just the stochastic matrix that maps every state of dimension $d_1$ to the uniform state of dimension $d_2$ and is consequently uniform-preserving.  It follows that the $(d_{\textnormal{out}}d_2) \times (d_{\textnormal{in}}d_1)$ matrix $D \otimes D_0$ is also a uniform-preserving stochastic matrix.  Given that $d_{\textnormal{out}}d_2 = d_{\textnormal{in}}d_1$, it follows that $D \otimes D_0$ is doubly-stochastic and hence can be implemented by a mixture of permutations. 

It follows that, if $D$ is any uniform-preserving stochastic matrix, then it can be implemented by first adjoining a uniform state of dimension $d_1$ (which is a noisy operation), implementing a mixture of permutations (which, as shown earlier in this proof, is a noisy operation), and finally marginalizing over the ancillary subsystem of dimension $d_2$ (which is also a noisy operation).  Given that every step of the implementation is a noisy operation, the overall operation is, as well.
\end{proof}

\vspace{2ex}\noindent\textbf{Lemma \ref{lemma:quantumNOtoclassicalNO}} \textit{
There exists a noisy quantum operation that achieves the quantum state conversion $\rho \mapsto \sigma$ if and only if there is a noisy classical operation that achieves the classical state conversion $\lambda(\rho) \mapsto \lambda(\sigma)$.
}

\begin{proof}
The possibility of the quantum state interconversion implies the existence of a quantum operation $\mathcal{E}$ such that
\begin{equation}
\sigma = \mathcal{E} (\rho).
\end{equation}
Denoting the $j$th eigenvector of $\sigma$ by $|\phi_j\rangle$ and the $k$th eigenvector of $\rho$ by $|\psi_k\rangle$, we have
\begin{equation}
\sum_k \lambda_k(\sigma) |\psi_k\rangle \langle \psi_k |  
= \mathcal{E} \left( \sum_j \lambda_j(\rho) |\phi_j\rangle \langle \phi_j |\right),
\end{equation}
or equivalently,
\begin{equation}\label{eq:lambda1}
\lambda_k(\sigma)
= \sum_j \langle \psi_k | \mathcal{E} \left( |\phi_j\rangle \langle \phi_j |\right) |\psi_k\rangle  \lambda_j(\rho).
\end{equation}
If we define the matrix $D$ by
\begin{equation}
D_{kj}= \langle \psi_k | \mathcal{E} \left( |\phi_j\rangle \langle \phi_j |\right) |\psi_k\rangle,
\end{equation}
then we can write Eq.~\eqref{eq:lambda1} as
\begin{equation}
\lambda(\sigma) =D \lambda(\rho).
\end{equation}
We seek to show that $\mathcal{E}$ is a noisy quantum operation if and only if $D$ is a noisy classical operation.  The forward implication follows from the fact that
\begin{align*}
\sum_{j=1}^{d_{\textnormal{in}}}D_{kj}  &  =\left\langle \psi_{k}\right\vert \mathcal{E}\left(
I_{\textnormal{in}} \right)  \left\vert \psi_{k}\right\rangle =\left\langle \psi_{k}\right\vert
\frac{d_{\textnormal{in}}}{d_{\textnormal{out}}}  I_{\textnormal{out}} \left\vert \psi_{k}\right\rangle =\frac{d_{\textnormal{in}}}{d_{\textnormal{out}}}\\
\sum_{k=1}^{d_{\textnormal{out}}}D_{kj}  &  =\sum_{k}\mathrm{Tr}\left[  \left\vert \psi_{k}
\right\rangle \left\langle \psi_{k}\right\vert \mathcal{E}\left(  \left\vert
\phi_{j}\right\rangle \left\langle \phi_{j}\right\vert \right)  \right] \\
 & = {\textnormal{Tr}}\left(\mathcal{E}(|\phi_j\rangle\langle\phi_j|)\right) = {\textnormal{Tr}}|\phi_j\rangle\langle\phi_j|=1,
\end{align*}
where we have used the fact that $\mathcal{E}$ is unital and the fact that
$\mathcal{E}$ is trace-preserving.
This implies that $D$ is a uniform-preserving stochastic matrix, which, by Lemma~\ref{lemma:NCO}, implies that it is a noisy classical operation.

The reverse implication follows from the fact that we can define $\mathcal{E}$ by
\begin{equation}
 \mathcal{E}(\cdot)= \sum_{k,j} D_{kj} |\psi_k \rangle \langle \phi_j | (\cdot) |\phi_j \rangle \langle \psi_k |.
\end{equation}
By assumption, $D$ is a noisy classical operation, so it can be implemented by adjoining an ancilla in the uniform state, performing a permutation on the composite, and summing over a subsystem.  Let $i$ be an index for the sample space of the ancilla, and let $l$ be an index for the sample space of the subsystem that is summed over.  Let $[R]_{kl,ji}$ denote the matrix elements of the permutation.  It follows from Eq.~\eqref{eq:NCO} that $D_{kj}= (1/d_a)\sum_{l,i} [R]_{kl,ji}$,
where we have used the fact that $[m_a]_{i}=1/d_a$, with $d_a$ the ancilla's dimension. Letting $\{ |\mu_{i}\rangle \}$ denote an orthogonal set of vectors for the ancilla $a$, and $\{ |\nu_{l}\rangle \}$ for the subsystem $\bar{a}$ that is summed over, we have
\begin{align}
&\mathcal{E}(\rho)\\
&=  \tfrac{1}{d_{a}} \sum_{k,j}  \sum_{l,i}  [R]_{kl,ji} |\psi_k \rangle  \langle \phi_j | (\rho) |\phi_j \rangle   \langle \psi_k |\\
 &= \textnormal{Tr}_{\bar{a}} \left[ \sum_{klji}  [R]_{kl,ji} |\psi_k \rangle |\nu_{l}\rangle \langle \mu_{i} | \langle \phi_j | (\rho \otimes \tfrac{I_{a}}{d_{a}} ) |\phi_j \rangle   | \mu_{i} \rangle \langle \nu_{l}|  \langle \psi_k | \right]\\
 &= \textnormal{Tr}_{\bar{a}} \left[ V (\rho \otimes \tfrac{1}{d_{a}} I_{a}) V^{\dag} \right],
\end{align}
where we have implicitly defined a unitary $V$.  In this form, it is clear that $\mathcal{E}$ is a noisy quantum operation.
\end{proof}

\vspace{2ex}\noindent\textbf{Lemma \ref{lemma:NOconviffUntlconv}} \textit{Let $\rho\in\mathcal{L}(\mathcal{H}_{\textnormal{in}})$ and $\sigma\in\mathcal{L}(\mathcal{H}_{\textnormal{out}})$. Then, the following propositions are equivalent:
\begin{itemize}
\item{(i)}  $\rho \mapsto \sigma$ by a noisy operation
\item{(ii)} $\rho \mapsto \sigma$ by a unital operation.
 \end{itemize}
If $\rho$ and $\sigma$ are of equal dimension, then (i) and (ii) are also equivalent to
\begin{itemize}
\item{(iii)} $\rho \mapsto \sigma$ by a mixture of unitaries.
 \end{itemize}}

\begin{proof}
Clearly, (i) implies (ii) because a noisy operation is unital, and (iii) implies (i) because a mixture of unitaries is a noisy operation.  The nontrivial implications are (ii) to (i) (for unequal dimension) and (ii) to (iii) (for equal dimensions).  The proof that (ii) implies (i) can be inferred from the proof of Lemma~\ref{lemma:quantumNOtoclassicalNO}, as follows.  In the forward direction of that proof, we inferred from the existence of a noisy quantum operation $\mathcal{E}$ taking $\rho$ to $\sigma$ that there was a noisy classical operation taking the spectrum of $\rho$ to the spectrum of $\sigma$.  But that inference only relied on the fact that $\mathcal{E}$ was unital.  
Then, in the reverse direction of that proof, we showed that from a noisy classical operation taking the spectrum of $\rho$ to the spectrum of $\sigma$, one can construct a noisy quantum operation taking $\rho$ to $\sigma$.  Therefore, (ii) implies (i).  In the case of equal dimension, we simply note that the noisy classical operation can be written as a mixture of permutations (by Birkhoff's Theorem \cite{Birkhoff}), and hence the noisy quantum operation that one constructs from it is a mixture of unitaries.  Therefore, (ii) implies (iii).
\end{proof}

\vspace{2ex}\noindent\textbf{Lemma \ref{LemMaxQC}} \textit{Let $\mathcal{D}$ be any contractive metric on the quantum states for which there exists a norm $\|\cdot\|$ on the
self-adjoint matrices such that $\mathcal{D}(\rho,\sigma)=\|\rho-\sigma\|$ (for example the trace distance,
$\mathcal{D}=\mathcal{D}_\tr$ with $\|\cdot\|=\frac 1 2 \|\cdot\|_1$).
Then the quantum version of the smooth max-entropy,
\[
   S_0^\varepsilon(\rho):=\min_{\bar\rho:\,\,\mathcal{D}(\rho,\bar\rho)\leq\varepsilon} S_0(\bar\rho),
\]
with $S_0(\bar\rho):=\log \rank(\bar\rho)$, agrees with the classical smooth max-entropy on the spectrum of $\rho$:
\[
   S_0^\varepsilon(\rho)=H_0^\varepsilon(\lambda(\rho)).
\]}

\begin{proof}
In the following, we will choose a basis such that $\rho$ is diagonal; in particular, $\rho={\rm diag}(r_1,r_2,\ldots,r_n)$ with $r_1\geq r_2\geq \ldots\geq r_n$.
Let $s$ be any distribution such that $\mathcal{D}(\lambda(\rho),s)\leq\varepsilon$ and $H_0^\varepsilon(\lambda(\rho))=H_0(s)$.
Let $\sigma:={\rm diag}(s_1,s_2,\ldots,s_n)$, then $\mathcal{D}(\rho,\sigma)=\mathcal{D}(\lambda(\rho),s)\leq\varepsilon$. Thus
\[
   S_0^\varepsilon(\rho)\leq S_0(\sigma)=H_0(s)=H_0^\varepsilon(\lambda(\rho)).
\]
To prove the converse inequality, let $\sigma$ be any quantum state with $\mathcal{D}(\rho,\sigma)\leq\varepsilon$ and
$S_0^\varepsilon(\rho)=S_0(\sigma)$. Denote by $s=(s_1,s_2,\ldots,s_n)$ the eigenvalues of $\sigma$ in
non-increasing order, i.e.\ $s_1\geq s_2\geq\ldots\geq s_n$. Contractivity of the metric implies unitary invariance of the norm.
Thus, it follows from~\cite[Ineq. (IV.62)]{Bhatia} that $\mathcal{D}(\lambda(\rho),s)\leq\mathcal{D}(\rho,\sigma)\leq\varepsilon$. Hence
\[
   H_0^\varepsilon(\lambda(\rho))\leq H_0(s)=S_0(\sigma)=S_0^\varepsilon(\rho).
\]
This proves the claim.
\end{proof}

Unfortunately, the purified distance does not satisfy the premise of Lemma~\ref{LemMaxQC}: it does not come from a norm.

\vspace{2ex}\noindent\textbf{Lemma \ref{LemH0eps}} \textit{If the trace distance $\mathcal{D}_\tr$ is chosen as reference metric, then any distribution $x$ has smooth R\'enyi $0$-entropy
\[
   H_0^\varepsilon(x)=\log k,
\]
where $k\geq 1$ is chosen as the smallest integer such that $\sum_{i=1}^k x_i^\downarrow\geq 1-\varepsilon$, with $x_1^\downarrow\geq x_2^\downarrow\geq\ldots$
denoting the components of $x$ in non-increasing order.}

\begin{proof}
Since $H_0^\varepsilon(x)$ is invariant with respect to permutations of the components of $x$, we may assume that the components of $x$
are ordered, such that $x_i=x_i^\downarrow$. For a given state $x=(x_1,\ldots,x_n)$, define $k$ as described above, and set $\mathcal{N}:=\sum_{i=1}^k x_i\geq 1-\varepsilon$.
Set $\tilde x:=(x_1,\ldots,x_k,0,\ldots,0)\in\mathbb{R}^n$, then $x':=\tilde x/\mathcal{N}$ is a state. A simple calculation shows that $\mathcal{D}_\tr(x,x')=1-\mathcal{N}\leq\varepsilon$,
thus $H_0^\varepsilon(x)\leq H_0(x')=\log k$.

Conversely, let $y\in\mathbb{R}^n$ be any state with $H_0(y)<\log k$, then consider the support of $y$, $\supp(y):=\{i\in\{1,\ldots,n\}\,\,|\,\, y_i\neq 0\}$.  This has cardinality
$|\supp(y)|<k$, thus $\sum_{i\in \supp(y)} x_i<1-\varepsilon$ by the construction of $k$.  Recalling an alternative definition of the classical trace distance~[\cite{NielsenAndChuang}, eq.~(9.4)] and comparing the values of $x$ and $y$ on the kernel of $y$, $\ker (y)$, we get
\begin{eqnarray*}
   \mathcal{D}_\tr(x,y) &= &  \sup_{\mathcal{S}} \left|\sum_{i \in \mathcal{S}}x_i - \sum_{i \in \mathcal{S}}y_i \right| \\
  &\geq  &  \left|\sum_{i \in \ker (y)}x_i - \sum_{i \in \ker (y)}y_i \right| \\
   &=& 1-\sum_{i\in \supp (y)}x_i >\varepsilon.
\end{eqnarray*}

This shows that $H_0^\varepsilon(x)\geq \log k$, proving the claim.
\end{proof}

In~\cite{Tomamichel}, it is shown that the purifed distance upper-bounds the trace distance:
\begin{equation}
   \mathcal{D}_\tr(x,y)\leq \mathcal{D}_p(x,y).
   \label{eqBoundDistances}
\end{equation}
This is the final ingredient to prove Lemma~\ref{LemAEP} from the main text, the asymptotic equipartition property:

\vspace{2ex}\noindent\textbf{Lemma \ref{LemAEP}} \textit{If the the reference metric is taken to be either the \emph{purified distance} $\mathcal{D}_p$ or the \emph{trace distance} $\mathcal{D}_\tr$, then we have for every $0<\varepsilon<1$
\begin{eqnarray*}
   \lim_{n\to\infty}\frac 1 n H_\infty^\varepsilon(x^{\otimes n}) &=& H(x)\\ 
   \lim_{n\to\infty}\frac 1 n H_0^\varepsilon(x^{\otimes n}) &=& H(x)
\end{eqnarray*}
for all finite probability distributions $x$, where $H(x)=-\sum_i x_i\log x_i$ is the Shannon entropy.}

\begin{proof}
For the sake of this proof, let $H_\infty^\varepsilon$ and $H_0^\varepsilon$ be the smooth entropies defined with
respect to the trace distance $\mathcal{D}_\tr$, and denote by $\bar H_\infty^\varepsilon$ and $\bar H_0^\varepsilon$ the smooth entropies defined with
respect to the purified distance $\mathcal{D}_p$.
Denoting the $\varepsilon$-ball around some state $x$ according to a metric $\mathcal{D}$ by
\[
   B^\varepsilon_{\mathcal{D}}(x):=\{y\,\,|\,\,\mathcal{D}(x,y)\leq\varepsilon\},
\]
inequality~(\ref{eqBoundDistances}) implies that
\begin{equation}
   B^\varepsilon_{\mathcal{D}_p}(x)\subseteq B^\varepsilon_{\mathcal{D}_\tr}(x).
   \label{eqUnitBalls}
\end{equation}
We show our lemma directly via the asymptotic equipartition property as given, for example, in Theorem\ 3.1.2 of \cite{CoverThomas}. Given any distribution $x$, define the \emph{$\varepsilon$-typical set $A_\varepsilon^{(n)}$}
of all sequences $s=(s_1,\ldots,s_n)$ of length $n$ as
\[
   A_\varepsilon^{(n)}:=\{s\,\,|\,\, 2^{-n(H(x)+\varepsilon)}\leq x^{\otimes n}(s)\leq 2^{-n(H(x)-\varepsilon)}\}.
\]
Then we can find a sequence $(\varepsilon_n)_{n\in\mathbb{N}}$ with $\varepsilon_n\stackrel{n\to\infty}\longrightarrow 0$
such that the sequence of sets $A^{(n)}:=A^{(n)}_{\varepsilon_n}$ satisfies
\begin{eqnarray*}
   x^{\otimes n}(A^{(n)})&>& 1-\varepsilon_n,\\
   (1-\varepsilon_n)2^{n(H(x)-\varepsilon_n)}&\leq& |A^{(n)}| \leq 2^{n(H(x)+\varepsilon_n)}.
\end{eqnarray*}
Let $\varepsilon\in(0,1)$ be arbitrary, and fix any $\delta>0$. We will now prove the following claim:
if $n$ is large enough, then any distribution $q^{(n)}$ on $n$ symbols with support on some set $Q^{(n)}$,
where $|Q^{(n)}|\leq 2^{n(H(x)-\delta)}$, has $\mathcal{D}_\tr (q^{(n)},x^{\otimes n})>\varepsilon$.
If we have a distribution $q^{(n)}$ with these properties, then
\[
   q^{(n)}(A^{(n)}\setminus Q^{(n)})=0,
\]
but
\begin{eqnarray*}
   x^{\otimes n}(A^{(n)}\setminus Q^{(n)})&=& x^{\otimes n}(A^{(n)})-x^{\otimes n}(A^{(n)}\cap Q^{(n)})\\
   &>& 1-\varepsilon_n - 2^{-n(H(x)-\varepsilon_n)}\cdot |Q^{(n)}|\\
   &\geq& 1-\varepsilon_n - 2^{-n(\delta-\varepsilon_n)}\\
   &\stackrel{n\to\infty}\longrightarrow& 1.
\end{eqnarray*}
Since the trace distance gives the maximal possible probability with which two distributions can be distinguished,
this implies that $\mathcal{D}_\tr (q^{(n)},x^{\otimes n})>\varepsilon$ for all distributions $q^{(n)}$ of this kind if $n$ is large enough.
Thus, denoting by $\supp\, q^{(n)}$ the set of $s$ with $q^{(n)}(s)>0$, we obtain
\[
   \mathcal{D}_\tr(x^{\otimes n},q^{(n)})\leq \varepsilon \Rightarrow |\supp\, q^{(n)}|>2^{n(H(x)-\delta)}.
\]
if $n$ is large enough, which also implies that $H_0(q^{(n)})>n(H(x)-\delta)$, and so
\[
   H_0^\varepsilon(x^{\otimes n})>n(H(x)-\delta).
\]
Consequently,
\[
   \liminf_{n\to\infty} \frac 1 n H_0^\varepsilon(x^{\otimes n})\geq H(x)-\delta.
\]
Since this is true for every $\delta>0$, we obtain
\begin{equation}
   \liminf_{n\to\infty} \frac 1 n H_0^\varepsilon(x^{\otimes n})\geq H(x).
   \label{eqLimInfHMax}
\end{equation}
Furthermore, the inequality
\begin{eqnarray*}
   \bar H_0^\varepsilon(y)&=&\min_{\bar y\in B^\varepsilon_{\mathcal{D}_p}(y)} H_0(\bar y)\\
   &\geq& \min_{\bar y \in B^\varepsilon_{\mathcal{D}_\tr}(y)} H_0(\bar y)=H_0^\varepsilon(y)
\end{eqnarray*}
implies that
\begin{equation}
   \liminf_{n\to\infty} \frac 1 n \bar H_0^\varepsilon(x^{\otimes n})\geq H(x).
   \label{eqLimInfHbarMax}
\end{equation}
Analogously, suppose that $r^{(n)}$ is a distribution that satisfies one of the two
equivalent conditions
\[
   \max_s r^{(n)}(s)\leq 2^{-n(H(x)+\delta)}\Leftrightarrow H_\infty(r^{(n)})\geq n(H(x)+\delta),
\]
where the maximum is over all sequences $s$ of length $n$. Then it follows
\begin{eqnarray*}
   r^{(n)}(A^{(n)})&\leq& 2^{-n(H(x)+\delta)}\cdot |A^{(n)}|\leq 2^{-n(\delta-\varepsilon_n)}
\end{eqnarray*}
which tends to zero for $n\to\infty$, while $x^{\otimes n}(A^{(n)})$ tends to one. This shows that
$\mathcal{D}_\tr(x^{\otimes n},r^{(n)})>\varepsilon$ for all distributions $r^{(n)}$ of this kind if $n$ is large enough. Thus,
\[
   \mathcal{D}_\tr(x^{\otimes n},r^{(n)})\leq\varepsilon\Rightarrow H_\infty(r^{(n)})<n(H(x)+\delta)
\]
if $n$ is large enough, and so
\[
   H_\infty^\varepsilon(x^{\otimes n})<n(H(x)+\delta).
\]
Consequently,
\[
   \limsup_{n\to\infty}\frac 1 n H_\infty^\varepsilon(x^{\otimes n})\leq H(x)+\delta.
\]
Since this is true for every $\delta>0$, we obtain
\begin{equation}
   \limsup_{n\to\infty}\frac 1 n H_\infty^\varepsilon(x^{\otimes n})\leq H(x).
   \label{eqLimSupHMin}
\end{equation}
Furthermore, the inequality
\begin{eqnarray*}
   \bar H_\infty^\varepsilon(y)&=&\max_{\bar y \in B^\varepsilon_{\mathcal{D}_p}} H_\infty(\bar y)\\
   &\leq& \max_{\bar y \in B^\varepsilon_{\mathcal{D}_\tr}} H_\infty(\bar y)=H_\infty^\varepsilon(y)
\end{eqnarray*}
implies that
\begin{equation}
   \limsup_{n\to\infty} \frac 1 n \bar H_\infty^\varepsilon(x^{\otimes n})\leq H(x).
   \label{eqLimSupHbarMin}
\end{equation}
Now define a distribution $q^{(n)}$ by
\[
   q^{(n)}(s):=\left\{
      \begin{array}{cl}
         x^{\otimes n}(s)/x^{\otimes n}(A^{(n)}) & \mbox{if }s\in A^{(n)} \\
         0 & \mbox{otherwise}.
      \end{array}
   \right.
\]
A simple calculation shows that
\[
   \mathcal{D}_\tr(q^{(n)},x^{\otimes n})=\frac 1 2 \left(\strut 1-x^{\otimes n}(A^{(n)})\right)\stackrel{n\to\infty}\longrightarrow 0,
\]
thus
\[
   H_0^\varepsilon(x^{\otimes n})\leq H_0(q^{(n)})=\log |A^{(n)}|\leq n(H(x)+\varepsilon_n)
\]
if $n$ is large enough. Thus
\begin{equation}
   \limsup_{n\to\infty} \frac 1 n H_0^\varepsilon(x^{\otimes n})\leq H(x).
   \label{eqLimSupHMax}
\end{equation}
Similarly,
\begin{eqnarray*}
   H_\infty^\varepsilon(x^{\otimes n})&\geq& H_\infty(q^{(n)})\\
   &=& \log x^{\otimes n}(A^{(n)})-\log \max_{s\in A^{(n)}} x^{\otimes n}(s)\\
   &>& \log(1-\varepsilon_n)-\log 2^{-n(H(x)-\varepsilon_n)}.
\end{eqnarray*}
This shows that
\begin{equation}
   \liminf_{n\to\infty} \frac 1 n H_\infty^\varepsilon(x^{\otimes n})\geq H(x).
   \label{eqLimInfHMin}
\end{equation}
Furthermore, an elementary calculation shows that
\[
   \mathcal{D}_p(q^{(n)},x^{\otimes n})=\sqrt{1-x^{\otimes n}(A^{(n)})}\stackrel{n\to\infty}\longrightarrow 0,
\]
and repeating the calculations above for the purified distance yields the inequalities
\begin{eqnarray}
   \label{eqLimSupHbarMax}
   \limsup_{n\to\infty}\frac 1 n \bar H_0^\varepsilon(x^{\otimes n})&\leq& H(x),\\
   \liminf_{n\to\infty}\frac 1 n \bar H_\infty^\varepsilon(x^{\otimes n})&\geq& H(x).\label{eqLimInfHbarMin}
\end{eqnarray}
The inequalities~(\ref{eqLimInfHMax}), (\ref{eqLimInfHbarMax}), (\ref{eqLimSupHMin}), (\ref{eqLimSupHbarMin}), 
(\ref{eqLimSupHMax}), (\ref{eqLimInfHMin}), (\ref{eqLimSupHbarMax}), and~(\ref{eqLimInfHbarMin}) prove the lemma.
\end{proof}

\vspace{2ex}\noindent\textbf{Lemma \ref{LemAsymptoticRateVanishing}} \textit{There is a protocol for transforming $n$ copies of $x$ into $m_n$ copies of $y$ with asymptotically vanishing error at rate
\[
   r:=\lim_{n\to\infty}\frac {m_n} n =\frac{I(x)}{I(y)};
\]
however, no higher rate is achievable by any protocol of this kind.}

\begin{proof}
Since every protocol with asymptotically vanishing error is a special case of a protocol with constant error, Lemma~\ref{LemAsymptoticRate} implies the second part
of the statement, i.e.\ that no higher rate than $r:=I(x)/I(y)$ is achievable.

Now we show how to construct a protocol with asymptotically vanishing error from a protocol with constant error. For every constant error $\varepsilon>0$,
denote the corresponding integer $m_n$ from Lemma~\ref{LemAsymptoticRate} by $m_n(\varepsilon)$. For every fixed $\varepsilon>0$, we have
\[
   \lim_{n\to\infty}\frac{m_n(\varepsilon)}n = r.
\]
Thus, for every integer $k\in\mathbb{N}$, there exists an integer $\tilde n_k$ such that
\begin{equation}
   \left| \frac{m_n(1/k)} n - r \right| < \frac 1 k
   \label{eqClose}
\end{equation}
for all $n\geq \tilde n_k$.
Let $n_k:=k+\max_{1\leq\ell\leq k} \tilde n_\ell$, then $n_k$ is strictly increasing in $k$, and~(\ref{eqClose}) holds in particular for all $n\geq n_k$.

The protocol now is as follows. Given any integer $n$, find the $k$ such that $n\in[n_k,n_{k+1}-1]$, which is possible if $n$ is large enough, i.e.\ $n\geq n_{k=1}$.
Then convert $x^{\otimes n}$ into a state $y'$ which is $(1/k)$-close to $y^{\otimes m_n(1/k)}$, which is possible by definition of the sequence
$m_n(\varepsilon)$. That is, we convert $n$ copies of $x$ approximately into $m_n$ copies of $y$, where $m_n=m_n(1/k)$. Since
\[
   \frac {m_n} n = \frac{m_n(1/k)} n \in \left( r-\frac 1 k,r+\frac 1 k\right)
\]
and $k\to\infty$ as $n\to\infty$, we have $\lim_{n\to\infty}m_n/n=r$. Furthermore, the error $\varepsilon_n=1/k$ tends to zero as $n\to\infty$.
\end{proof}

\section{Strong noisy-trumping and $p<0$ R\'enyi nonuniformities}
\label{Sec:appnoisytrumping}

In this section, as announced in Subsection~\ref{Sec:Catalysis}, we study a notion of noisy-trumping which is an alternative to Definition~\ref{DefNoisyTrumping}.
We do this here in the appendix because this alternative definition turns out to have undesirable mathematical properties, rendering its physical interpretation problematic.
However, we still think that the discussion may be interesting within the broader question of how to define a proper notion of trumping, and also because it has
been mentioned in~\cite{1ShotAtherm2}.

In the discussion of Subsection~\ref{Sec:Catalysis}, the first idea for defining a notion of noisy-trumping that removed the discontinuity mentioned there was
to relax the demand of producing the target state perfectly, and instead \emph{allow the target state to be obtained to arbitrary accuracy.}
The following definition formalizes this idea.
In contrast to Definition~\ref{DefNoisyTrumping}, it will require the input state to be prepared perfectly.
\begin{definition}[Strong noisy-trumping]
\label{DefStrongNoisyTrumping}
We say that a state $x$ \emph{strongly noisy-trumps} another state $y$ if and only if there is a sequence of state $(y_n)_{n\in\mathbb{N}}$, all of
the same dimension as $y$, which converges to $y$, i.e.\ $\lim_{n\to\infty} y_n=y$,
and a sequence of catalysts $(z_n)_{n\in\mathbb{N}}$ such that
\[
   x\otimes z_n\cconv y_n\otimes z_n\mbox{ for every } n\in\mathbb{N}.
\]
\end{definition}
Note that the sequence of catalysts does not have to converge; moreover, the catalysts $z_n$ may have different dimensionalities ($d_{z_n}$ may even
grow unboundedly).

This is indeed a stronger notion of trumping: if $x$ strongly noisy-trumps $y$, then $x$ also noisy-trumps $y$ in the sense of Definition~\ref{DefNoisyTrumping}.
This is because if $x$ strongly noisy-trumps $y$, then (according to the definition above) $x$ noisy-trumps every $y_n$, and thus due to Corollary~\ref{CorTrumpingClosedness}
$x$ also noisy-trumps $y$. Below, we will see that noisy-trumping and strong noisy-trumping are strictly different notions of trumping; this will follow from the fact
that the corresponding quasi-orders are characterized by different complete sets of monotones.

We start by giving an analog of Lemma~\ref{lemma:trumping} for strong noisy-trumping. The proof again relies on the results by~\textcite{klimesh}
and~\textcite{Turgut}.

\begin{lemma}[Conditions for strong noisy-trumping]
\label{LemStrongNoisyTrumping}
Let $x$ and $y$ be any pair of states. Then the strong noisy-trumping relation is characterized by two cases:
\begin{itemize}
\item[(i)] \emph{The state $x$ does not contain zeros.} Then $x$ strongly noisy-trumps $y$ if and only if
\[
   I_p(x)\geq I_p(y)\mbox{ for all }p\in\mathbb{R}.
\]
\item[(ii)] \emph{The state $x$ contains zeros.} Then $x$ strongly noisy-trumps $y$ if and only if
\[
   I_p(x)\geq I_p(y)\mbox{ for all }p\geq 0.
\]
\end{itemize}
\end{lemma}
\begin{proof}
Suppose that $x$ strongly noisy-trumps $y$; that is, there is a sequence $(y_n)_{n\in\mathbb{N}}$ with $\lim_{n\to\infty} y_n=y$ and a sequence
of catalysts $(z_n)_{n\in\mathbb{N}}$ such that $x\otimes z_n\cconv y_n\otimes z_n$.
Since every $I_p$ is a nonuniformity monotone, we get
\[
   I_p(x)+I_p(z_n)=I_p(x\otimes z_n)\geq I_p(y_n\otimes z_n)=I_p(y_n)+I_p(z_n).
\]
If $I_p(z_n)$ is finite, then we can cancel it from both sides of the inequality and we are done. The only possibility for $I_p(z_n)$ to \emph{not} be finite is for $p<0$
and for $z_n$ to have components that are zero.  However, by Corollary~\ref{truncatecatalyst}, if $z_n$ catalyzes the conversion of $x$ to $y_n$
then so does $z_n^{\langle I_0(z_n)\rangle}$ and the latter has full support. One can therefore repeat the argument for $z_n^{\langle I_0(z_n)\rangle}$.

Thus $I_p(x)\geq I_p(y_n)$. If $p\neq 0$, then $I_p$ is continuous, so $y_n\to y$ for $n\to\infty$ implies $I_p(y_n)\to I_p(y)$ (including
the possibility that $I_p(y)=\infty$). Thus, we can take the limit of the inequality and conclude
\[
   I_p(x)\geq I_p(y)\mbox{ for all }p\in\mathbb{R}\setminus\{0\}.
\]
Finally, we can take the limit $p\searrow 0$ of the latter inequality, which proves that it also holds for $p=0$.

Now suppose that $x$ does not contain zeros, and that $I_p(x)\geq I_p(y)$ for all $p\in\mathbb{R}$. We may assume that $d_x=d_y$ for the same reasons as were given in the proof of Lemma~\ref{lemma:trumping}.
We may also assume that $y$ is not maximally mixed.
We have to show that $x$ strongly noisy-trumps $y$; we will do so by using
the results of Lemma~\ref{LemKlimesh}. Translating our inequality for $I_p$ into the notation of this lemma, we obtain that $f_p(x)\geq f_p(y)$ for all $p\neq 0$.
If $y$ contained zeros, then we would have $I_p(y)=\infty$ and $I_p(x)<\infty$ for $p<0$ which would contradict $I_p(x)\geq I_p(y)$. Thus $y$ has no zeros either. Thus,
we can use the limit
\[
   \lim_{p\searrow 0} \frac 1 p I_p(z)=-\log d_z - \frac 1 {d_z}\sum_{i=1}^{d_z}\log z_i
   \equiv -\log d_z+\frac 1 {d_z} f_0(z)
\]
to conclude that $f_0(x)\geq f_0(y)$ (again, the same trick as was used in~\cite{1ShotAtherm2}).

For $n\in\mathbb{N}$, set $y_n:=(1-1/n)y+(1/n) m^{(d_y)}$, then $y_n$ does not contain zeros. Since all $f_p$ are strictly Schur-concave, we have
$f_p(y_n)<f_p(y)$, and so
\[
   f_p(x)>f_p(y_n)\mbox{ for all }p\in\mathbb{R}.
\]
Lemma~\ref{LemKlimesh} implies that there exists some catalyst $z_n$ such that $x\otimes z_n \cconv y_n\otimes z_n$.
Since this works for every $n\in\mathbb{N}$ and $\lim_{n\to\infty}y_n=y$, this shows that $x$ strongly noisy-trumps $y$.

Finally, suppose that $x$ contains zeros, and define $y_n$ as above. If $p\leq 0$ then $\infty=f_p(x)>f_p(y_n)$ is automatically satisfied,
and we only need to prove that $f_p(x)>f_p(y_n)$ for $p>0$, which follows from $I_p(x)\geq I_p(y)$ for $p> 0$.
\end{proof}

This result motivates the following definition.
\begin{definition}
\label{DefLambdaMinus}
For an arbitrary pair of states $x$ and $y$, define
\[
   \Lambda_{\rm cat}^{\rm strong}(x\|y) := \inf_{p\in S_x}\left(\strut I_p(x)-I_p(y) \right),
\]
where
\[
   S_x:=\left\{
      \begin{array}{cl}
      \mathbb{R}^+ & \mbox{ if $x$ contains zeros},\\
      \mathbb{R} & \mbox{otherwise}.
      \end{array}
   \right.
\]
\end{definition}

If $x$ contains zeros, then $\Lambda_{\rm cat}^{\rm strong}(x\|y)=\Lambda_{\rm cat}(x\|y)$; in general, we have the inequality
\[
   \Lambda_{\rm cat}^{\rm strong}(x\|y)\leq \Lambda_{\rm cat}(x\|y).
\]
The following proposition is a reformulation of Lemma~\ref{LemStrongNoisyTrumping} in terms of this new quantity:

\begin{proposition}
\label{PropStrongTrumpingLambda}
For an arbitrary pair of states $x$ and $y$, $x$ strongly noisy-trumps $y$ if and only if
\[
  \Lambda_{\textnormal{cat}}^{\rm strong}(x\| y)\ge0.
\]
Equivalently, the function $\Lambda_{\textnormal{cat}}^{\rm strong}(x\| y)$ is a complete witness for strong noisy-trumping.
\end{proposition}

The nonuniformity of formation and the distillable nonuniformity (as determined in Proposition~\ref{lemma:distillation}
and~\ref{lemma:formation}) are not changed by catalysis in the sense of strong noisy-trumping.
In other words, Corollaries~\ref{corollary:formationcatalysis} and \ref{corollary:distilationcatalysis} also hold for strong noisy-trumping and the proofs are exactly analogous.  The move from noisy-trumping to strong noisy-trumping \emph{does} make a difference, however, for the nonuniformity cost and yield of state conversion.

\begin{proposition}
If $\Lambda_{\rm cat}^{\rm strong}(x\|y) > 0$, then
$x$ strongly noisy-trumps $y \otimes s_I$ if and only if
$$I \le \Lambda_{\textnormal{cat}}^{\rm strong}(x\|y).$$
\end{proposition}
\begin{proof}
The case $I=0$ is just Proposition~\ref{PropStrongTrumpingLambda}; thus, we may assume $I>0$.
First, suppose that $x$ does not contain any zeros. But then
$0=I_0(x)<I_0(y\otimes s_I)$, and so $x$ cannot strongly noisy-trump $y\otimes s_I$. In fact,
\begin{eqnarray*}
   \Lambda_{\rm cat}^{\rm strong}(x\|y)&=&\inf_{p\in\mathbb{R}}\left(\strut I_p(x)-I_p(y)\right)\leq I_0(x)-I_0(y)\\
   &=&-I_0(y)\leq 0 <I.
\end{eqnarray*}
Now consider the case that $x$ contains zeros. Then
\begin{eqnarray*}
   \Lambda_{\rm cat}^{\rm strong}(x\|y\otimes s_I) &=&\inf_{p\geq 0}\left(\strut I_p(x)-I_p(y\otimes s_I)\right)\\
   &=&\inf_{p\geq 0}\left(\strut I_p(x)-I_p(y)\right)-I \\
   &=&\Lambda_{\rm cat}^{\rm strong}(x\|y)-I.
\end{eqnarray*}
According to Proposition~\ref{PropStrongTrumpingLambda}, $x$ trumps $y\otimes s_I$ if and only if this expression is non-negative, that is,
if any only if $\Lambda_{\rm cat}^{\rm strong}(x\|y)\geq I$.
\end{proof}

Surprisingly, the nonuniformity cost of ``strong'' catalytic state conversion is not determined by $\Lambda_{\rm cat}^{\rm strong}$, but by the
quantity $\Lambda_{\rm cat}$ from Definition~\ref{DefLambda}:

\begin{proposition}\label{lemma:formationstrongcatstateconversion}
If $\Lambda_{\rm cat}(x\|y) < 0$, then
$x \otimes s_I$ strongly noisy-trumps $y$ if and only if
$$I \ge -\Lambda_{\textnormal{cat}}(x\|y).$$
\end{proposition}
\begin{proof}
Since $x\otimes s_I$ contains zeros, we have
\begin{eqnarray*}
   \Lambda_{\rm cat}^{\rm strong}(x\otimes s_I\| y) &=& \inf_{p\geq 0}\left(\strut I_p(x\otimes s_I)-I_p(y)\right) \\
   &=& I+\inf_{p\geq 0}\left(\strut I_p(x)-I_p(y)\right) \\
   &=& I+\Lambda_{\rm cat}(x\|y).
\end{eqnarray*}
According to Proposition~\ref{PropStrongTrumpingLambda}, $x\otimes s_I$ strongly noisy-trumps $y$ if and only if this expression is non-negative,
that is, if and only if $I+\Lambda_{\rm cat}(x\|y)\geq 0$.
\end{proof}

In the case where both $x$ and $y$ do not contain zeros, such that $\Lambda_{\rm cat}^{\rm strong}(x\|y)< \Lambda_{\rm cat}(x\|y)$ is possible,
Propositions~\ref{PropStrongTrumpingLambda} and \ref{lemma:formationstrongcatstateconversion} show
an undesirable discontinuity in the nonuniformity
cost of catalytic state conversion in the case of strong noisy-trumping: $x\otimes s_I$ strongly noisy-trumps $y$ if and only if
\[
   \left\{
      \begin{array}{cl}
         I\geq  -\Lambda_{\rm cat}^{\rm strong}(x\|y) & \mbox{if }I=0 \\
         I \geq -\Lambda_{\rm cat}(x\|y) & \mbox{if }I>0.
      \end{array}
   \right.
\]
A discontinuity like this was not present in the definition of noisy trumping used in the main text (Definition~\ref{DefNoisyTrumping}).

As a final observation, Proposition~\ref{lemma:formationstrongcatstateconversion} implies that two different possible definitions of noisy trumping are equivalent:
\begin{lemma}
Let $x$ and $y$ be any pair of states. Then
\begin{equation}
   x\otimes s_\delta\mbox{ strongly noisy-trumps }y\mbox{ for all }\delta>0
   \label{eqThirdWay}
\end{equation}
if and only if $x$ noisy-trumps $y$.
\end{lemma}
\begin{proof}
First, consider the case that $\Lambda_{\rm cat}(x\|y)<0$. According to Proposition~\ref{PropTrumpingLambda},
this means that $x$ does not noisy-trump $y$. Indeed, if $\delta$ is some logarithm of a rational number such that $0<\delta<-\Lambda_{\rm cat}(x\|y)$,
then Proposition~\ref{lemma:formationstrongcatstateconversion} proves that $x\otimes s_\delta$ does not strongly noisy-trump $y$, so statement~(\ref{eqThirdWay})
is false. This proves the lemma for the case $\Lambda_{\rm cat}(x\|y)<0$.

In the remaining case $\Lambda_{\rm cat}(x\|y)\geq 0$, Proposition~\ref{PropTrumpingLambda} implies that $x$ noisy-trumps $y$.
By definition, for every $\delta>0$, there is a catalyst $z$ such that $x\otimes s_\delta\otimes z \cconv y\otimes z$. According to Definition~\ref{DefStrongNoisyTrumping},
this means that $x\otimes s_\delta$ strongly noisy-trumps $y$. This shows that~(\ref{eqThirdWay}) is true in this case, and completes the proof of the lemma.
\end{proof}

This lemma shows that there is no third alternative definition of noisy trumping. A conceivable alternative definition would involve the possibility to consume
an arbitrarily small amount of nonuniformity (i.e.\ $s_\delta$ for arbitrarily small $\delta>0$) \emph{and} at the same time allowing for non-perfect state production,
i.e.\ conversion to arbitrary non-perfect accuracy. Formally, this would be~(\ref{eqThirdWay}); the lemma shows, however, that this only reproduces the
notion of noisy-trumping from the main text. This result is not surprising due to the ``closedness'' property of noisy-trumping which was proven
in Corollary~\ref{CorTrumpingClosedness}.

\section{Approximate distillation}

The goal of this section is to give an expression for the maximal nonuniformity of any sharp state that can be extracted from some given state $x$,
as announced in Eq.~(\ref{eqOptimalMainText}) in the main text. The result in Lemma~\ref{LemMinMaxMonotonicity} was achievable, but not optimal.
To this end, we need a few lemmas.

\begin{lemma}
\label{LemExtension}
Let $\mathcal{D}$ be any contractive metric.
Consider a bipartite classical system $AB$ of finite dimension $d_{AB}=d_A d_B$. Let $x^{AB}$ be any state on $AB$, with marginal $x^A$ on $A$,
and $y^A$ any other state on $A$. Then there exists an extension $y^{AB}$ of $y^A$ (i.e.\ a state that has $y^A$ as its marginal) with the property
that
\[
   \mathcal{D}(y^{AB},x^{AB})=\mathcal{D}(y^A,x^A).
\]
\end{lemma}
\textbf{Remark.} In the quantum case, this statement is still true for the purified distance $\mathcal{D}_p$~\cite{Tomamichel}, but the proof does not carry
over to all contractive metrics on quantum states.

\begin{proof}
We can think of $A$ as describing a random variable $a$ that takes values in the discrete sample space $\{1,\ldots,d_A\}$, and similarly for $B$.
Then we use the short-hand notation
\[
   x^{AB}_{i,j}:={\rm Prob}(a=i,b=j),
\]
and we get $x^A_i=\sum_j x^{AB}_{i,j}$. This allows us to define the conditional probability distribution
\[
   x^{AB}_{j|i}:=\left\{
      \begin{array}{cl}
         x^{AB}_{i,j}/x^A_i & \mbox{if }x^A_i\neq 0,\\
         \delta_{i,j} & \mbox{otherwise},
      \end{array}
   \right.
\]
where $\delta_{i,j}=1$ if $i=j$ and $0$ otherwise. We obtain
\begin{equation}
   x^{AB}_{i,j}=x^{AB}_{j|i}\cdot x^A_i.
   \label{eqDefChannel}
\end{equation}
Moreover, since $x_{j|i}^{AB}\geq 0$ and
\[
   \sum_j x^{AB}_{j|i}=1 \mbox{ for all }i,
\]
the conditional probability $x^{AB}_{j|i}$ defines a stochastic matrix, i.e.\ a classical channel $N$ from the probability distributions on $A$ to those on $AB$,
such that~(\ref{eqDefChannel}) becomes $x^{AB}=N(x^A)$. Define $y^{AB}:=N(y^A)$, then
\[
   \sum_j y^{AB}_{i,j} = \sum_j x^{AB}_{j|i}y^A_i = y^A_i,
\]
such that $y^{AB}$ is indeed an extension of $y^A$. Due to contractivity of $\mathcal{D}$, we obtain
\[
   \mathcal{D}(y^{AB},x^{AB})=\mathcal{D}(N(y^A),N(x^A))\leq\mathcal{D}(y^A,x^A);
\]
on the other hand, marginalization is also a channel, which proves the converse inequality.
\end{proof}

This extension property implies an inequality for the max-entropy:
\begin{lemma}
\label{LemIneqMaxEnt}
If $x^{AB}$ is any bipartite state with marginal $x^A$, then
\[
   H_0^\varepsilon(x^{AB})\leq H_0^\varepsilon(x^A)+\log d_B
\]
for every $\varepsilon\geq 0$.
\end{lemma}
\begin{proof}
First we prove the statement for $\varepsilon=0$. Suppose that $i$ is such that $x^A_i=0$, then
\[
   x^A_i=0=\sum_j x^{AB}_{i,j}\quad\Rightarrow\quad x^{AB}_{i,j}=0\mbox{ for all }j.
\]
Counting the zero-probability events, this implies that
\[
   \#\{(i,j):x^{AB}_{i,j}=0\} \geq d_B\cdot \#\{i:x^A_i=0\}.
\]
Thus,
\begin{eqnarray*}
   2^{H_0(x^{AB})}&=& \#\{(i,j):x^{AB}_{i,j}\neq 0\}  \\
   &=& d_{AB} - \#\{(i,j): x^{AB}_{i,j}=0\} \\
   &\leq& d_A d_B - d_B\cdot \#\{i:x^A_i=0\} \\
   &=& d_B\cdot \#\{i:x^A_i \neq 0\}\\
   &=& d_B\cdot 2^{H_0(x^A)}.
\end{eqnarray*}
Now we turn to the case $\varepsilon>0$. Let $\tilde x^A$ be any state that achieves the optimization in the definition of $H_0^\varepsilon(x^A)$,
i.e.\ $\mathcal{D}(\tilde x^A,x^A)\leq\varepsilon$ and $H_0^\varepsilon(x^A)=H_0(\tilde x^A)$. According to Lemma~\ref{LemExtension},
there is an extension $\tilde x^{AB}$ of $\tilde x^A$ such that $\mathcal{D}(\tilde x^{AB},x^{AB})=\mathcal{D}(\tilde x^A,x^A)\leq\varepsilon$. Thus
\begin{eqnarray*}
   H_0^\varepsilon(x^{AB})&\leq& H_0(\tilde x^{AB}) \leq H_0(\tilde x^A)+\log d_B\\
   &=&H_0^\varepsilon(x^A)+\log d_B.
\end{eqnarray*}
This proves the claim.
\end{proof}

Now we have collected enough information to prove our main result on the distillable purity.
\begin{lemma}
Any distribution $x$ can be converted by noisy operations into another distribution which is $\varepsilon$-close to a sharp state of nonuniformity
\[
   I_0^\varepsilon(x\otimes m^{(d)}),
\]
where $m^{(d)}$ is the uniform state in $\mathbb{R}^d$. This expression is increasing in $d$.
Furthermore, the \emph{maximal} nonuniformity of any sharp state that can be extracted from $x$ is given by
\begin{equation}
   \lim_{d\to\infty} I_0^\varepsilon(x\otimes m^{(d)}).
   \label{eqOptimal}
\end{equation}
\end{lemma}
\begin{proof}
Clearly, the first statement is a simple consequence of Lemma~\ref{LemDistillSufficient}: before extracting nonuniformity, we can add on a uniform state of any dimension we like. Lemma~\ref{LemIneqMaxEnt} implies that
\begin{eqnarray*}
   I_0^\varepsilon(x\otimes m^{(d)})&=& \log(d_x\cdot d)-H_0^\varepsilon(x\otimes m^{(d)}) \\
   &\geq & \log d_x+\log d -H_0^\varepsilon(x)-\log d \\
   &=& I_0^\varepsilon(x),
\end{eqnarray*}
so that adding on a uniform state does not decrease the value of $I_0^\varepsilon$. Hence $I_0^\varepsilon(x\otimes m^{(d)})$
is increasing in $d$, and the expression in~(\ref{eqOptimal}) exists as a real number or plus infinity.
For the time being, denote this expression by $J_{0}^\varepsilon(x)$. Generalizing the previous calculation, we see that
Lemma~\ref{LemIneqMaxEnt} also implies that
\[
   I_0^\varepsilon(x^{AB})\geq I_0^\varepsilon(x^A),
\]
and so also $J_{0}^\varepsilon(x^{AB})\geq J_{0}^\varepsilon(x^A)$. In other words, $J_{0}^\varepsilon$ is non-increasing under
marginalization. Furthermore, by construction, $J_{0}^\varepsilon(x\otimes m^{(d)})=J_{0}^\varepsilon(x)$; that is, adding on maximally
mixed states leaves $J_{0}^\varepsilon$ invariant.
Finally, note that
\[
   J_{0}^\varepsilon(\Pi(x))=J_{0}^\varepsilon(x)
\]
for any representation $\Pi$ of a permutation; this property is inherited from $I_0^\varepsilon$ and follows from the permutation-invariance of any
contractive metric. These three properties together show that $J_{0}^\varepsilon$ is a nonuniformity monotone.

Now suppose that $x\stackrel{\rm noisy}{\longrightarrow} y$ such that $\mathcal{D}(y,s_I)\leq\varepsilon$ for the sharp state $s_I$ with nonuniformity $I$.
Then
\[
   H_0^\varepsilon(y)\leq H_0(s_I)=\log d_u,
\]
if $I=\log d/d_u$ for $d=d_y$. It follows that $I_0^\varepsilon(y)\geq I$, and since $J_{0}^\varepsilon$ is a nonuniformity monotone,
\[
 J_{0}^\varepsilon(x) \geq J_{0}^\varepsilon(y) \geq I_0^\varepsilon(y) \geq I.
\]
This shows that the expression in~(\ref{eqOptimal}) is optimal.
\end{proof}

\bibliography{Ref_v11}

\end{document}